\documentclass[11pt]{article}
\usepackage{require_package}
\makeatletter
\def\BState{\State\hskip-\ALG@thistlm}
\makeatother

\title{Single- and Multi-Level Fourier-RQMC Methods for Multivariate
Shortfall Risk}

\author[1]{Chiheb Ben Hammouda}
\author[1]{Truong Ngoc Nguyen\thanks{Corresponding Author: Email: n.t.nguyen@uu.nl}}

\affil[1]{Mathematical Institute, Utrecht University, 3584 CD Utrecht, the Netherlands}

\begin{document}
	\date{}
\maketitle

\begin{abstract}
Multivariate shortfall risk measures provide a principled framework for quantifying systemic risk and determining capital allocations prior to aggregation in interconnected financial systems. Despite their well-established theoretical properties, the numerical estimation of multivariate shortfall risk and the corresponding optimal allocations remains computationally challenging, as existing Monte Carlo–based approaches can be numerically expensive due to slow convergence.

In this work, we develop a new class of single- and multilevel numerical algorithms for estimating multivariate shortfall risk and the associated optimal allocations, based on a combination of Fourier inversion techniques and randomized quasi–Monte Carlo (RQMC) sampling.  Rather than operating in physical space, our approach evaluates the relevant expectations appearing in the risk constraint and its optimization in the frequency domain, where the integrands exhibit enhanced smoothness properties that are well suited for RQMC integration. We establish a rigorous mathematical framework for the resulting Fourier-RQMC estimators, including convergence analysis and computational complexity bounds. Beyond the single-level method, we introduce a multilevel RQMC scheme that exploits the geometric convergence of the underlying deterministic optimization algorithm to reduce computational cost while preserving accuracy.

Numerical experiments demonstrate that the proposed Fourier–RQMC methods outperform sample average approximation and stochastic optimization benchmarks in terms of accuracy and computational cost across a range of models for the risk factors and loss structures. Consistent with the theoretical analysis, these results demonstrate improved asymptotic convergence and complexity rates relative to the benchmark methods, with additional savings achieved through the proposed multilevel RQMC construction.

\vspace{0.3em}
\end{abstract}
\newpage
\tableofcontents

\section{Introduction}

Systemic risk is the risk that the entire financial system could fail due to the inherent characteristics of the system itself. Such a risk can trigger severe economic losses, requiring effective tools for its assessment and control. To this end, risk assessment methods must be accurate, coherent, and capable of assessing risks and informing capital allocation across interconnected financial components, such as portfolios, financial institutions, and clearinghouse members. 

Two modelling approaches have been proposed in the literature. In the \emph{post-aggregation} view (i.e., \emph{aggregate first, then allocate})\cite{chen_axiomatic_2013,brunnermeier_measuring_2019}, systemic risk measures are interpreted as the minimal amount of capital required to secure the financial system after aggregating losses across individual institutions, i.e., we first compress the multivariate loss vector $\mathbf{X}=(X_1,\ldots,X_d)$ through an aggregator $\Lambda$ and then apply a univariate risk functional $\eta$,
		\[
		R(\mathbf{X})=\eta\!\big(\Lambda(\mathbf{X})\big).
		\]
 While conceptually appealing, this aggregation step collapses the multivariate loss profile into a single scalar quantity, potentially obscuring information about interdependencies and interactions among system components not encoded by $\Lambda$.  In contrast, the \emph{pre-aggregation} view (i.e., \emph{allocate first, then aggregate}) \cite{feinstein_measures_2017,biagini_unified_2019,kromer_systemic_2016} allocates capital componentwise before aggregation,  leading to  risk measures of the  form
	\[
	R(\mathbf{X}) = \inf\Bigl\{\sum_{i=1}^d m_i : \Lambda(\mathbf{X}+\mathbf{m})\in\mathcal{A}\Bigr\},\]
	where $\mathbf{m}=(m_1,\dots,m_d)$ denotes the capital allocation vector and $\mathcal{A}$ is an acceptance set. This formulation preserves the multivariate dependence structure among the components of the system and yields both a total risk measure and consistent component-wise capital allocations.

 Beyond their theoretical formulation, systemic risk measures must be monitored on a regular basis—monthly or even weekly—making computational efficiency and numerical scalability essential. In this work, we adopt the pre-aggregation perspective and focus on the efficient computation of the Multivariate Shortfall Risk Measure (MSRM) introduced by \cite{armenti_multivariate_2018}.

Computing the MSRM requires repeatedly evaluating expectations involving the loss function and its gradient, which enter the  first-order optimality conditions of the associated allocation problem along an optimization trajectory. This repeated evaluation of expectation-based quantities constitutes the main computational bottleneck. Two main  numerical approaches were introduced in the literature. {The first is \emph{Sample-Average Approximation} (SAA) \cite{armenti_multivariate_2018}, where $N$ i.i.d.\ samples of $\mathbf X$ are drawn \emph{once} at the beginning, required expectations are approximated by their Monte Carlo (MC) estimators, and subsequently solved using standard deterministic optimization methods. In unconstrained stochastic optimization, SAA attains the canonical,
dimension-independent convergence rate $\mathcal{O}\paren{N^{-1/2}}$ for optimal solutions under suitable
regularity conditions \cite{homem-de-mello_monte_2014,fu_guide_2015}, a rate that is often regarded as slow in practice \cite{glasserman_monte_2003}. In the MSRM setting, however, the literature typically does not provide a convergence analysis for the corresponding SAA estimators. In contrast, \cite{kaakai_estimation_2024} proposes a second approach based on stochastic approximation (SA), in which the optimizer is updated iteratively using noisy gradient information obtained from samples of the loss vector $\mathbf{X}$, following a Robbins–Monro–type scheme \cite{bardou_computing_2009,dunkel_stochastic_2010}.
While these SA algorithms are carefully adapted to the MSRM framework and are shown to converge theoretically, the numerical experiments in \cite{kaakai_estimation_2024} suggest that their numerical performance can be inferior to that of SAA-based methods.

In this work, we address this computational challenge by developing efficient numerical methods for MSRM that operate in the frequency domain and are tailored to the optimization-driven structure of the problem. Our approach combines Fourier-based representations of the MSRM optimality conditions with randomized quasi–Monte Carlo (RQMC) sampling to efficiently evaluate the expectation-based quantities required along the allocation optimization trajectory. Compared to previously proposed Monte Carlo–based approaches \cite{armenti_multivariate_2018, kaakai_estimation_2024}, our method exploits the smoothness of frequency-domain integrands and the improved error convergence properties of RQMC. Building on this single-level framework, we further introduce a multilevel strategy that leverages the fast local (geometric) convergence of the underlying numerical optimization scheme to reduce overall computational cost while preserving accuracy.

In risk measurement, QMC methods have been used primarily within the post-aggregation paradigm, where computations are performed in physical space on the aggregated portfolio loss, most notably for quantile-based risk measures such as VaR and CVaR \cite{kreinin_measuring_1998,bardou_recursive_2009,he_convergence_2021}. Under suitable smoothness assumptions, QMC estimators in this setting can achieve convergence rates of order $\mathcal{O}(N^{-1})$ \cite{he_convergence_2021}. A key limitation of these approaches is their strong sensitivity to both the effective dimension and the regularity of the integrand. In risk measurement, loss functions and their gradients often involve kinks or discontinuities, which can severely degrade QMC efficiency when treated directly in physical space. These issues are further exacerbated in pre-aggregation frameworks such as MSRM, where expectations depend on an evolving allocation vector along a high-dimensional optimization trajectory. As a result, regularity must be preserved uniformly over the allocation iterates generated by the optimization algorithm, rather than only at a single scalar value. Consequently, despite their success in post-aggregation settings, applications of QMC and RQMC methods to pre-aggregation risk measures remain scarce. 

One possible route to improving the regularity structure\footnote{Besides alternative approaches based on analytical or numerical smoothing, as explored for example in the context of option pricing in \cite{bayer_hierarchical_2020, bayer_numerical_2023}.} is to work in the frequency domain, leveraging Fourier representations when characteristic functions of the loss vector are available. In the MSRM setting, Fourier-based techniques have been explored primarily in \cite{armenti_multivariate_2018}, where they are used mainly as numerical benchmarks against SAA-type methods and are reported to exhibit limited practical efficiency relative to SAA. However, this approach remains largely heuristic and does not systematically exploit the potential advantages of QMC integration. In particular, the proposed implementation in \cite{armenti_multivariate_2018} relies on generic Gaussian quadrature rules that scale poorly with dimension, offers limited guidance on the choice of contour-shift (damping) parameters, which are selected ad hoc despite their critical role in controlling  the smoothness of the resulting Fourier integrands \cite{bayer_optimal_2023}, and does not provide a unified analytical framework ensuring stability and accuracy along the optimization trajectory. Moreover, neither a rigorous error analysis nor a  computational complexity analysis is provided.

These limitations highlight the need for a systematic formulation that tightly couples Fourier representations, regularity control, and optimization-aware numerical integration, while remaining robust along the optimization trajectory. We address this need by developing an optimization-aware Fourier–RQMC framework for MSRM, together with a rigorous error and  complexity analysis that demonstrates its convergence properties and computational advantages over  SA and SAA-based approaches. Although developed in the context of MSRM, the proposed theoretical and numerical framework naturally extends to other classes of multivariate risk measures for which suitable Fourier representations and optimization regularity conditions are available.

\paragraph{Our contributions.} The main contributions of this work are summarized as follows:
\begin{itemize}

\item  We develop a new class of single-level and multilevel numerical algorithms for estimating multivariate shortfall risk measures and the associated optimal capital allocations, based on a combination of Fourier inversion techniques and RQMC sampling. The proposed design incorporates carefully constructed damping rules and domain transformations to preserve the regularity structure of the Fourier integrands associated with the loss functions and their gradients, uniformly along the allocation iterates (Sections \ref{sec:mapping_MSRM_Fourier} and \ref{sec:Methodology of our Approach}).
\item We introduce a multilevel RQMC construction whose sample allocation explicitly exploits the local geometric convergence of the underlying deterministic optimization algorithm, leading to work-optimal sampling strategies and further reductions in computational complexity.

\item We provide a rigorous error and  complexity analysis of the proposed single-level and multilevel Fourier–RQMC schemes for the MSRM problem, establishing improved asymptotic convergence and complexity rates relative to the  benchmark methods (Section \ref{sec:num_analysis}; Appendix \ref{appendix:sup_err_analysis}).
    
\item  Building on ideas introduced in \cite{bayer_optimal_2023}  in the context of option pricing, we design an adaptive damping strategy across optimization iterations, complemented by a regularized update rule that ensures robustness under repeated evaluation of evolving integrands along the optimization trajectory. In addition, we provide theoretical guarantees for the convexity of the resulting damping-selection problem (Section \ref{subsec:opti_damp} and Appendix \ref{Appendix:supp_optimal_damping}).

      \item   We adapt the domain transformation for RQMC integration proposed in  \cite{bayer_quasi-monte_2025} to the repeated integration problems arising along an optimization trajectory. Moreover, we provide an alternative derivation of the transformation rule based on a detailed analysis and control of boundary oscillations tailored to our setting (Section~\ref{subsubsec:domain_trans_QMC}; Appendix~\ref{Appendix:supp_results_domain_trans}).

		\item  We demonstrate the computational advantages of the proposed Fourier–RQMC methods through different numerical experiments, using SAA and SA as benchmarks across different loss functions, risk-factor distributions, and dimensional settings. The experiments illustrate the additional savings achieved through the multilevel RQMC construction   (Section \ref{sec:num_exp_results}).
	\end{itemize}
\paragraph{Organization of the paper.}  The outline of this paper is as follows. Section~\ref{sec:Problem Setting and Background}   introduces the problem setting and the associated optimization framework.  Section~\ref{sec:mapping_MSRM_Fourier} presents the Fourier representations of the risk-measure problems together with the proposed optimal damping strategy. Section~\ref{sec:Methodology of our Approach}  develops the optimization-aware Fourier–RQMC framework and its multilevel extension. Section~\ref{sec:num_analysis}  provides a rigorous error and work–accuracy complexity analysis of the proposed methods. Finally, Section~\ref{sec:num_exp_results} reports the numerical experiments and results.

\section{{Optimization Framework for Multivariate Shortfall Risk}}\label{sec:Problem Setting and Background}
This section introduces the analytical setting and optimization framework for MSRM, together with the constrained optimization formulation used throughout the paper. We establish notation, define admissible multivariate loss functions and dependence structures, and recall key existence, uniqueness, and first-order optimality results for risk allocations. These results underpin the Sequential Quadratic Programming (SQP)-based numerical framework used to solve the MSRM problem.
\begin{notation}\
\begin{itemize}[leftmargin=*,label=\normalfont\textbullet]
   \item $\norm{.}$ denotes the Euclidean norm for vectors and the associated Frobenius norm for matrices, unless stated otherwise.
  \item
  Let $(\Omega,\mathcal{F},\mathbb{P})$ be a probability space and let $d\in\mathbb{N}$. We denote by $L^0\paren{\Omega;\mathbb{R}^d}$ the space of 
  $\mathbb{R}^d$-valued $\mathcal{F}$-measurable random variables on
  $(\Omega,\mathcal{F},\mathbb{P})$.
  For $p\in[1,\infty)$ we define  $ L^p\paren{\Omega;\mathbb{R}^d}
    := \bigl\{ \mathbf{X} \in L^0\paren{\Omega;\mathbb{R}^d} :
                \mathbb{E}[\|\mathbf{X}\|^p] < \infty \bigr\}$.

  \item For $p\in[1,\infty)$, we {denote} the space of  $p$-integrable functions on $\mathbb{R}^d$ as

$    
L^p(\mathbb R^d)
:=\Bigl\{ f:\mathbb R^d\to\mathbb R \ \big|\ \int_{\mathbb R^d}|f(x)|^p\,dx<\infty\Bigr\}.
$  
  \item
For $\mathbf{x},\mathbf{y} \in \mathbb{R}^d$ we define
$\mathbf{x} \ge \mathbf{y}
\;\Longleftrightarrow\;
x_k \ge y_k$ for all $k \in \{1,\dots,d\}$,
and $\mathbf{x} > \mathbf{y}$ if $x_k > y_k$ for all $k$.
Here $x_k$ and $y_k$ denote the $k$-th components of $\mathbf{x}$ and
$\mathbf{y}$, respectively.

  \item A generic element $\mathbf X=(X_1,\ldots,X_d)\in L^0(\Omega;\mathbb R^d)$ denotes an $\mathbb R^d$-valued random vector of (monetary) financial losses {(i.e., for each $k=1,\ldots,d$, the component $X_k$ represents the loss  of institution $k$}. We assume that $\mathbf X$ admits a joint density $f_{\mathbf X}$ on $\mathbb R^d$ with the parameters denoted by $\boldsymbol \Theta_\mathbf{X}$.
\end{itemize}
\label{notation:prob_space}
\end{notation}
We begin by defining a multivariate loss function $\ell : \mathbb{R}^d \to (-\infty, \infty]$, which serves as the basic building block of the multivariate shortfall risk framework.

\begin{definition}[Multivariate Loss Function]
A function $\ell : \mathbb{R}^d \to (-\infty, \infty]$ is called a \emph{loss function} if:
\begin{enumerate}[label=(A\arabic*), ref=A\arabic*]
    \item\label{ass:A1} $\ell$ is increasing: $\ell(\mathbf{x}) \geq \ell(\mathbf{y})$ whenever $\mathbf{x} \geq \mathbf{y}$,{with $\mathbf{x}, \mathbf{y} \in \mathbb{R}^d$};
    \item\label{ass:A2} $\ell$ is convex and lower semicontinuous, with $\inf \ell < 0$;
    \item\label{ass:A3} {For $\mathbf{x} \in \mathbb{R}^d$, $\ell(\mathbf{x}) \geq \sum_{k=1}^d x_k - c$ for some constant {$c \in \mathbb{R}$}.}
\end{enumerate}
\label{def:assumption_loss_func}
\end{definition}
Throughout this work, we shall refer to $\ell$ as a (multivariate) loss function if it satisfies Assumptions \eqref{ass:A1}–\eqref{ass:A3}. We will additionally impose permutation invariance (\ref{ass:A4}) when needed.
\begin{enumerate}[label=(A\arabic*), ref=A\arabic*]
\setcounter{enumi}{3} 
    \item \label{ass:A4} $\ell$ is \emph{permutation invariant}, which means  $\ell(\mathbf{x}) = \ell(\pi(\mathbf{x}))$ for every permutation $\pi$ of the components.
\end{enumerate}
Property~(\ref{ass:A1}) states that the risk measure increases with losses. Property~(\ref{ass:A2}) reflects that the diversification should not increase risk; the lower semi-continuity ensures that losses may exhibit one-sided jumps while still guaranteeing that the infimum of the risk measure is attained on the domain. Property~(\ref{ass:A3}) ensures that the risk measure penalizes large losses more heavily than a risk-neutral evaluation. Assumption~(\ref{ass:A4}) implies that when the considered risk components are of the same nature, such as banks, clearinghouse members, or trading desks within the same trading floor, the loss function $\ell$ encodes fairness, meaning that it should not discriminate against any particular component.

We now illustrate how to build a multivariate loss function from a one-dimensional loss function together with a coupling term that encodes dependence and interaction across components of the loss vector $\mathbf{X}$. 
\begin{example}
Let $h:\mathbb R\to\mathbb R$ be a one-dimensional loss function satisfying \eqref{ass:A1}--\eqref{ass:A3}, and let $\Xi:\mathbb R^d\to\mathbb R$ be a coupling functional modelling the dependence structure among the components of the loss vector, also satisfying \eqref{ass:A1}--\eqref{ass:A3}. We construct a multivariate loss function $\ell$ by\footnote{This generalizes class (C3) in \cite[Example~2.3]{armenti_multivariate_2018}. As discussed in \cite[Section~4]{armenti_multivariate_2018}, the parameter $\alpha$ influences risk allocation through the dependence structure of the components of the loss vector $\mathbf X$.}
\begin{equation}
    \ell\paren{\mathbf{x}} := \sum h(x_k) + \alpha \Xi(\mathbf{x}), \quad 0 \leq \alpha \leq 1.
\label{eq:loss_dependence}
\end{equation}
In this work, we consider two examples of multivariate loss functions constructed from \eqref{eq:loss_dependence}
\begin{enumerate}[label=(\roman*)]
\item Exponential (entropic-type) loss function
\begin{equation}\label{eq:multi_entropy}
\ell_{\mathrm{exp}}(\mathbf{x})
:= \frac{1}{1+\alpha}\left[
        \sum_{k=1}^d e^{\beta x_k}
        + \alpha\, e^{\beta \sum_{k=1}^d x_k}
     \right]
     - \frac{\alpha + d}{\alpha + 1},
\qquad \alpha \ge 0,\;\beta \ge 0.
\end{equation}
\item Quadratic pairwise coupling loss function (QPC)
\begin{equation}\label{eq:multi_qcl}
\ell_{\mathrm{qpc}}(\mathbf{x})
:= \sum_{k=1}^d x_k
     + \frac{1}{2}\sum_{k=1}^d (x_k^+)^2
     + \alpha \sum_{1 \le j < k \le d} x_j^+ x_k^+
     - 1,
\qquad \alpha \ge 0.
\end{equation}
\end{enumerate}
where $\alpha$ and $\beta$ denote the systemic weight and the risk-aversion coefficient, respectively.
\label{exam:cross_dependent_losses}
\end{example}

\begin{assumption}[Integrability of the loss vector]
To ensure integrability, we assume that the loss vector $\mathbf{X}$ belongs to the following multivariate Orcliz heart:\footnote{Orlicz spaces have been widely used in the study of risk measures \cite{cheridito_risk_2009}. A detailed discussion of their properties in the multivariate sense can be found in Appendix B of \cite{armenti_multivariate_2018}.}
\begin{equation}
    M^\gamma = \left\{ \mathbf{X} \in L^0\paren{\Omega;\mathbb{R}^d}: \mathbb{E}[\gamma(\lambda \mathbf{X})] < \infty \text{ for all } \lambda > 0 \right\}, \quad \text{where } \gamma(\mathbf{x}) := \ell(|\mathbf{x}|),\ \mathbf{x} \in \mathbb{R}^d.
    \label{eq:Multivariate_Orcliz}
\end{equation}
\label{ass:multivariate_orclizt_space}
\end{assumption}
Next, we recall the definition of the MSRM as introduced in \cite[Definition 2.7]{armenti_multivariate_2018}.
\begin{definition}[Multivariate shortfall risk]
Let $\ell$ be a multivariate loss function and that $\mathbf{X} \in M^\gamma$ {be defined on $(\Omega,\mathcal F,\mathbb P)$}. The multivariate shortfall risk $R(\mathbf{X})$ is defined as
\begin{equation}
	R(\mathbf{X}) := \inf \left\{ \sum_{k=1}^d m_k : \mathbf{m} \in \mathcal{A}(\mathbf{X}) \right\} = \inf \left\{ \sum_{k=1}^d m_k : \mathbb{E}[\ell(\mathbf{X-m})] \leq 0  \right\},
\label{eq:deter_optimization_MSRM}
\end{equation}
where {$\mathbf{m}=(m_1,\dots,m_d)$ denotes the risk (capital) allocation vector}, and the acceptance set is characterised as
\begin{equation}
    \mathcal{A}(\mathbf{X}) := \left\{ \mathbf{m} \in \mathbb{R}^d : \mathbb{E}[\ell(\mathbf{X-m})] \leq 0 \right\}.
\end{equation}
\label{def:loss_acceptance_set}
\end{definition}
We now address the question of existence and uniqueness of risk allocations. {To this end, we introduce the following definition and impose additional assumptions on the loss function $\ell$ and the loss vector $\mathbf{X}$.}
\begin{definition}
A vector $\mathbf{m} \in \mathbb{R}^d$ is called an \emph{acceptable monetary risk allocation} if $\mathbf{m} \in \mathcal{A}(\mathbf{X})$ such that $R(\mathbf{X}) = \sum_{k=1}^d m_k$.
\end{definition}
\begin{assumption}\
  \begin{enumerate}[label=(A\arabic*), ref=A\arabic*]
\setcounter{enumi}{4}
    \item\label{ass:A5} For every $\mathbf{m_0} \in \mathbb{R}^d$, the mapping $\mathbf{m} \mapsto \ell(\mathbf{X-m})$ is differentiable at $\mathbf{m}_0$ almost surely (a.s).
    \item \label{ass:A6} Let $
   \mathcal{U} := \left\{ \mathbf{u} \in \mathbb{R}^d : \sum u_i = 0 \right\} $
be the zero-sum allocations set. We assume that, for every  $\mathbf{x} \in \mathbb R^d$,
the function  $\mathbf{m}\mapsto \ell(\mathbf{x}+\mathbf{m})$ is strictly convex on $\mathcal U$, and that  $\ell(\mathbf{x}) \ge 0$.
\end{enumerate}
\label{assump:additional_assump_loss}
\end{assumption}

\begin{theorem}[Theorem 3.4 in \cite{armenti_multivariate_2018}]
Let $\ell$ be a loss function that satisfies Assumptions \eqref{ass:A4}--\eqref{ass:A5} and $\mathbf{X} \in M^\gamma$. Then a risk allocation $\mathbf{m}^* \in \mathbb{R}^d$ exists and is characterized by the first-order conditions (F.O.C.):
\begin{equation}
\begin{pmatrix}
\lambda^*\,\mathbb{E}\brac{\nabla \ell(\mathbf{X}-\mathbf{m^*})}-\mathbf{1} \\[3pt]
\mathbb{E}[\ell(\mathbf{X}-\mathbf{m^*})]
\end{pmatrix}
= \mathbf 0,
\label{eq:KKT_system}
\end{equation}
where $\lambda^* > 0$ is the Lagrange multiplier, and $\mathbf{1}:= (1, \dots,1)^T \in \rset^d$.\\
Moreover, if Assumption \eqref{ass:A6} holds, then the risk allocation $\mathbf{m}^*$ is unique.
\label{theorem:existence_uniq_allocation}
\end{theorem}
{\begin{remark}
Throughout the paper, the gradient $\nabla$ and the Hessian $\nabla^2$ are taken with respect to (w.r.t.)  the allocation vector $\mathbf m$, unless the differentiation variable is explicitly indicated.\end{remark}}
\begin{remark}
Theorem 3.4 in \cite{armenti_multivariate_2018} is formulated in terms of subdifferentials and holds under the weaker standing Assumption \eqref{ass:A4}, without requiring \eqref{ass:A5}. Under Assumption \eqref{ass:A5}, the subdifferential inclusion reduces to the F.O.C. in Theorem \ref{theorem:existence_uniq_allocation}.
\end{remark} 
\begin{remark}[Interchanging Differentiation and Expectation]
{Based on property \eqref{ass:A2}, the loss function $\ell$ is locally Lipschitz on the interior of its effective domain; see \cite[Theorem B]{roberts_another_1974}. Together with Assumption \eqref{ass:A5}, this allows interchanging  differentiation and expectation, yielding  $\nabla \mathbb{E}\brac{\ell\paren{\mathbf{X}-\mathbf{m}}} = \mathbb{E}[\nabla \ell(\mathbf{X}-\mathbf{m})]$.}
\label{rema:interchange_diff_expectation} 
\end{remark}
When extending to the multivariate setting, the uniqueness of risk allocations becomes crucial. Without uniqueness, the total capital requirement may be distributed arbitrarily across components, potentially leading to outcomes that are not acceptable from a regulatory standpoint. For loss functions of the form \eqref{eq:loss_dependence}, we have the following corollary. 
\begin{corollary}[Uniqueness of the optimal allocation]\label{rem:uniqueness_allocation}
Let $\ell$ be defined as in \eqref{eq:loss_dependence} with a strictly convex univariate loss function $h$.
Assume that the coupling term $\Xi$ satisfies \eqref{ass:A4} and that Assumption~\eqref{ass:A5} holds. Then, for every $\mathbf X\in M^\gamma$, the associated multivariate shortfall risk admits a unique optimal allocation $\mathbf m^*$.
\end{corollary}
\begin{proof}
    The argument follows \cite[Proposition~3.8]{armenti_multivariate_2018}. The term $\sum_{k=1}^d h(x_k)$ is permutation-invariant and thus $\ell$ satisfies \eqref{ass:A4} whenever $\Xi$ does. Moreover, strict convexity of $h$ yields the strict convexity property required for uniqueness (Assumption~\eqref{ass:A6}). Hence, the uniqueness conclusion follows from Theorem~\ref{theorem:existence_uniq_allocation}.
\end{proof}
For convenience, we define $g\paren{\mathbf{m}}:= \mathbb{E}\brac{\ell \paren{\mathbf{X}-\mathbf{m}}}$ and collect the primal–dual variables in $\mathbf{z}:=(\mathbf{m},\lambda)$; this notation will be used throughout the following Sections when constructing Fourier–RQMC surrogate estimators for the gradient and Hessian of the Lagrangian.

The F.O.C. in \eqref{eq:KKT_system} can be rewritten as:
\begin{equation}
   \nabla_{\mathbf{z}} \mathcal{L}(\mathbf{z^*}) := \begin{pmatrix}
\lambda^*\,\nabla g\paren{\mathbf{m^*}}-1 \\[3pt]
g\paren{\mathbf{m^*}}
\end{pmatrix}
=0
\label{eq:KKT_FOC_g_form}
\end{equation}
For the numerical analysis in Section~\ref{sec:num_analysis}, {we also make use of the Hessian of the Lagrangian. Under Assumption~\eqref{ass:C1}, it is given by}
\begin{equation}
\nabla^{2}_{\mathbf{z}} \mathcal{L}(\mathbf{z})
:=
\begin{bmatrix}
-\lambda \nabla^2 g(\mathbf{m}) & \nabla g(\mathbf{m}) \\[3pt]
\nabla g(\mathbf{m})^{\!\top} & 0
\end{bmatrix}.
\label{eq:hessian_KKT_system}
\end{equation}
When solving \eqref{eq:KKT_system} numerically, the expectation-based terms 
\(g(\mathbf{m})\), \(\nabla g(\mathbf{m})\), and \(\nabla^2 g(\mathbf{m})\) in \eqref{eq:KKT_FOC_g_form}--\eqref{eq:hessian_KKT_system} 
need to be estimated. Two common approaches are: (i) constructing \emph{deterministic surrogates}
and applying deterministic optimizers \cite{armenti_multivariate_2018} or (ii) using \emph{stochastic approximation} (SA) methods \cite{kaakai_estimation_2024}. In this work,  we
construct our methodology based on the former one. {In particular,} $g(\mathbf m)$, $\nabla g(\mathbf m)$, and $\nabla^2 g(\mathbf m)$ are approximated by
$\hat g^{\mathrm{Fou}}(\mathbf m)$, $\hat g^{\mathrm{Fou}}_{\nabla}(\mathbf m)$, and $\hat g^{\mathrm{Fou}}_{\nabla^2}(\mathbf m)$, respectively, {using Fourier transform representations combined with single-level and multilevel RQMC methods. This will be explained in Sections~\ref{sec:mapping_MSRM_Fourier} and~\ref{sec:Methodology of our Approach}}. Once these surrogates are in place, the MSRM problem in \eqref{eq:deter_optimization_MSRM} becomes a deterministic nonlinear constrained optimization problem in $\mathbf m$.  We then compute $\mathbf{z}^* = \paren{\mathbf{m}^*, \lambda^*}$ by solving \eqref{eq:KKT_FOC_g_form} numerically, for which a natural choice numerical optimizer is SQP \cite[Chapter~18]{nocedal_numerical_2006}. 

Algorithm~\ref{alg:SLSQP_full} summarizes the generic SQP framework used throughout this work to solve the MSRM problem. The framework is built around surrogate-based Fourier–RQMC estimators and incorporates explicit control of statistical and optimization errors.   In our numerical implementation, we adopt a practical variant of this framework, described in Remark~\ref{rema:SLSQP}. Concrete constructions of the estimators and their associated error bounds are provided in the following sections.
\begin{algorithm}[ht!]
\caption{SQP for MSRM Problem}
\small
\label{alg:SLSQP_full}
\begin{algorithmic}[1]
\Require Initial point $\mathbf{z}_1 = (\mathbf{m}_1, \lambda_1)$, surrogates $\hat{g}^{\mathrm{Fou}},\hat{g}^{\mathrm{Fou}}_{\nabla},\hat{g}^{\mathrm{Fou}}_{\nabla^2}$, the prescribed convergence tolerance $\varepsilon$.
\noindent\textbf{Note:} The surrogates are estimated via Single-level Fourier-RQMC (see Section \ref{subsec:sing-RQMC}, Algorithm \ref{alg:Single-level RQMC}) or Multilevel Fourier-RQMC (see Section \ref{subsec:Multilevel Fou-RQMC}, Algorithm \ref{alg:RQMC_Fou_multi}).
\For{each iteration $j = 1, 2, \ldots J$}

    \State \textbf{Step 1: Linearization of the active constraint}
    \Statex Around $\mathbf{m}^{(j)}$, we approximate
    \[
        \hat{g}^{\mathrm{Fou}}(\mathbf{m}^{(j)} + \mathbf{d}^{(j)}) 
        \approx 
        \hat{g}^{\mathrm{Fou}}\paren{\mathbf{m}^{(j)}}
        + \hat{g}^{\mathrm{Fou}}_{\nabla}\paren{\mathbf{m}^{(j)}}^{\!\top} \mathbf{d}^{(j)} 
        \le 0.
    \]

    \State \textbf{Step 2: Formulation of the Quadratic subproblem (QP)}
    \Statex Solve the quadratic program
    \[
    \begin{aligned}
    \min_{\mathbf{d}^{(j)} \in \mathbb{R}^d} \quad 
    & 
      \mathbf{1}^{\!\top} \mathbf{d}^{(j),\top}
      + \tfrac{1}{2} \mathbf{d}^{(j),\top} 
         \lambda^{(j)}  \hat g^{\mathrm{Fou}}_{\nabla^2}\paren{\mathbf{m}^{(j)}}
        \, \mathbf{d}^{(j)} \\[3pt]
    \text{s.t.} \quad 
    & \hat{g}^{\mathrm{Fou}}\paren{\mathbf{m}^{(j)}}
      +  \hat{g}^{\mathrm{Fou}}_{\nabla}\paren{\mathbf{m}^{(j)}}^{\!\top} \mathbf{d}^{(j)} 
      = 0.
    \end{aligned}
    \]
    \State \textbf{Step 3: Solve the QP problem}
    \Statex The pair $\Delta \mathbf{z}^{(j)} :=(\mathbf{d}^{(j)}, p^{(j)})$ satisfies the KKT system:
    \begin{equation}
    \underbrace{
    \begin{bmatrix}
       -\lambda^{(j)}  \hat g^{\mathrm{Fou}}_{\nabla^2}\paren{\mathbf{m}^{(j)}}& \hat{g}^{\mathrm{Fou}}_{\nabla}\paren{\mathbf{m}^{(j)}}^{\!\top} \\[3pt]
        \hat{g}^{\mathrm{Fou}}_{\nabla}\paren{\mathbf{m}^{(j)}}& 0
    \end{bmatrix}}_{:=\hat{\mathcal{L}}_{\nabla_{\mathbf{z}}^2}^{\mathrm{Fou}}\paren{\mathbf{z}^{(j)}}}
    \begin{bmatrix}
        \mathbf{d}^{(j)} \\[3pt]
        p^{(j)}
    \end{bmatrix}
    =
    \underbrace{
    \begin{bmatrix}
        \mathbf{1} - \lambda^{(j)} \hat{g}^{\mathrm{Fou}}_{\nabla}\paren{\mathbf{m}^{(j)}} \\[3pt]
        -\hat{g}^{\mathrm{Fou}}\paren{\mathbf{m}^{(j)}}
    \end{bmatrix}}_{:=\hat{\mathcal{L}}_{\nabla_{\mathbf{z}}}^{\mathrm{Fou}}\paren{\mathbf{z}^{(j)}}}.
    \label{eq:Equality_KKT_system}
    \end{equation}

    \State \textbf{Step 4: Line search and update}
    \Statex Determine a step size $\alpha_j \in (0,1]$ via a backtracking line search 
    using an appropriate merit function 
    (see \cite[Section~18.3]{nocedal_numerical_2006}), 
    and update:
    \begin{equation}
        \mathbf{z}^{(j+1)}
        = \mathbf{z}^{(j)} + \alpha^{(j)} \Delta\mathbf{z}^{(j)}
    \label{eq:update_zk_slsqp}
    \end{equation}
    \State \textbf{Step 5: Convergence check} 
    \If{$\abs{\Delta \mathbf{z}^{(j)}} \leq \varepsilon_{\mathrm{opt}}$, with $\varepsilon_{\mathrm{opt}} \leq \frac{\varepsilon}{2}$}
        \State \textbf{stop.}
    \EndIf
\EndFor
\end{algorithmic}
\end{algorithm}   

\begin{remark}[SLSQP]
In the numerical experiments, Algorithm \ref{alg:SLSQP_full} is realized using Sequential Least Squares Programming (SLSQP), a standard quasi-Newton variant of the SQP methodology. A key advantage of SLSQP is that it does not require explicit evaluation of exact second-order derivative information at every iteration. Instead, it relies on a quasi-Newton approximation (e.g., BFGS; see  \cite[Chapters~6 and~18]{nocedal_numerical_2006}) to capture the second-order information associated with the SQP subproblem. More precisely, SLSQP maintains a symmetric matrix 
$\hat{\mathbf{B}}^{\mathrm{Fou}}\paren{\mathbf{m}^{(j)}}\in\mathbb{R}^{d\times d}$ which serves as an approximation of the Hessian of the Lagrangian. The initialization is typically chosen as  $\hat{\mathbf{B}}^{\mathrm{Fou}}\paren{\mathbf{m}^{(1)}}=\boldsymbol I_d$ (or a diagonal approximation).
At iteration $j$, the update takes the form
\[
\hat{\mathbf{B}}^{\mathrm{Fou}}\paren{\mathbf{m}^{(j+1)}}
=
\hat{\mathbf{B}}^{\mathrm{Fou}}\paren{\mathbf{m}^{(j)}}
+
\frac{\mathbf{y}^{(j)} \mathbf{y}^{(j),\top}}{\mathbf{y}^{(j),\top} \mathbf{s}^{(j)}}
-
\frac{\hat{\mathbf{B}}^{\mathrm{Fou}}\paren{\mathbf{m}^{(j)}}\,\mathbf{s}^{(j)} \mathbf{s}^{(j),\top}\,\hat{\mathbf{B}}^{\mathrm{Fou}}\paren{\mathbf{m}^{(j)}}}
{\mathbf{s}^{(j),\top}\hat{\mathbf{B}}^{\mathrm{Fou}}\paren{\mathbf{m}^{(j)}}\,\mathbf{s}^{(j)}},
\]
where
\[
\mathbf{s}^{(j)} := \mathbf{m}^{(j+1)}-\mathbf{m}^{(j)},
\qquad
\mathbf{y}^{(j)} := \lambda^{(j+1)}\hat g^{\mathrm{Fou}}_{\nabla}\!\paren{\mathbf{m}^{(j+1)}}
- \lambda^{(j)} \hat g^{\mathrm{Fou}}_{\nabla}\!\paren{\mathbf{m}^{(j)}}.
\]
\label{rema:SLSQP}
\end{remark}

\section{Fourier Representations of the MSRM Problem}\label{sec:mapping_MSRM_Fourier}
In this section we derive Fourier-domain representations of the MSRM objective, its gradient, and its Hessian along the allocation trajectory. Our approach extends the framework of \cite{drapeau_fourier_2014}  to the multivariate setting. Although the Fourier methodology in \cite{bayer_optimal_2023,bayer_quasi-monte_2025} was developed in option pricing, we adapt it here to risk measurement. We first introduce the required notation and integrability assumptions, and then obtain a unified Fourier representation in  Corollary~\ref{coro:Multivariate_Fourier_pricing}.
\begin{notation}\label{notation:X_and_corresponding_CF}\
\begin{itemize}
    
    \item For $\mathbf{y} \in \mathbb{C}^d$, $
    \Phi_{\mathbf{X}}(\mathbf{y}) := \mathbb{E} \left[ e^{\mathrm{i} \langle \mathbf{y}, \mathbf{X} \rangle} \right],
    $ denotes the {joint extended characteristic function {(CF)}} of $\mathbf{X}$. Here $\langle.,.\rangle$ denotes the inner product on $\mathbb{R}^d$ extended bi-linearly to $\mathbb{C}^d$, {i.e.,} for $\mathbf{w,t} \in \mathbb{C}^d, \langle \mathbf{w}, \mathbf{t} \rangle = \sum_{k=1}^d w_k t_k $
   
   \item {For $\mathbf{y} \in \mathbb{C}^d$, the function $\hat{f}(\mathbf{y})
:=  \int_{\mathbb{R}^d}
e^{-\mathrm{i}\,\langle \mathbf{y}, \mathbf{x} \rangle}\,
f(\mathbf{x})\,\mathrm{d}\mathbf{x}$ represents the (extended) Fourier transforms of the function $f$.}
\item The Fourier transform of $\nabla_\mathbf{x}f$ is taken componentwise, i.e.,
$\hat{f}_{\nabla_\mathbf{x}}(\mathbf{y})
=
\big(\hat f_{\partial_{x_1}}(\mathbf{y}),\dots,\hat f_{\partial_{x_d}}(\mathbf{y})\big)$.
Similarly, the Fourier transform of $\nabla_\mathbf{x}^2f$ is taken entrywise:
$\hat{f}_{\nabla_\mathbf{x}^2}(\mathbf{y})
=
\paren{\hat f_{\partial^2_{x_ix_j}}(\mathbf{y})}_{i,j=1}^d$.
\item $\mathrm{i}$ denotes the imaginary unit number, $\Re[\cdot]$ and $\Im[\cdot]$ are the real and imaginary parts of a complex number, respectively.
\item $\boldsymbol{\Theta}_\ell$ denotes the parameters of the loss function
$\ell$.
\item For $\nu=0,1,2$, define $\ell^{(\nu)}:=\ell,\,\nabla\ell,\,\nabla^2\ell$, respectively, and let $\widehat{\ell}^{(\nu)}:=\widehat{\ell},\,\widehat{\ell}_{\nabla},\,\widehat{\ell}_{\nabla^2}$ denote their corresponding Fourier transforms.
Let $\delta_{\ell}^{(\nu)}$ denote the strip of analyticity of $\widehat{\ell}^{(\nu)}$, i.e.,
\[
\delta_{\ell}^{(\nu)}:=\{\mathbf K^{(\nu)}\in\mathbb R^d \mid \mathbf x \mapsto e^{\langle \mathbf K^{(\nu)},\mathbf x\rangle}\ell^{(\nu)}(\mathbf{x})\in L^1(\mathbb{R}^d)\}.
\]
\item $\delta_X := \{ \mathbf{K} \in \mathbb{R}^d : \mathbb{E}[e^{\langle \mathbf{K},\mathbf{X}\rangle}] < \infty \}$.

\end{itemize}
\end{notation}
\begin{assumption}[Admissible contour shifts]
For each $\nu\in\{0,1,2\}$, there exists $\mathbf K^{(\nu)}\in\delta_X$ such that $
\mathbf w \mapsto \Phi_{\mathbf{X}}(\mathbf{w}+\mathrm{i} \mathbf K^{(\nu)})\,\widehat{\ell}^{(\nu)}(\mathbf{w}+\mathrm{i} \mathbf K^{(\nu)})
\in L^1(\mathbb{R}^d)$. We denote the corresponding admissible set by
\[
\delta_K^{(\nu)}
:=\left\{\mathbf K\in\delta_X:\ 
\Phi_{\mathbf{X}}(\cdot+\mathrm{i} \mathbf K)\,\widehat{\ell}^{(\nu)}(\cdot+\mathrm{i} \mathbf K)\in L^1(\mathbb{R}^d)
\right\},
\]
and assume $\delta_K^{(\nu)}\neq\varnothing$.
 \label{ass:A7}
\end{assumption}
\begin{remark}
By construction, $\delta_K^{(\nu)} \subseteq \delta_X \cap \delta_\ell^{(\nu)}$.
\end{remark}

The corollary and its proof below are adapted for the MSRM problem, which is based on \cite[Proposition 2.4]{bayer_optimal_2023}.
\begin{corollary}[Fourier representations for MSRM problem]\
\label{coro:Multivariate_Fourier_pricing}
Suppose Assumption \ref{ass:A7} holds. Then, for  any choice of 
\(\mathbf{K}^{(\nu)} \in \delta_{K}^{(\nu)}\), the Fourier-based representation for {\(g(\mathbf{m})\), 
\(\nabla g(\mathbf{m})\) and \(\nabla^2 g(\mathbf{m})\) in \eqref{eq:KKT_FOC_g_form}-\eqref{eq:hessian_KKT_system}} are given in unified form by:
\begin{align}
\hat g^{(\nu),\mathrm{Fou}}(\mathbf{m})
&:= (2\pi)^{-d} \,
\Re\!\left[
\int_{\mathbb{R}^d} 
e^{\langle \mathbf{K}^{(\nu)} - \mathrm{i} \mathbf{w},\, \mathbf{m} \rangle}
\, \Phi_{\mathbf{X}}\paren{\mathbf{w} + \mathrm{i}\mathbf{K}^{(\nu)}}
\, \widehat{\ell}^{(\nu)}\paren{\mathbf{w} + \mathrm{i}\mathbf{K}^{(\nu)}}
\, \mathrm{d}\mathbf{w}
\right], \quad \nu=0,1,2.\label{eq:g_fou}
\end{align}
\end{corollary}
\begin{proof}
    The proof for Corollary \ref{coro:Multivariate_Fourier_pricing} is presented in Appendix \ref{Appendix:proof_corro_fourier_pricing}.
\end{proof}
{In view of the Fourier representations \eqref{eq:g_fou}, we introduce the aggregate integrands}
\begin{equation}
\begin{aligned}
h^{(\nu)}\paren{\mathbf{w};\mathbf{m},  \mathbf{K}^{(\nu)},\boldsymbol{\Theta}}
:= (2\pi)^{-d}
    e^{\langle \mathbf{K}^{(\nu)}-\mathrm{i}\mathbf{w},\,\mathbf{m} \rangle}
    \,\Phi_{\mathbf{X}}\paren{\mathbf{w}+\mathrm{i}\mathbf{K}^{(\nu)}}\,  \widehat{\ell}^{(\nu)}\paren{\mathbf{w}+\mathrm{i}\mathbf{K}^{(\nu)}},\quad \nu = 0,1,2.
\end{aligned}
\label{eq:aggregate_integrands}
\end{equation}
for $\mathbf{w}\in\mathbb{R}^d$, {where $\mathbf{K}^{(\nu)}\in\delta_K^{(\nu)}$, and 
$\boldsymbol{\Theta}:=(\boldsymbol{\Theta}_{\mathbf{X}},\boldsymbol{\Theta}_{\ell})$}.

Motivated by the structure of the multivariate loss functions in Example~\ref{exam:cross_dependent_losses}, which combine marginal terms with dependence components of finite interaction order $q_{\ell}$, we exploit this structure to decompose the aggregate Fourier integrands into componentwise contributions indexed by interaction order and coordinate subsets. This decomposition is not introduced for dimension adaptivity but rather as a fundamental tool for the construction of our numerical methods and their subsequent analysis, allowing expectations, gradients, and Hessians to be expressed as finite sums of lower-dimensional Fourier integrals. The decomposition is formalized in the following notation.
\begin{notation}[Component selection and componentwise integrands]\label{not:component_integrands} Let $q_{\ell}\in\mathbb N, q_{\ell} \leq d$ denote the maximal interaction order appearing in the dependence structure of the loss function, $\ell$; for instance, $q_{\ell}=2$ corresponds to pairwise interactions, $q_{\ell}=3$ to triplet interactions, and so on. Let $\mathcal{I}_{q_{\ell}} \subset \{1,\dots,q_{\ell}\}$. For each $ k \in \mathcal{I}_{q_{\ell}}$, let $\mathcal I_k := \{ \mathbf{p} = (p_1,\dots,p_k) : 1 \le p_1 < \cdots < p_k \le d \}$ denote the collection of all $k$-dimensional coordinate subsets. For each $\mathbf{p} \in \mathcal I_k$, let $P_{k,p} \in \mathbb R^{k \times d}$ denote the corresponding coordinate selection matrix, whose $r$-th row equals the canonical basis vector $e_{p_r}^\top$, \footnote{In particular, its entries are given by 
$(P_{k,p})_{rj} = \begin{cases}
1, & \text{if } j = p_r,\\
0, & \text{otherwise},
\end{cases}
\quad r=1,\dots,k,\; j=1,\dots,d.$} and define the projected vectors $$\mathbf{m}_{k,p} := P_{k,p} \; \mathbf{m} \in \mathbb R^k,
 \qquad \mathbf{X}_{k,p} := P_{k,p} \; \mathbf{X} \in \mathbb R^k.$$
Then, for $\nu = 0,1,2$,  the Fourier-based componentwise integrands are defined by
\begin{equation}
h_{k,p}^{(\nu)}\paren{\mathbf{u};\mathbf{m}_{k,p},\mathbf{K}_{k,p}^{(\nu)},,\boldsymbol{\Theta}_{k,p}}
:= (2\pi)^{-k}
    \exp\!\bigl(\langle \mathbf{K}_{k,p}^{(\nu)}-\mathrm{i}\mathbf{u},\,\mathbf{m}_{k,p}\rangle\bigr)\,
\Phi_{\mathbf{X}_{k,p}}\paren{\mathbf{u}+\mathrm{i}\mathbf{K}_{k,p}^{(\nu)}}\,
\hat{\ell}^{(\nu)}_{k,p}\paren{\mathbf{u}+\mathrm{i}\mathbf{K}_{k,p}^{(\nu)}}, \quad \mathbf{u} \in \mathbb{R}^k 
\label{eq:loss_component_integrands}
\end{equation}
where $\boldsymbol{\Theta}_{k,p}:= (\boldsymbol{\Theta}_{\mathbf{X}_{k,p}},\boldsymbol{\Theta}_{\ell_{k,p}})$ collects the corresponding parameters for component $(k,p)$, and $\mathbf{K}_{k,p} \in \delta_{K_{k,p}}^{(\nu)}:= \delta_{X_{k,p}}^{(\nu)} \cap \delta_{l_{k,p}}^{(\nu)}$.

The aggregated loss function,  integrands in~\eqref{eq:aggregate_integrands} and integrals in~\eqref{eq:g_fou} admit the following finite decomposition: 
\begin{equation}
\begin{aligned}
\ell^{(0)}
&= \sum_{k\in\mathcal I_{q_{\ell}}}\sum_{\mathbf p\in\mathcal I_k} \ell^{(0)}_{k,p},
&
h^{(0)}
&= \sum_{k\in\mathcal I_{q_{\ell}}}\sum_{\mathbf p\in\mathcal I_k} h^{(0)}_{k,p},
&
g^{(0),\mathrm{Fou}}
&= \sum_{k\in\mathcal I_{q_{\ell}}}\sum_{\mathbf p\in\mathcal I_k} g^{(0),\mathrm{Fou}}_{k,p},
\\
\ell^{(1)}
&= \sum_{k\in\mathcal I_{q_{\ell}}}\sum_{\mathbf p\in\mathcal I_k} P_{k,p}^{\top}\, \ell^{(1)}_{k,p},
&
h^{(1)}
&= \sum_{k\in\mathcal I_{q_{\ell}}}\sum_{\mathbf p\in\mathcal I_k} P_{k,p}^{\top}\, h^{(1)}_{k,p},
&
g^{(1),\mathrm{Fou}}
&= \sum_{k\in\mathcal I_{q_{\ell}}}\sum_{\mathbf p\in\mathcal I_k} P_{k,p}^{\top}\, g^{(1),\mathrm{Fou}}_{k,p},
\\ 
\ell^{(2)}
&= \sum_{k\in\mathcal I_{q_{\ell}}}\sum_{\mathbf p\in\mathcal I_k} P_{k,p}^{\top}\, \ell^{(2)}_{k,p}\,P_{k,p},
&
h^{(2)}
&= \sum_{k\in\mathcal I_{q_{\ell}}}\sum_{\mathbf p\in\mathcal I_k} P_{k,p}^{\top}\, h^{(2)}_{k,p}\,P_{k,p},
&
g^{(2),\mathrm{Fou}}
&= \sum_{k\in\mathcal I_{q_{\ell}}}\sum_{\mathbf p\in\mathcal I_k} P_{k,p}^{\top}\, g^{(2),\mathrm{Fou}}_{k,p}\,P_{k,p}. 
\end{aligned}
\label{eq:aggregate_components}
\end{equation}
\end{notation}
 {The decomposition in Notation \ref{not:component_integrands} makes explicit that, under finite interaction order $q_{\ell}$, the Fourier representations of the MSRM objective, its gradient, and its Hessian can be written as finite sums of lower-dimensional Fourier integrals, each involving at most $k \le q_{\ell}$ coordinates. This representation will be exploited in the subsequent sections to construct numerical schemes whose complexity is governed by the interaction order rather than the full dimension $d$ of the loss vector.}
\begin{remark}
    For the choice of loss functions in Example~\ref{exam:cross_dependent_losses}, we restrict attention to the admissible dimension set $\mathcal{I}_{q_{\ell}} = \{1,q_{\ell}\}$, where $q_{\ell}=2$ for \eqref{eq:multi_entropy} and  $q_{\ell}=d$ for \eqref{eq:multi_qcl}.
\label{rema:admissible_dimension}
\end{remark}
For the loss functions in Example \ref{exam:cross_dependent_losses}, the componentwise Fourier transforms (see Appendix~\ref{Appendix: Fourier_transform_loss}) and their strips of analyticity, $\delta_{\ell_{k,p}}^{(\nu)}$, can be characterized explicitly (see Table~\ref{tab:strip_analyticity_loss}). Moreover, for the loss families in Example \ref{exam:cross_dependent_losses} the admissible damping domains for $\ell, \nabla \ell$, and $\nabla^2\ell$ coincide; hence we use a single domain for all $\nu\in\{0,1,2\}$. For loss components whose Fourier representations rely on a domain decomposition, admissible damping parameters are characterized componentwise on the corresponding one-sided domains. In particular, for the exponential loss this leads naturally to admissible pairs $(K^-_{k,p}, K^+_{k,p})$, as summarized in Table \ref{tab:strip_analyticity_loss}.
\begin{table}[H]
\centering
\renewcommand{\arraystretch}{1.8} 
\begin{tabular}{|c|c|}
\hline
\textbf{Loss function} & $\delta_{ \ell_{k,p}}^{(\nu)},  {\nu = {0,1,2}}$\\ \hline
\textbf{Exponential} 
& $\left\{
\left( K^{-}_{k,p},\, K^{+}_{k,p} \right) \in \mathbb{R}^k \times \mathbb{R}^k
\;\middle|\;
K^{-}_{k,p} < \beta < K^{+}_{k,p},
\quad 
\mathbf{p} \in \mathcal{I}_k,\;
k \in \mathcal{I}_{q_{\ell}}
\right\}$
 \\ \hline 
\textbf{QPC} 
& $\left\{
\mathbf{K}_{k,p} \in \mathbb{R}^k :
\mathbf{K}_{k,p} < 0,\;
\mathbf p \in \mathcal I_k,\;
k \in \mathcal I_{q_{\ell}}
\right\}$ \\ \hline 
\end{tabular} 
\caption{Strips of analyticity for $\hat{\ell}_{k,p}^{(\nu)}$.}
\label{tab:strip_analyticity_loss}
\end{table}
For the loss vector  $\mathbf{X}$, we focus on continuous distributional families admitting closed-form extended characteristic functions, namely Gaussian and Normal Inverse Gaussian (NIG). Appendix \ref{appendix:loss_vector_distr} provides the corresponding parameterizations and extended characteristic functions of the marginals $\mathbf X_{k,p}$, while Table \ref{tab:strip_analyticity} summarizes the associated analyticity domains ${\delta_{X_{k,p}}}$.
\begin{table}[H]
\centering
\renewcommand{\arraystretch}{1.5}
\begin{tabular}{|c|c|}
\hline
\textbf{Distribution} & ${\delta_{X_{k,p}}}$ \\ \hline
\textbf{Gaussian}& $\{\, {\mathbf{K}_{k,p}} \in \mathbb{R}^k,\ \mathbf{p} \in \mathcal{I}_k,\ k \in \mathcal{I}_{q_{\ell}} \,\}$ \\ \hline
\hline
\textbf{NIG} & $\left\{ {\mathbf{K}_{k,p}} \in \mathbb{R}^k, \ \left( \alpha_{k,p}^2 - \langle (\boldsymbol \beta_{k,p} - \mathbf{K}_{k,p}), \ \boldsymbol{\Gamma}_{k,p} (\boldsymbol \beta_{k,p} - \mathbf{K}_{k,p}) \rangle \right) > 0 ,\ \mathbf{p} \in \mathcal{I}_k,\ k \in \mathcal{I}_{q_{\ell}} \,\right\}$ \\ \hline
\end{tabular}
\caption{Strip of analyticity for the extended CF $\boldsymbol \Phi _{\mathbf{X}_{k,p}}$.}
\label{tab:strip_analyticity}
\end{table}
The choice of damping parameters $\mathbf{K}_{k,p}^{(\nu)}$ is crucial to control the integrability and smoothness of component integrands. As shown in Figure~\ref{fig:1D_choice_of_damping}, an inappropriate choice of $\mathbf{K}_{k,p}^{(\nu)}$ can produce ill-behaved integrands and might destabilize the numerical optimization procedure in Algorithm \ref{alg:SLSQP_full}. We therefore select  $\mathbf{K}_{k,p}^{(\nu)}$ using the optimal damping rule developed in Section~\ref{subsec:opti_damp}.
\enlargethispage{2\baselineskip} 
\begin{figure}
\centering
\begin{subfigure}[t]{0.5\linewidth}
  \centering
\includegraphics[width=1\linewidth]{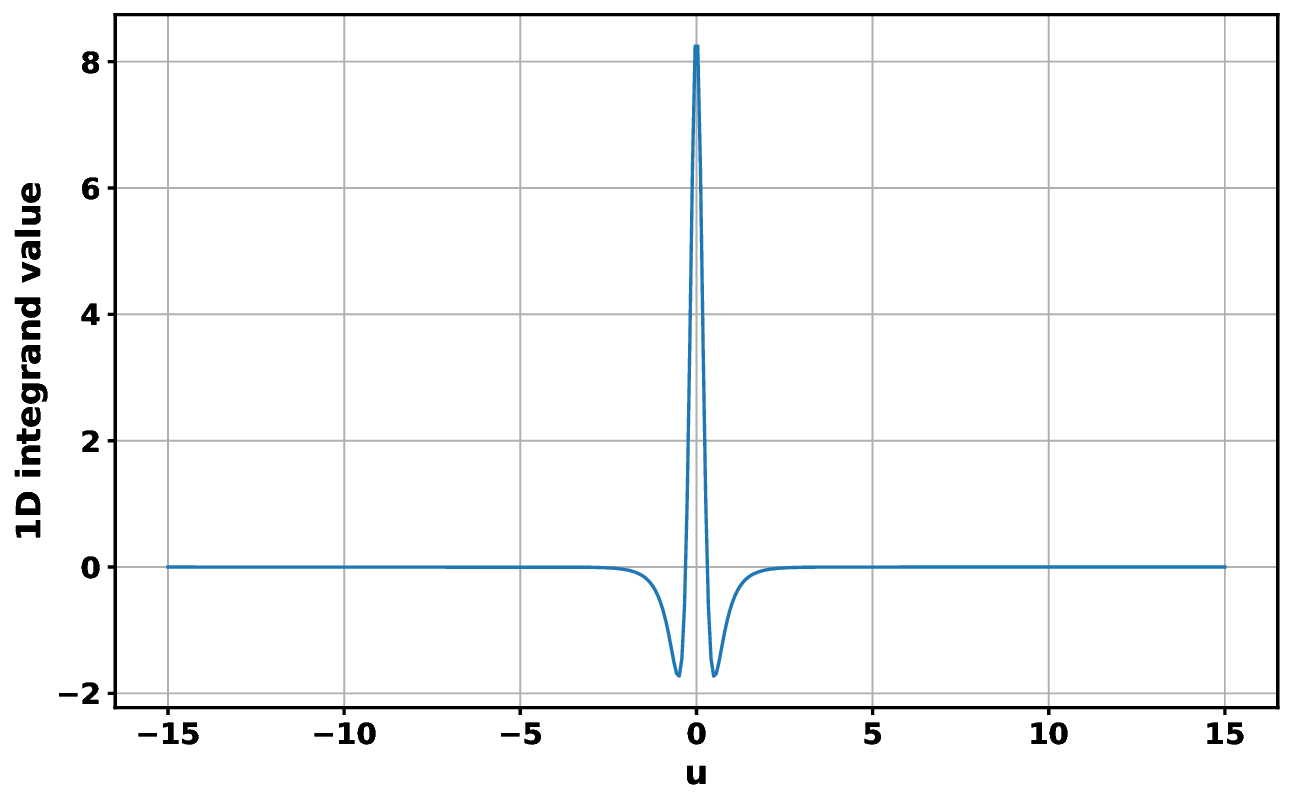}
  \caption{${\mathbf{K}_{1,1}^{(0)}}=0.5$}
\end{subfigure}\hfill
\begin{subfigure}[t]{0.5\linewidth}
  \centering
\includegraphics[width=1\linewidth]{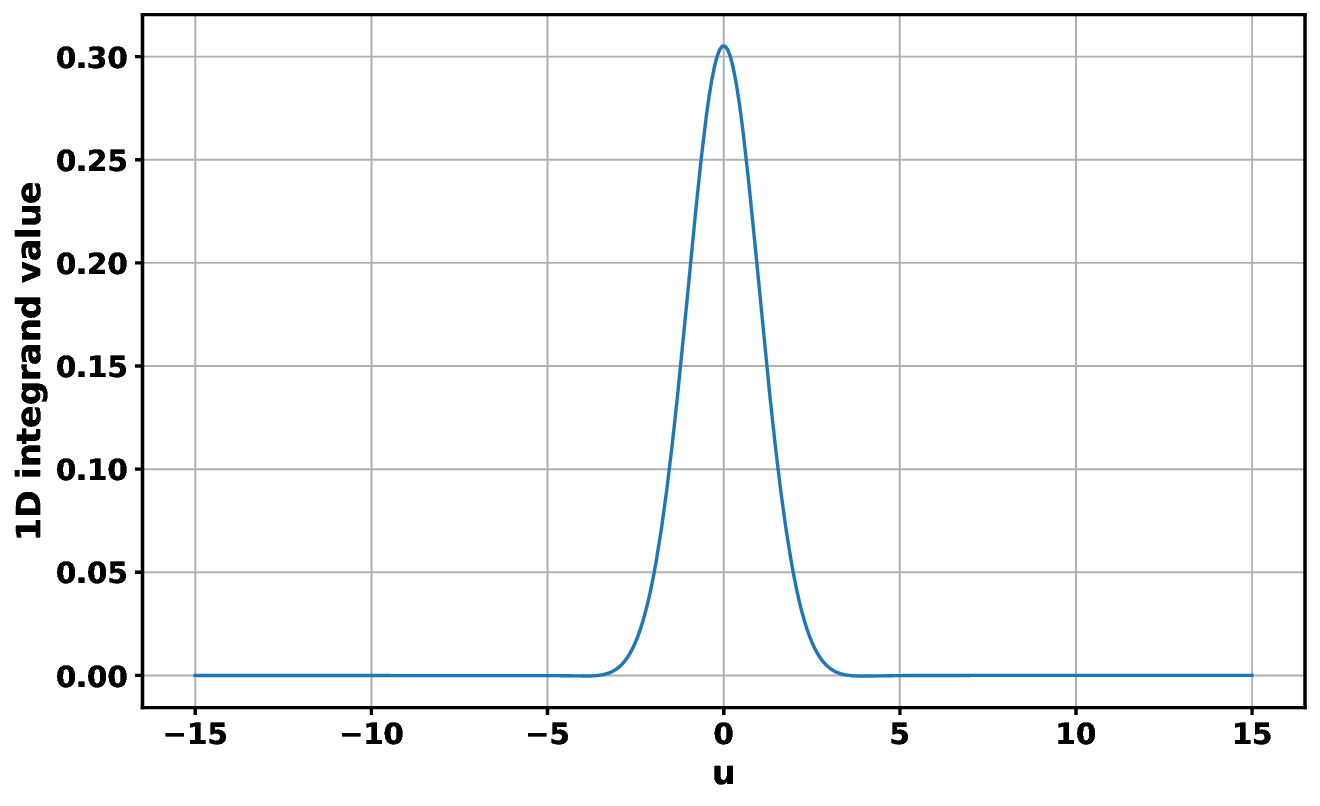}
  \caption{${\mathbf{K}_{1,1}^{(0)}}=2.5$}
\end{subfigure}
\caption{\small Effect of the damping parameter ${\mathbf{K}_{1,1}^{(0)}}$ on the QPC loss integrand component $h^{(0)}_{1,1}$ for a $10$-dimensional Gaussian loss vector (Example in Section \ref{subsec:10D_gaussian_qcl}).}
\label{fig:1D_choice_of_damping}
\end{figure}
\subsection{Optimal Damping Rule}
\label{subsec:opti_damp}
We adopt the optimal damping rule from \cite{bayer_optimal_2023}, which is originally developed in an option-pricing context, and extend it to the MSRM setting. Specifically, we select damping vectors $\mathbf{K}_{k,p}^{(\nu)}$ and update them along the optimization trajectory. From this section onward, the index $(\nu)$ is understood to take values $\nu=0,1$, unless stated otherwise.

Corollary~\ref{coro:optimal_damping} provides the derivation of the optimal damping rule for the component $(k,p)$.

\begin{corollary}[Damping Rule]\label{coro:optimal_damping}
For component integrands $h_{k,p}^{(\nu)}$ defined in \eqref{eq:loss_component_integrands}  with $\mathbf{K}_{k,p}^{(\nu)} \in \delta_{K_{k,p}}^{(\nu)}$, we have
\begin{equation}
\begin{aligned}\label{eq:optimal_damp_loss}
   \mathbf{K}^{(\nu),*}_{k,p}\paren{\mathbf{m}_{k,p},\boldsymbol\Theta_{k,p}}
&:=\argmin_{\mathbf K_{k,p}^{(\nu)}\in\delta_{K_{k,p}}^{(\nu)}}
\sup_{\mathbf u\in\mathbb R^k}\,
\Big|h_{k,p}^{(\nu)}\big(\mathbf u;\mathbf m_{k,p},\mathbf K_{k,p}^{(\nu)},\boldsymbol\Theta_{k,p}\big)\Big|\\
&=\argmin_{\mathbf K_{k,p}^{(\nu)}\in\delta_{K_{k,p}}^{(\nu)}}
\Big|h_{k,p}^{(\nu)}\big(\mathbf 0_{\mathbb R^k};\mathbf m_{k,p},\mathbf K_{k,p}^{(\nu)},\boldsymbol\Theta_{k,p}\big)\Big|\\
&=\argmin_{\mathbf K_{k,p}^{(\nu)}\in\delta_{K_{k,p}}^{(\nu)}}
\upsilon_{k,p}^{(\nu)}\!\left(\mathbf m_{k,p},\mathbf K_{k,p}^{(\nu) },\boldsymbol\Theta_{k,p}\right), 
\end{aligned}
\end{equation}
where
\begin{equation}
\begin{aligned}
\upsilon_{k,p}^{(\nu)}\!\left(\mathbf m_{k,p},\mathbf K_{k,p}^{(\nu)},\boldsymbol\Theta_{k,p}\right)
&:=\ln\Big|h_{k,p}^{(\nu)}\big(\mathbf 0_{\mathbb R^k};\mathbf m_{k,p},\mathbf K_{k,p}^{(\nu)},\boldsymbol\Theta_{k,p}\big)\Big|\\
&= -k\ln(2\pi)
+\langle \mathbf K_{k,p}^{(\nu)}, \mathbf m_{k,p}\rangle
+\ln\Big|\mathbf\Phi_{\mathbf X_{k,p}}\!\big(\mathrm i\,\mathbf K_{k,p}^{(\nu)}\big)\Big|
+\ln\Big|\hat\ell_{k,p}^{(\nu)}\!\big(\mathrm i\,\mathbf K_{k,p}^{(\nu)}\big)\Big|.
\end{aligned}
\label{eq:K_m_relate}
\end{equation}
\end{corollary}
\begin{proof}
Appendix~\ref{appendix:proof_optimal_damping} presents the proof.
\end{proof}
We need to solve \eqref{eq:optimal_damp_loss} numerically, since it is generally not available in closed form. For numerical convenience, we apply a logarithmic transformation to the minimization problem in \eqref{eq:optimal_damp_loss}. Moreover, by Proposition~\ref{prop:psi_convex} and Remark~\ref{rema:postive_definite_esscher}, this log-transformed objective is strictly convex or even strongly convex in $\mathbf{K}_{k,p}^{(\nu)}$, so standard numerical optimization routines typically converge quickly to the minimizer.

In most cases, \eqref{eq:optimal_damp_loss} can be solved efficiently along the optimization trajectory to obtain $\mathbf K_{k,p}^{(\nu),*}$ with given $\mathbf{m}_{k,p}$. A potential issue arises, because $\mathbf{m}_{k,p}$ enters $\nu_{k,p}^{(\nu)}$ linearly (through the term $\langle \mathbf K_{k,p}^{(\nu)}, \mathbf m_{k,p}\rangle$); as a result, the minimizer ${\mathbf K_{k,p}^{(\nu),*}} \paren{\mathbf m_{k,p}}$ might approach the boundary of the analyticity strip $\delta_{K_{k,p}}^{(\nu)}$ for certain iterates $\mathbf{m}_{k,p}$. In this case, minimizing $\upsilon_{k,p}^{(\nu)}$ can drive $\mathbf K_{k,p}^{(\nu)}$ toward the boundary of $\delta_{K_{k,p}}^{(\nu)}$, potentially yielding component integrands that are numerically unstable.
\begin{figure}[H]
    \centering
    \begin{subfigure}[t]{0.48\textwidth}
        \centering
        \includegraphics[width=\linewidth]{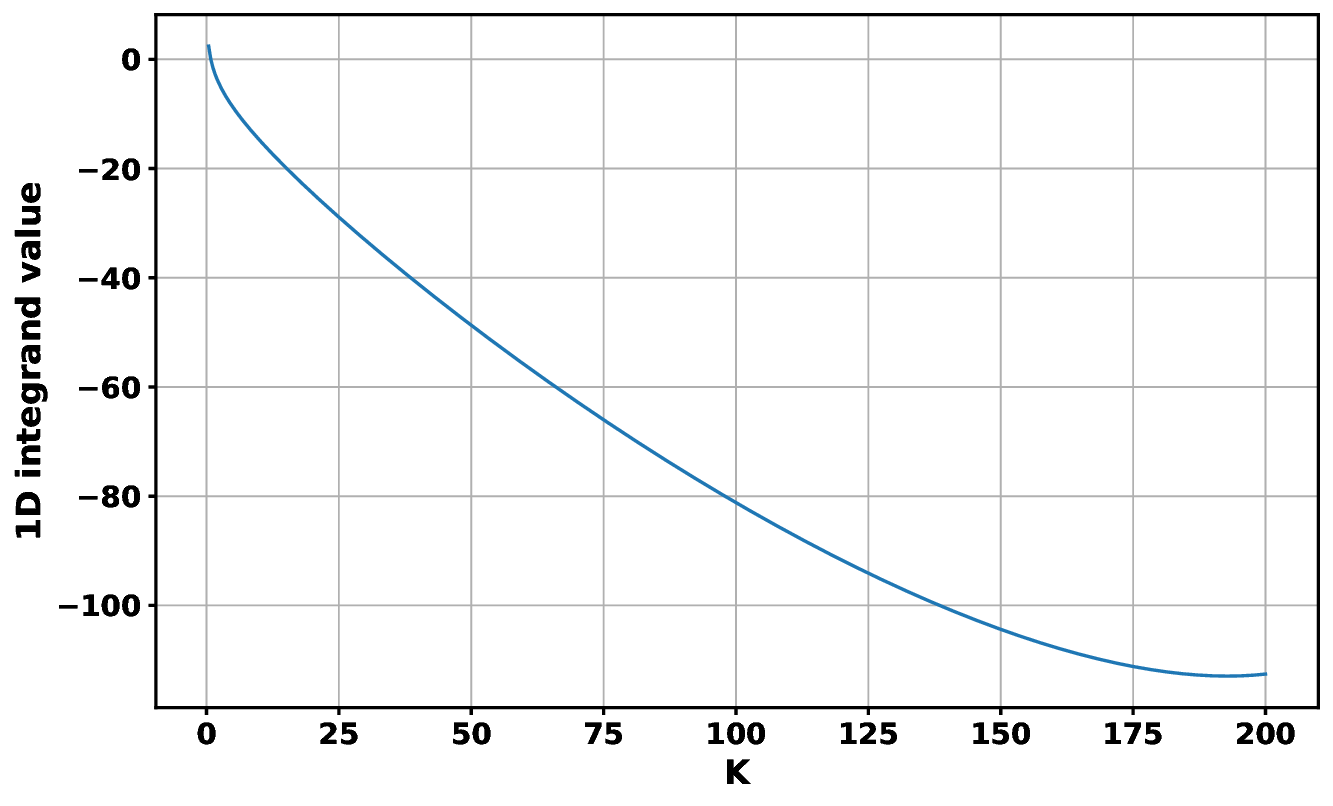}
        \caption{$v_{1,2}^{(0)}$ with $\mathbf{m}_{1,2}=-0.8$, varying damping $\mathbf{K}_{1,2}^{(0)}$.}
        \label{fig:damping_peak}
    \end{subfigure}
    \hfill
    \begin{subfigure}[t]{0.48\textwidth}
        \centering
        \includegraphics[width=\linewidth]{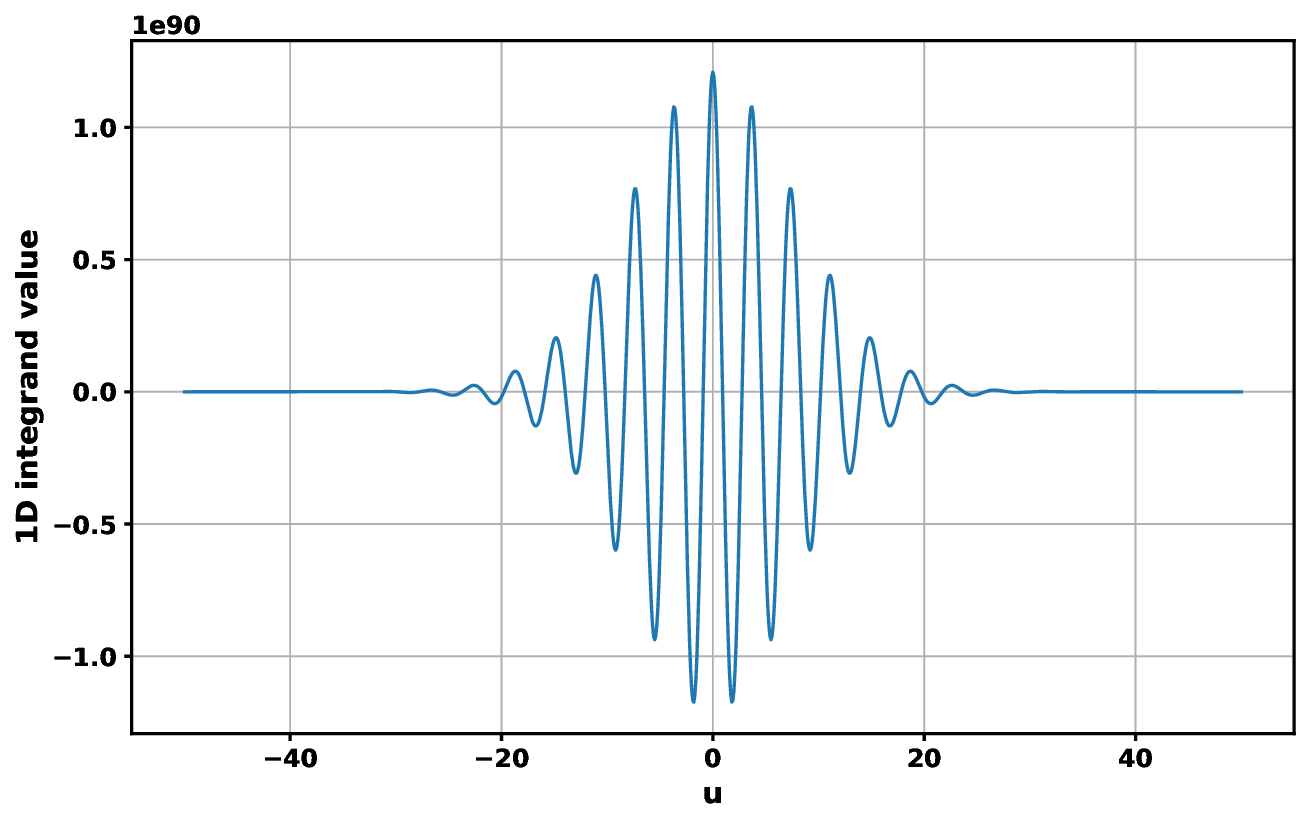}
        \caption{Integrand $h_{1,2}^{(0)}$ at $\mathbf{m}_{1,2}=-0.8$, $\mathbf{K}_{1,2}^{(0)} \approx 200$, varying $\mathbf{u}$.}
        \label{fig:damping_u}
    \end{subfigure}
   \caption{\small Unregularized optimal damping selection for the QPC loss with  a $3$-dimensional  NIG loss vector (Example in Section \ref{subsec:3D_NIG_qcl}).}
    \label{fig:damping_subplots}
\end{figure}
For illustration, as shown in Figure~\ref{fig:damping_peak}, the minimizer of $\upsilon_{1,2}^{(0)}$ is attained close to the boundary of the analyticity strip, at approximately $\mathbf K_{1,2}^{(0)}\approx 200$. Using this value in $h_{1,2}^{(0)}$, Figure~\ref{fig:damping_u} shows that the resulting function becomes highly oscillatory in $\mathbf u$ and attains very large magnitudes. One way to alleviate this problem is to establish the anisotropic Tikhonov-regularization for \eqref{eq:optimal_damp_loss} (see Appendix \ref{appendix:regularization_damping} for more details), which yields to solving the following problem for the optimal damping parameters 
\begin{equation}
    \min_{\mathbf{K}_{k,p}^{(\nu)}} 
    \;\upsilon\paren{\mathbf{m}_{k,p},\mathbf{K}_{k,p}^{(\nu)},\boldsymbol \Theta_{k,p}}
    + \tfrac{\lambda_{k,p}}{2}\|\mathbf{K}_{k,p}^{(\nu)}\|_{\boldsymbol W_{k,p}}^2
    \quad\text{s.t.}\quad \mathbf{K}_{k,p}^{(\nu)}\in\delta_{K_{k,p}}^{(\nu)}.
\label{eq:tikhonov_penalized_damping}
\end{equation}
where $\boldsymbol W_{k,p}\succ 0$ is a  weighting matrix, and $\lambda_{k,p}>0$ is a regularization parameter controlling the strength of the penalty.
\begin{figure}
    \centering
    \begin{subfigure}[t]{0.48\textwidth}
        \centering
        \includegraphics[width=\linewidth]{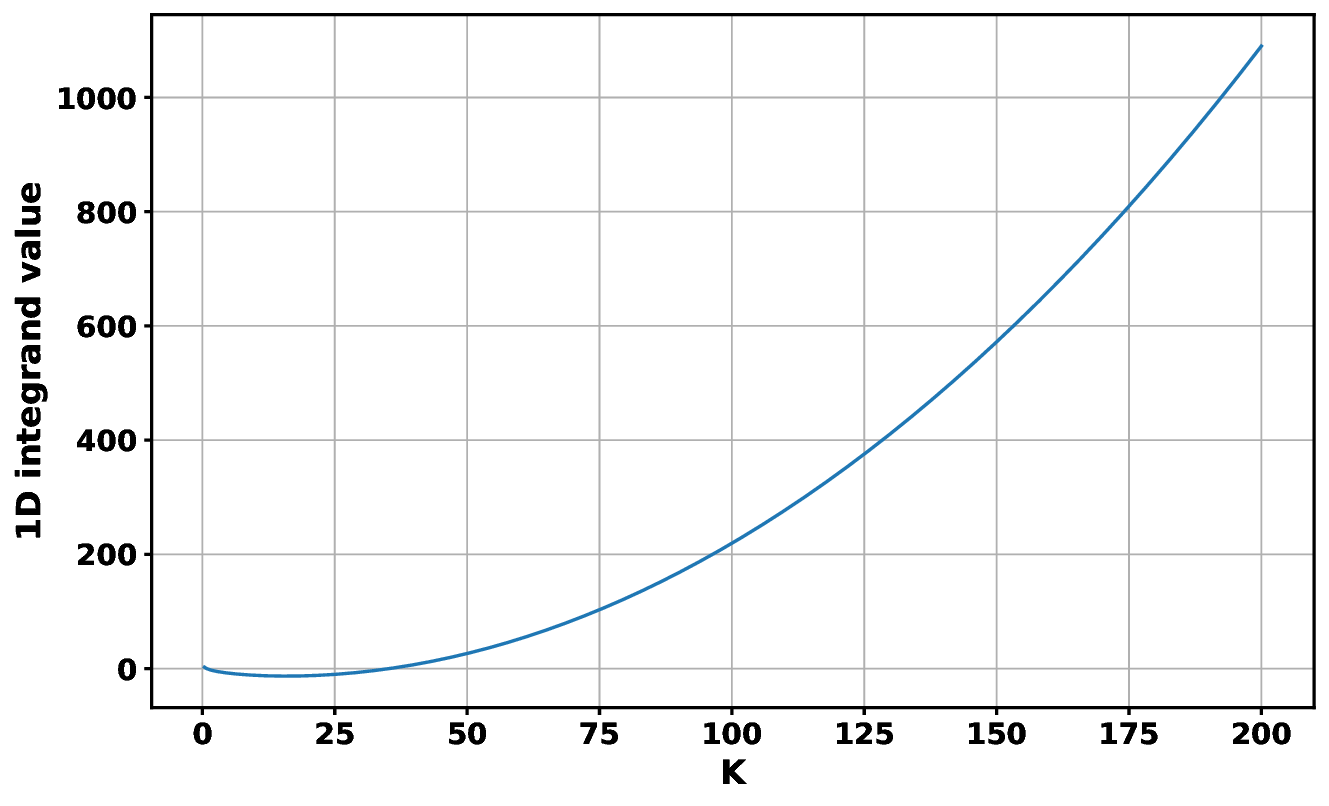}
        \caption{$v_{1,2}^{(0)}$ with $\mathbf{m}_{1,2}=-0.8$, varying damping $\mathbf{K}_{1,2}^{(0)}$.}
        \label{fig:damping_peak_corrected}
    \end{subfigure}
    \hfill
    \begin{subfigure}[t]{0.48\textwidth}
        \centering
        \includegraphics[width=\linewidth]{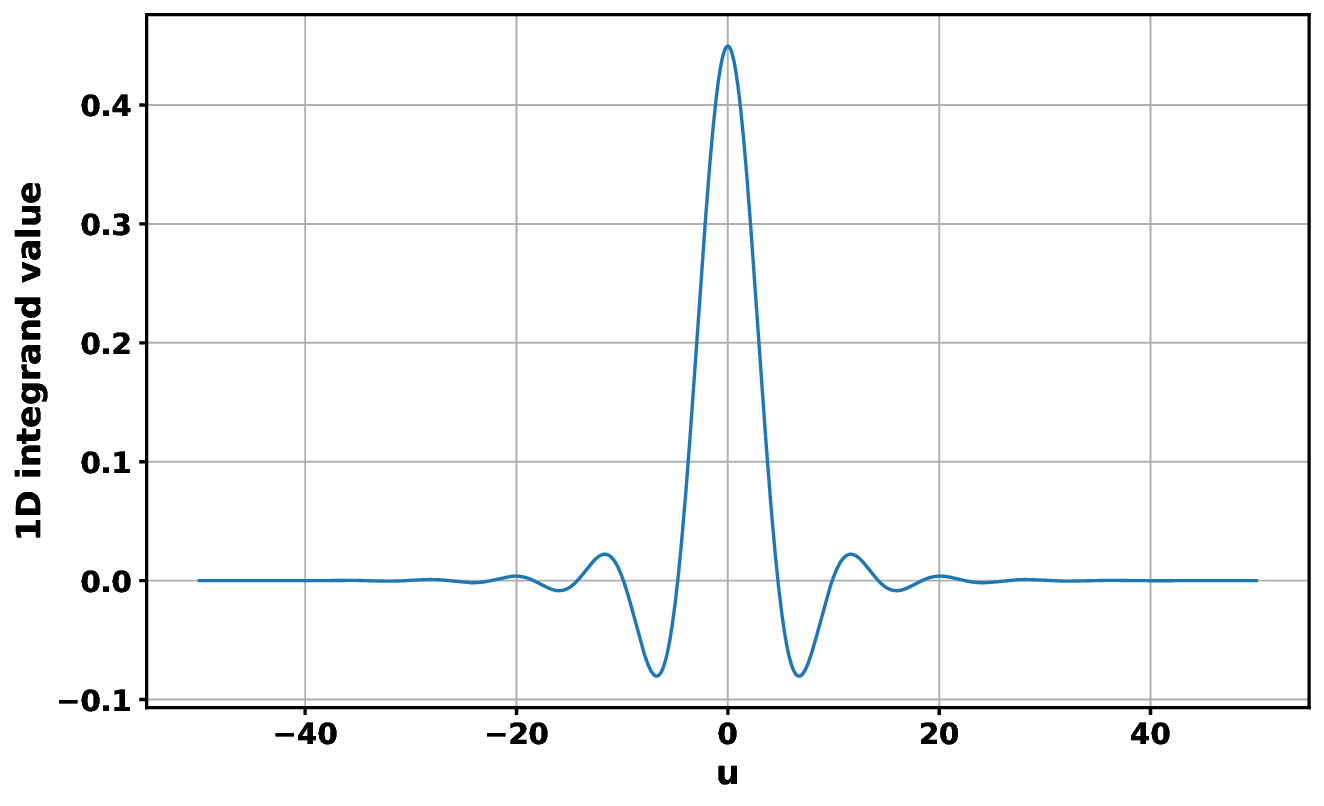}
         \caption{Integrand $h_{1,2}^{(0)}$ at $\mathbf{m}_{1,2}=-0.8$, $\mathbf{K}_{1,2}^{(0)} \approx 5$, varying $\mathbf{u}$.}
        \label{fig:damping_u_corrected}
    \end{subfigure}
    \caption{\small  Regularized damping selection for the QPC loss with   a $3$-dimensional NIG loss vector (parameter setting in Section~\ref{subsec:3D_NIG_qcl}), using the anisotropic weighting matrix $\boldsymbol W_{1,2}$.
Compared to Figure \ref{fig:damping_subplots}, regularization shifts the optimal damping away from the boundary of the analyticity strip, yielding a smoother and better-conditioned integrand $h^{(0)}_{1,2}$.}
    \label{fig:damping_subplots_corrected}
\end{figure}

With the inclusion of the anisotropic term $\boldsymbol W_{1,2}$, the minimizer of $\upsilon_{1,2}^{(0)}$ is attained at 
$\mathbf{K}_{1,2}^{(0)}\approx 5$. In this case, the resulting integrand exhibits a much more favorable shape compared to the scenario where only the peak $\mathbf{u}=0$ is minimized. 

Having determined the optimal choice of $\mathbf{K}_{k,p}^{(\nu)}$, we can next set up the suitable numerical scheme to compute ${h}_{k,p}^{(\nu)}$ in Section \ref{sec:Methodology of our Approach}. We summarize the main idea in Algorithm \ref{alg:Fou_opti_damp} for computing the optimal damping vector $\mathbf{K}_{k,p}^{(\nu)}, $ for the component integrands along the optimization trajectory.
\begin{remark}[Choosing $\lambda_{k,p}$]
In our numerical experiments, we observe boundary-hugging behavior (Figure~\ref{fig:damping_peak}) for the NIG loss vector $\mathbf X$.
Accordingly, in the Gaussian case we set $\lambda_{k,p}=0$.
In contrast, for the NIG case, we set a positive penalty and choose $\lambda_{k,p}\in[0.1,0.5]$ to ensure that the optimizer $\mathbf K_{k,p}^{(\nu)}$ remains in a reasonable interior region.
\label{rema:choose_lambda_pen}
\end{remark}

\begin{algorithm}[H]
\caption{Selecting optimal damping vectors at each optimization step $j$}
\label{alg:Fou_opti_damp}
\begin{algorithmic}[1]
\Require components $h_{k,p}^{(\nu)}$ in \eqref{eq:loss_component_integrands}, allocation $\mathbf{m}_{k,p}^{(j)}$, marginal distribution $\mathbf{X}_{k,p}$ of $\mathbf{X}$ .
\State  
Find $\mathbf{K}_{k,p}^{(\nu,j)}$ by solving the optimization problem \eqref{eq:tikhonov_penalized_damping} 
for $h_{k,p}^{(\nu)}$. The choice of the regularization parameter $\lambda_{k,p}$ follows Remark \ref{rema:choose_lambda_pen}. The resulting problem can be efficiently solved using a numerical optimizer (e.g., SLSQP or trust-constr).
\end{algorithmic}
\end{algorithm}

\section{{Single- and Multilevel RQMC Approximation of Fourier-Based MSRM Integrals}}\label{sec:Methodology of our Approach}
Building on the Fourier representations derived in Section \ref{sec:mapping_MSRM_Fourier}, we now approximate the resulting Fourier integrals numerically. These integrals can be moderately high-dimensional and must be evaluated repeatedly along the iterates of the constrained optimization algorithm. We therefore use (R)QMC methods, which are computationally efficient, provide practical error quantification, and perform well in moderate dimensions. Section \ref{subsec:sing-RQMC} introduces a single-level Fourier–RQMC estimator with a suitable domain transformation, and Section \ref{subsec:Multilevel Fou-RQMC} develops a multilevel extension that exploits the geometric convergence of the deterministic optimizer. The resulting estimators are subsequently employed as deterministic surrogate models for objective and gradient evaluations within the constrained optimization algorithm \ref{alg:SLSQP_full}.

\subsection{Single-Level Fourier–RQMC Approximation with Domain Transformation}\label{subsec:sing-RQMC}
We begin by constructing a single-level RQMC estimator for the Fourier-based integrals derived in Section \ref{sec:mapping_MSRM_Fourier}. This requires mapping the Fourier integration domain $\mathbb{R}^k$ to the unit cube $[0,1]^k$ and applying an RQMC rule with tractable error estimation.

The component integrands $h_{k,p}^{(\nu)}$ in \eqref{eq:loss_component_integrands} are defined on $\mathbb{R}^k$, $1 \le k \le q \le d$. To apply (R)QMC methods, we perform a change of variables $\mathbf{v} = G(\mathbf{u})$ mapping $\mathbb R^k$ to the unit cube $[0,1]^k$. The transformation $G$ is drawn from a fixed, distribution-driven family whose functional form is independent of $(k,p)$, while its dimension is determined by $k$ and its parameters  may depend on $(k,p)$ (see Section \ref{subsubsec:domain_trans_QMC} for details). The resulting transformed integrands are given by
\begin{equation}\label{eq:transformed_components}
    \tilde h_{k,p}^{(\nu)}\left(\mathbf{v};\cdot\right) := 
h_{k,p}^{(\nu)}\left(\mathbf{u};\cdot\right)\,
  \bigl|\det J_{G^{-1}}(\mathbf v; \cdot)\bigr|
\end{equation}
Here $J_{G^{-1}}(\mathbf v;\cdot)$ denotes the Jacobian matrix of the inverse transformation $G^{-1}$  w.r.t. $\mathbf v$. We assume that $G$ is invertible almost everywhere with an almost-everywhere differentiable inverse.

 The QMC estimator for the integral of transform component integrands $\tilde h_{k,p}^{(\nu)}:[0,1]^k \to \mathbb R$ is an $N$-point equal-weight quadrature rule, denoted by 
$I_{N}^{\mathrm{QMC}}$, defined as
\begin{equation} \label{eq:QMC_estimator}
\hat g^{(\nu),\mathrm{Fou}}_{k,p}\paren{\mathbf{m}_{k,p}} \;=\; \int_{[0,1]^k} \tilde h_{k,p}^{(\nu)}\paren{\mathbf{v};\mathbf{m}_{k,p}}\, \mathrm{d}\mathbf v
    \;\approx\; I_{N}^{\mathrm{QMC}}\brac{ \tilde h_{k,p}^{(\nu)}\paren{\cdot;\mathbf{m}_{k,p}}}
    := \frac{1}{N}\sum_{n=1}^N \tilde h_{k,p}^{(\nu)}\paren{\mathbf v_n; \mathbf{m}_{k,p}},
\end{equation}
where $\{\mathbf v_n\}_{n=1}^N \subset [0,1]^k$ is a deterministic low-discrepancy 
sequence (e.g., Halton, Faure, Sobol; see \cite{dick_high-dimensional_2013} for details). The advantage of the QMC estimator in \eqref{eq:QMC_estimator} over standard MC lies in the more uniform coverage of the unit cube $[0,1]^k$ provided by low-discrepancy sequences, which often leads to improved convergence in practice. However, since the quadrature points are deterministic and exhibit strong dependence, the classical i.i.d. central limit theorem (CLT) does not apply directly, and probabilistic error bounds are not immediately available. Instead, convergence of deterministic QMC estimators is typically analyzed via discrepancy-based bounds, most notably the Koksma–Hlawka inequality \cite{hlawka_funktionen_1961}. To evaluate this error bound, we need to compute the integral involving the first mixed partial derivatives of $\tilde h_{k,p}^{(\nu)}$, which is often more difficult than evaluating the original integrand $\tilde h_{k,p}^{(\nu)}$ itself. To recover probabilistic error quantification while retaining the favorable space-filling properties of QMC, we employ a randomized version of the estimator in \eqref{eq:QMC_estimator}, referred to as the RQMC estimator \cite[Chapter 17]{owen_practical_2023}, which is expressed as
\begin{equation}
I^{\text{RQMC}}_{N,S_{\mathrm{shift}}}\brac{\tilde h_{k,p}^{(\nu)}\paren{\cdot;\mathbf{m}_{k,p}}} := \frac{1}{S_{\text{shift}}} \sum_{s=1}^{S_{\text{shift}}} \frac{1}{N} 
    \sum_{n=1}^{N} 
    \tilde h_{k,p}^{(\nu)}\left(\mathbf{v}_n^{(s)};\mathbf{m}_{k,p}\right),
\label{eq:RQMC_estimator}
\end{equation}
where $\curly{\mathbf{v}_n^{(s)}}_{s=1}^{S_{\mathrm{shift}}}$ is obtained by applying independent digital shifts to the underlying deterministic digital net $\{v_n\}_{n=1}^N \subset [0,1]^k$, while preserving the low-discrepancy structure. Various randomization schemes exist, each with different theoretical guarantees (see, e.g., \cite[Chapter 17]{owen_practical_2023}). In this work, we employ Sobol sequences \cite{sobol_construction_2011} with \emph{digital shifting} \cite[Section 5.2]{lecuyer_recent_2002} as our randomization method. In order for this randomization to yield a valid RQMC estimator, we need the following assumption
\begin{assumption}[Square-integrability of transformed integrands]\label{ass:A10}
For each $(k,p)$ and for the selected damping vectors $\mathbf{K}_{k,p}^{(\nu)}$ along the optimization iterates, the transformed integrands
\(
\tilde h_{k,p}^{(\nu)}
\)
belong to $L^2\!\left([0,1]^k\right)$, with $\nu = 0,1,2$.
\end{assumption}
Under Assumption~\ref{ass:A10} and using independent digital shifts, the RQMC estimator~\eqref{eq:RQMC_estimator} is unbiased, i.e.,
\[
\mathbb{E}\!\left[ I^{\mathrm{RQMC}}_{N,S_\mathrm{shift}}\brac{\tilde h_{k,p}^{(\nu)}\paren{\cdot;\mathbf{m}_{k,p}}} \right]
=\hat g^{(\nu),\mathrm{Fou}}_{k,p}\paren{\mathbf{m}_{k,p}},
\]
 and enables us to derive the root mean squared error (RMSE) of the estimator:
\begin{equation}
\varepsilon_{N,S_{\mathrm{shift}}}^{\mathrm{RQMC}}\brac{\tilde h_{k,p}^{(\nu)}\paren{\cdot;\mathbf{m}_{k,p}}} = C_{\alpha}\sqrt{\frac{1}{S_{\mathrm{shift}}(S_{\mathrm{shift}}-1)} \sum_{s=1}^{S_{\text{shift}}} \left(\frac{1}{N} 
    \sum_{n=1}^{N} 
    \tilde h_{k,p}^{(\nu)}\left(\mathbf{v}_n^{(s)};\mathbf{m}_{k,p}\right) -I^{\text{RQMC}}_{N,S_{\mathrm{shift}}}\brac{\tilde h_{k,p}^{(\nu)}\paren{\cdot;\mathbf{m}_{k,p}}} \right)^2}
\label{eq:RMSE_RQMC}
\end{equation}
where $C_{\alpha}$ denotes the $(1-\frac{\alpha}{2})$-quantile of the standard normal distribution for a confidence level $0 < \alpha \ll1$.

Since the smoothness of the transformed integrands $\tilde h^{(v)}_{k,p}$ near the boundary of $[0,1]^k$ depends critically on the choice of the domain transformation $G$, we adopt the boundary-singularity framework of \cite{owen_halton_2006} to characterize the resulting convergence rate of the RQMC estimator in \eqref{eq:RQMC_estimator}. Specifically, there exists $C<\infty$ such that 
\begin{equation}
\label{eq:bound_growth_derivatives_RQMC}
\bigl| \partial^{\boldsymbol{\kappa}} \tilde h^{(\nu)}_{k,p}(\mathbf{v}) \bigr|
\;\le\;
C \prod_{j=1}^{k} \min(v_j,1-v_j)^{-A^{(\nu)}_{j,p}-\mathbf{1}_{j\in \boldsymbol{\kappa}}},
\qquad
\forall\, \boldsymbol{\kappa} \subseteq \{1,\dots,k\},
\end{equation}
where $
\partial^{\boldsymbol{\kappa}} \tilde h^{(\nu)}_{k,p}(\mathbf{v})
\;:=\;
\prod_{j\in\boldsymbol{\kappa}}
\frac{\partial}{\partial v_j}
\,\tilde h^{(\nu)}_{k,p}(\mathbf{v})$,  and the exponents
\(A^{(\nu)}_{j,p} > 0\) quantify the boundary growth of the transformed integrands and its derivatives as
\(v_j \to 0\) or \(v_j \to 1\).
To obtain a single convergence exponent that is valid uniformly across all transformed integrands and will be used later in the statistical error analysis of Section \ref{subsec:stat_err_analysis}, we define the worst-case boundary singularity exponent
\begin{equation}
    A_{\mathrm{sing}}^*
:= \max_{\nu\in\{0,1\}}\max_{k \in \mathcal{I}_{q_{\ell}}} \max_{\mathbf{p}\in\mathcal I_k}  \max_{1\le j\le k} A_{j,p}^{(\nu)}.
\label{eq:worst_exponent_component}
\end{equation}
Then, for any $\varsigma>0$, \cite[Theorem~5.7]{owen_halton_2006} implies that the RQMC estimator
\begin{equation}
\varepsilon_{N,S_{\mathrm{shift}}}^{\mathrm{RQMC}}\brac{\tilde h_{k,p}^{(\nu)}\paren{\cdot;\mathbf{m}_{k,p}}}
= \mathcal{O}\!\left(N^{-r}\right),
\label{eq:stat_error_RQMC_rate}
\end{equation}
with $r := 1 - A_{\mathrm{sing}}^* - \varsigma$. Equation \eqref{eq:stat_error_RQMC_rate} shows that RQMC can  outperform MC when $A_{j,p}^{(\nu)}<\tfrac12$, although the convergence rate deteriorates as the boundary singularities become more severe. Moreover, \cite{liu_randomized_2025-1}  provides
a complementary spectral interpretation by linking the boundary-growth condition~\eqref{eq:bound_growth_derivatives_RQMC} to the decay of Fourier/Walsh coefficients. In particular,
larger exponents $A^{(v)}_{j,p}$ (and hence a larger $A^*_{\mathrm{sing}}$)
correspond to slower spectral decay and increased oscillatory behavior near
the boundary. This observation motivates designing a domain transformation $G$ so that the transformed integrands $\tilde h_{k,p}^{(\nu)}$ exhibit sufficiently mild boundary growth, as discussed in detail in Section \ref{subsubsec:domain_trans_QMC}. We summarize the resulting single-level Fourier–RQMC procedure in  Algorithm~\ref{alg:Single-level RQMC}.
\begin{remark}
In  \eqref{eq:stat_error_RQMC_rate}, the exponent $r$ satisfies $r\le 1$ under the general boundary-growth assumptions adopted here. In more specific settings, higher (R)QMC rates may be attainable. For instance, when RQMC is combined with importance sampling, asymptotic rates of order $\mathcal O\!\paren{N^{-\frac{3}{2}+\varsigma}}$ have been reported \cite{ouyang_achieving_2024}. Moreover, suitably chosen domain transformations can substantially improve integrand regularity, which may translate into markedly better non-asymptotic error decay in practice \cite{bayer_quasi-monte_2025, liu_nonasymptotic_2026}.
\label{rema:non_asymptoic_err_decay}
\end{remark}
\begin{remark}
The exponent $A^*_{\mathrm{sing}}$ yields a uniform (worst-case) convergence rate for the RQMC estimator across all components $(k,p,\nu)$ and along the optimization trajectory. Such a worst-case bound is essential for the subsequent error propagation analysis in Section \ref{sec:num_analysis}. In practice, however, the observed RQMC convergence can be substantially faster than the rate predicted by $A^*_{\mathrm{sing}}$, owing to effective low dimensionality,  the smoothing effect of the domain transformation~$G$, and  boundary growth that is milder than the worst-case behavior permitted by the theoretical bounds. Less conservative bounds can be obtained by retaining the coordinate- and component-wise boundary exponents and  anisotropic constructions (e.g., transform employing tuning or weighted QMC rules). We do not pursue such refinements here and instead leverage variance reduction through iteration-indexed multilevel differences.
\end{remark}

\begin{algorithm}[H]
\small
\caption{Single-level Fourier-RQMC at optimization step $j$}
\begin{algorithmic}[1]
\Require Allocation $\mathbf m^{(j)}$, Sobol size $N\in\mathbb N$, shifts $S_{\mathrm{shift}}$. For $k \in \mathcal{I}_{q_{\ell}}$: index set $\mathcal I_k$; for $\mathbf{p}\in\mathcal I_k$: integrands $h_{k,p}^{(\nu)}\paren{\mathbf v;\mathbf m_{k,p}^{(j)}}$, marginal $\mathbf X_{k,p}$, damping vectors $\mathbf K_{k,p}^{(\nu,j)}$ (Algorithm~\ref{alg:Fou_opti_damp}).

  \State Based on the distribution of $\mathbf{X}_{k,p}$, apply the appropriate 
transformation from \eqref{eq:domain_transform_MN} or 
\eqref{eq:domain_transform_MNIG} to obtain the transformed integrands 
$\tilde{h}_{k,p}^{(\nu)}\paren{\mathbf{v}; \mathbf{m}_{k,p}^{(j)}}$.

\For{$k \in \mathcal{I}_{q_{\ell}}$}
  \State \textbf{(Digital net generation)}
  \State Generate once per $k$ an unshifted base-2 digital net $\{\mathbf{u}^{(s)}_n\}_{n=1}^{N} \subset [0,1]^{k}$ and reuse it for all $\mathbf{p} \in \mathcal{I}_k$ \footnotemark.
  \State Draw a digital shift $\boldsymbol{\Delta}^{(s)}\in[0,1]^{k}$ (seed $1:S_{\mathrm{shift}}$) and set
        \[
          \mathbf{v}^{(s)}_n \;=\; \mathbf{u}^{(s)}_n \,\oplus\, \boldsymbol{\Delta}^{(s)}, 
          \qquad n=1,\dots,N,
        \]
        where $\oplus$ denotes the base-2 digital (bitwise XOR) shift.
\EndFor

\State Set ${I}_{N,S_{\mathrm{shift}}}^{\mathrm{RQMC}} \brac{\tilde h^{(\nu)} \paren{\cdot;\mathbf{m}^{(j)}}}\gets 0$.
\For{$k \in \mathcal{I}_{q_{\ell}}$}
  \For{$\mathbf{p} \in \mathcal{I}_k$}
    \State Compute the RQMC estimate ${I}_{N,S_{\mathrm{shift}}}^{\mathrm{RQMC}}\brac{\tilde h_{k,p}^{(\nu)}\paren{\cdot;\mathbf{m}_{k,p}^{(j)}}}$ using $\{\mathbf v^{(s)}_n\}_{n=1}^{N}$.
    
    \State ${I}_{N,S_{\mathrm{shift}}}^{\mathrm{RQMC}} \brac{\tilde h^{(\nu)}\paren{\cdot;\mathbf{m}^{(j)}}} \mathrel{+}={I}_{N,S_{\mathrm{shift}}}^{\mathrm{RQMC}}\brac{\tilde h_{k,p}^{(\nu)}\paren{\cdot;\mathbf{m}_{k,p}^{(j)}}}$
  \EndFor
\EndFor
\State \Return ${I}_{N,S_{\mathrm{shift}}}^{\mathrm{RQMC}} \brac{\tilde h^{(\nu)}\paren{\cdot;\mathbf{m}^{(j)}}}$.
\end{algorithmic}
\label{alg:Single-level RQMC}
\end{algorithm}
\footnotetext{For the NIG transform in \eqref{eq:domain_transform_MNIG}, an additional mixing variable introduces one extra integration dimension. Hence, we generate $\{\mathbf{u}^{(s)}_n\}_{n=1}^{N} \subset [0,1]^{k+1}$ instead of $[0,1]^k$.}

\subsubsection{Oscillation-Aware, Distribution-Dependent Domain Transformation}\label{subsubsec:domain_trans_QMC}

To control oscillatory behavior and boundary growth in the Fourier-based integrands, we introduce a distribution-dependent, oscillation-aware change of variables $G$, mapping $\mathbb{R}^k$ to $[0,1]^k$, for $k \in \mathcal{I}_{q_{\ell}}$.  Our construction builds on ideas from Fourier-based option pricing \cite{bayer_quasi-monte_2025}, but is adapted to the multivariate risk setting and complemented by a dedicated analysis of the induced oscillatory behavior.

We  rewrite $h_{k,p}^{(\nu)}$ in \eqref{eq:loss_component_integrands} as a standard oscillatory (Fourier-type) integrand, with a (complex) amplitude $a_{k,p}^{(\nu)}\paren{\mathbf{u};\mathbf{m}_{k,p},\mathbf{K}_{k,p}^{(\nu)},\boldsymbol \Theta_{k,p}}$ and oscillatory phase $w_{k,p}\paren{\mathbf{u};\mathbf{m}_{k,p}}$ as follows:
\begin{equation*}
    h_{k,p}^{(\nu)}(\mathbf{u};\mathbf{m}_{k,p},\mathbf{K}_{k,p}^{(\nu)},\boldsymbol \Theta_{k,p})
    \;=\;
    a_{k,p}^{(\nu)}\paren{\mathbf{u};\mathbf{m}_{k,p},\mathbf{K}_{k,p}^{(\nu)},\boldsymbol \Theta_{k,p}}\,
    \exp\!\bigl(-\mathrm{i}\,w_{k,p}(\mathbf{u};\mathbf{m}_{k,p})\bigr),
\end{equation*}
where:
\begin{equation*}
    \begin{aligned}
a_{k,p}^{(\nu)}\paren{\mathbf{u};\mathbf{m}_{k,p},\mathbf{K}_{k,p}^{(\nu)},\boldsymbol \Theta_{k,p}}
    &:=
    (2\pi)^{-k}
    \exp\!\bigl(\langle \mathbf{K}_{k,p}^{(\nu)}, \mathbf{m}_{k,p}\rangle\bigr)\,
    \Phi_{\mathbf{X}_{k,p}}\paren{\mathbf{u}+\mathrm{i}\mathbf{K}_{k,p}^{(\nu)}}\,
    \hat{\ell}_{k,p}^{(\nu)}\paren{\mathbf{u}+\mathrm{i}\mathbf{K}_{k,p}^{(\nu)}}, \\
    w_{k,p}(\mathbf{u};\mathbf{m}_{k,p})
    &:= \mathbf{m}_{k,p}^\top \mathbf{u},   
    \end{aligned}
\end{equation*}
with $\mathbf{u} \in \mathbb{R}^k, \mathbf{K}_{k,p}^{(\nu)} \in \delta_{K_{k,p}}^{(\nu)}$. The transformed integrand is then expressed as:
\begin{equation}
    \tilde h_{k,p}^{(\nu)}(\mathbf{v};\mathbf{m}_{k,p})=\;a_{k,p}^{(\nu)}\paren{\mathbf{u};\mathbf{m}_{k,p},\mathbf{K}_{k,p}^{(\nu)},\boldsymbol \Theta_{k,p}}\,\exp{\paren{-\mathrm{i}\,\varpi(\mathbf{v};\boldsymbol \Theta_{k,p})}}\,\big|\det J_{G^{-1}}(\mathbf{v};\boldsymbol\Theta_{k,p})\big|
\label{eq:transform_integrand_v}
\end{equation}
with \(
\varpi\paren{\mathbf{v};\boldsymbol \Theta_{k,p}}
  :
  = \mathbf{m}_{k,p}^\top G^{-1}\paren{\mathbf{v};\boldsymbol\Theta_{k,p})}.
\)

If the domain transformation $G$ is not chosen appropriately, it can amplify oscillations of the transformed integrand near the boundary of $[0,1]^k$ and thereby deteriorate the convergence of RQMC methods. To guide the choice of $G$, we analyze in Appendix~\ref{Appendix:boundary_oscillation} the boundary oscillatory behavior of $\tilde h_{k,p}^{(\nu)}$ and derive distribution-dependent oscillation counts. These results motivate adopting the density-driven change of variables as proposed in \cite{bayer_quasi-monte_2025}. 

The effectiveness of the domain transformation is governed by the choice of a reference density
 $\psi(\,\cdot\,;\boldsymbol\Theta_{\mathbf X_{k,p}})$, with associated shape matrix $\tilde{\boldsymbol\Sigma}_{k,p}$.
This reference density is chosen to control the boundary growth of the transformed integrand
$\tilde h_{k,p}^{(\nu)}\paren{\mathbf v,\mathbf m_{k,p}}$ defined in \eqref{eq:transform_integrand_v}.  For Gaussian marginals, the reference density is Gaussian, whereas for NIG marginals 
we employ an auxiliary exponential (Laplace-type) reference density arising from the 
mixture representation.

Given $\tilde{\boldsymbol\Sigma}_{k,p}$, let $\widetilde {\boldsymbol{L}}_{k,p}$ denote a Cholesky factor such that
$\tilde{\boldsymbol\Sigma}_{k,p} = \widetilde{\boldsymbol L}_{k,p}\widetilde{\boldsymbol L}_{k,p}^{\top}$.
For the loss vector models considered in Section \ref{sec:mapping_MSRM_Fourier}, we employ the following distribution-dependent inverse transformations $G^{-1}$.
\begin{itemize}
\item \textbf{Gaussian.}  
Let $G^{-1}=T_{\mathrm{Gauss}}:[0,1]^k\to\mathbb{R}^k$ be defined by
\[
\mathbf v_{1:k} \;\mapsto\; \mathbf u := \tilde{\boldsymbol L}_{k,p}\,\Psi^{-1}(\mathbf v_{1:k}; \mathbf I_k),
\]
where $\Psi$ denotes the standard Gaussian CDF applied componentwise.
\item \textbf{NIG.}  
Let $W\sim\mathrm{Exp}(1)$ be an auxiliary mixing variable with CDF $\Psi_W(w)=1-e^{-w}$.
Define $G^{-1}=T_{\mathrm{NIG}}:[0,1]^{k+1}\to\mathbb{R}^k\times(0,\infty)$ by
\[
(\mathbf v_{1:k},v_{k+1}) \;\mapsto\; (\mathbf u,w)
:=\Big(\sqrt{\Psi_W^{-1}(v_{k+1})}\,\tilde{\boldsymbol L}_{k,p}\Psi^{-1}(\mathbf v_{1:k};\mathbf I_k),\;\Psi_W^{-1}(v_{k+1})\Big).
\]
\end{itemize}
where $\mathbf v_{1:k}:=(v_1,\ldots,v_k)$, and $\boldsymbol I_k$ denotes the $k\times k$ identity matrix.
 
The resulting choice for $\tilde {\boldsymbol \Sigma}_{k,p}$ in this work is presented in Table \ref{tab:choice_of_matrix_trans}.  We derive this choice in detail in Appendix \ref{Appendix:choice_of_trans_matrix}, where  we quantify the boundary oscillations induced by the transformation and show how the choice of $ \boldsymbol{\tilde\Sigma}_{k,p}$ and the scaling parameter $c$ controls the resulting oscillatory behavior. The scalar parameter $c>1$ acts as a regularity control, trading off boundary oscillations
against concentration of the transformed integrand.
\begin{table}[H]
\centering
\renewcommand{\arraystretch}{2}
\begin{tabular}{|c|c|}
\hline
\textbf{Distribution} & $\tilde {\boldsymbol \Sigma}_{k,p}$ \\ \hline
\textbf{Gaussian}& $\displaystyle c \, {\boldsymbol \Sigma}_{k,p}^{-1}  $ \\ \hline
\textbf{NIG} & $\displaystyle \frac{2c}{\delta_{k,p}^2 }\, \boldsymbol{\Gamma}_{k,p}^{-1}$ \\ \hline
\end{tabular}
\caption{Choice of $\tilde{\boldsymbol\Sigma}_{k,p}$. In the Gaussian case, $\boldsymbol\Sigma_{k,p}$ is defined in Example~\ref{ex:Gaussian_components}. In the NIG case, $\boldsymbol\Gamma_{k,p}$ and $\delta_{k,p}$ are defined in Example~\ref{ex:MNIG_components}. The scaling parameter $c>1$.}
\label{tab:choice_of_matrix_trans}
\end{table}
With this choice of reference density and associated scaling, the domain transformation for the Gaussian case takes the form
\begin{equation}  
\int_{\mathbb{R}^k} h_{k,p}^{(\nu)}(\mathbf u;\mathbf{m}_{k,p}) \,\mathrm{d}\mathbf{u}
=
\int_{[0,1]^k}
\underbrace{
   \frac{
      h_{k,p}^{(\nu)}\!\left(
         \tilde{\boldsymbol L}_{k,p}\,
         \Psi^{-1}(\mathbf{v}; \boldsymbol I_k);
         \mathbf{m}_{k,p}
      \right)
   }{
      \psi_{k,p}\!\left(
         \tilde{\boldsymbol L}_{k,p}\,
         \Psi^{-1}(\mathbf{v}; \boldsymbol I_k)
      \right)
   }
}_{:=\tilde h_{k,p}^{(\nu)}(\mathbf v;\mathbf m_{k,p})}
\,\mathrm{d}\mathbf v,
\label{eq:domain_transform_MN}
\end{equation}
For the NIG case, this yields the transform
\begin{equation}
\int_{\mathbb{R}^k} h_{k,p}^{(\nu)}(\mathbf u; \mathbf{m}_{k,p}) \,\mathrm{d}\mathbf{u}
=
\int_{[0,1]^{k+1}}
\underbrace{
   \frac{
      h_{k,p}^{(\nu)}\!\left(
         \sqrt{\Psi_W^{-1}(v_{k+1})}\,
         \tilde{\boldsymbol L}_{k,p}\,
         \Psi^{-1}(\mathbf{v}_{1:k};\boldsymbol I_k);
         \mathbf{m}_{k,p}
      \right)
   }{
      \psi_{k,p}^{\mathrm{lap}}\!\left(
         \sqrt{\Psi_W^{-1}(v_{k+1})}\,
         \tilde{\boldsymbol L}_{k,p}\,
         \Psi^{-1}(\mathbf{v}_{1:k};\boldsymbol I_k)
      \right)
   }
}_{:=\tilde h_{k,p}^{(\nu)}(\mathbf v;\mathbf m_{k,p})}
\,\mathrm{d}\mathbf v,
\label{eq:domain_transform_MNIG}
\end{equation}
where $\psi^{\mathrm{lap}}$ denotes the reference mixing density induced by the exponential variable $W$ and the associated change of variables, including the Jacobian factors arising from the scaling $u \mapsto \sqrt{W}\,u$ and from the inverse CDF $\Psi_W^{-1}$.

The resulting transformed integrands can be shown (Appendix~\ref{Appendix:boundary_oscillation}–\ref{Appendix:choice_of_trans_matrix}) to satisfy the polynomial-type boundary growth condition \eqref{eq:bound_growth_derivatives_RQMC}, with exponents $A^{(\nu)}_{j,p}$ depending on the reference parameters and the scaling $c$.

Figures~\ref{fig:diff_scale_NIG_1D} and~\ref{fig:diff_scale_Gauss_2D} illustrate the distribution-dependent boundary oscillation behavior analyzed in Appendix~\ref{Appendix:boundary_oscillation}.  In the Gaussian case, the proposed transformation suppresses boundary oscillations more effectively than in the NIG case, even though the Gaussian example is higher-dimensional (10D) than the NIG example (3D).  When $c=1$, the transformed component integrands exhibit kink-like oscillations near the boundary in both settings, with the effect being more pronounced under the NIG transformation. Increasing the scale parameter to $c>1$ improves the regularity of the transformed integrands, leading to smoother behavior and reduced boundary oscillations.
\enlargethispage{2\baselineskip} 
\begin{figure}
    \centering
    \begin{subfigure}[t]{0.45\textwidth}
        \centering
        \includegraphics[width=1\linewidth]{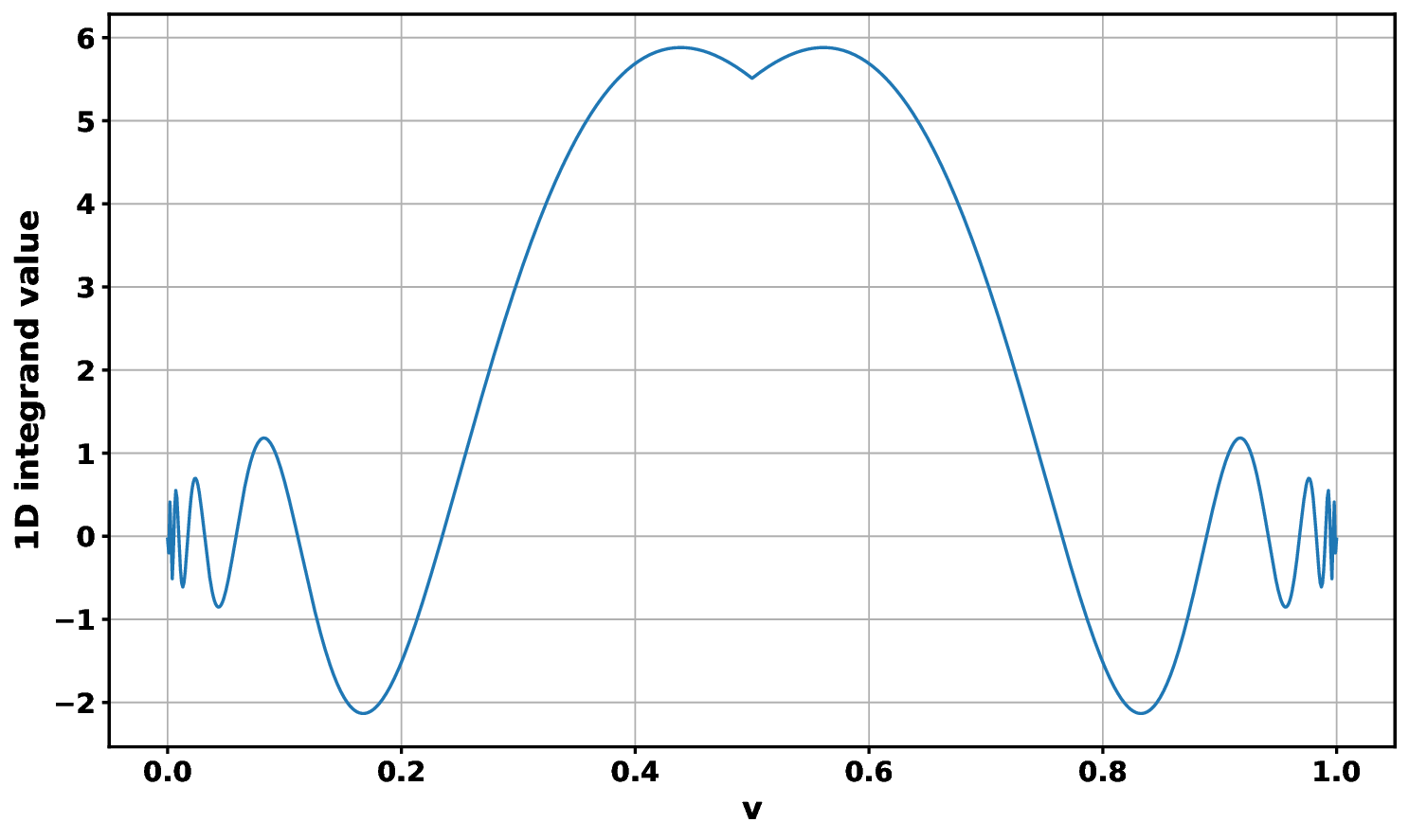}
        \caption{$c=1$}
        \label{fig:scale_1}
    \end{subfigure}
    \hfill
    \begin{subfigure}[t]{0.45\textwidth}
        \centering
        \includegraphics[width=1\linewidth]{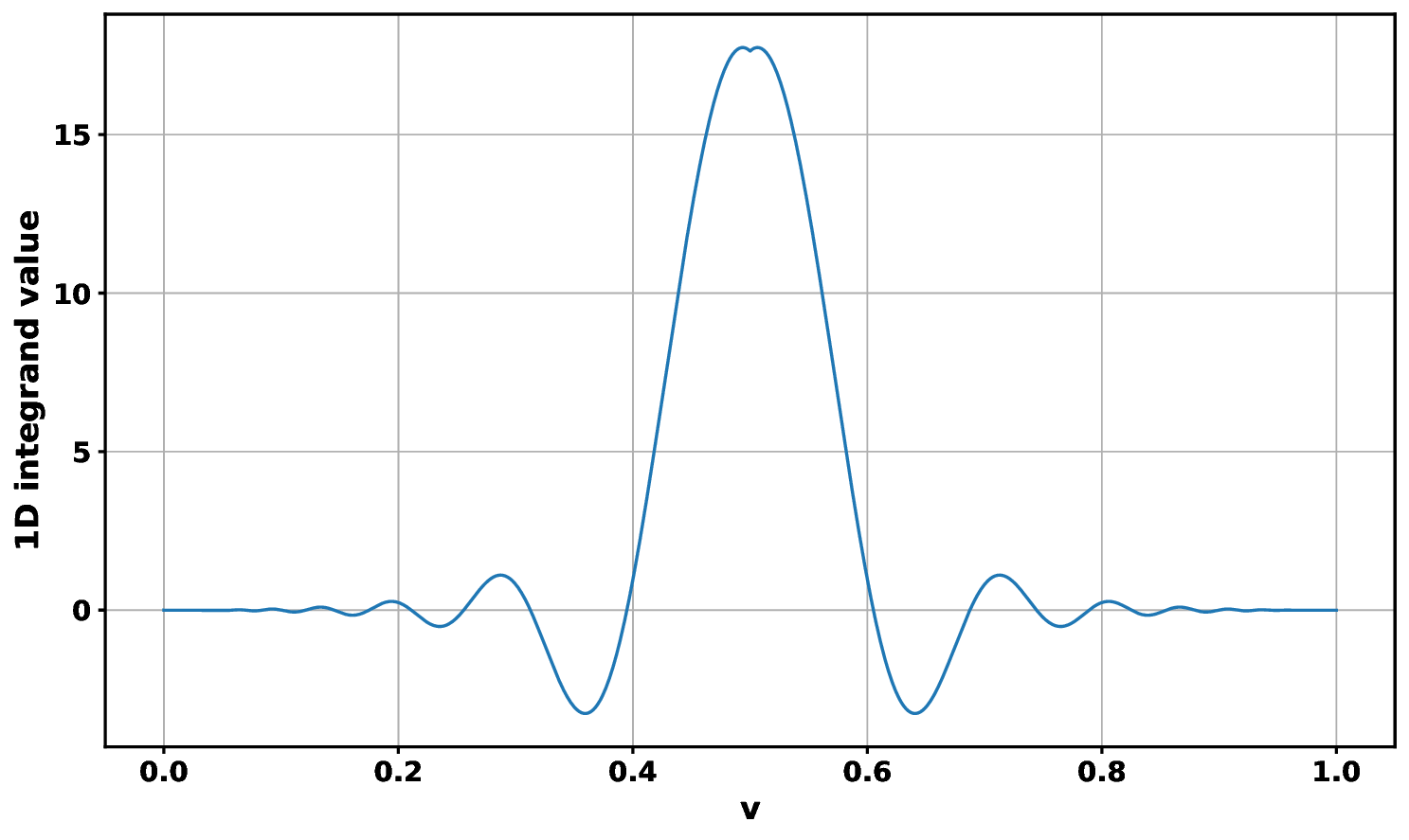}
         \caption{$c = 8$}
        \label{fig:scale_10}
    \end{subfigure}
   \caption{\small Transformed integrand component $\tilde h^{(0)}_{1,1}$ for the QPC loss and   a $3$-dimensional NIG loss vector, with $\mathbf K^{(0)}_{1,1}=4.5$ and $\mathbf m_{1,1}=-0.8$ (parameter setting in Section~\ref{subsec:3D_NIG_qcl}).}

    \label{fig:diff_scale_NIG_1D}
\end{figure}
\begin{figure}[H]
    \centering
    \begin{subfigure}[t]{0.48\textwidth}
        \centering
        \includegraphics[width=\linewidth]{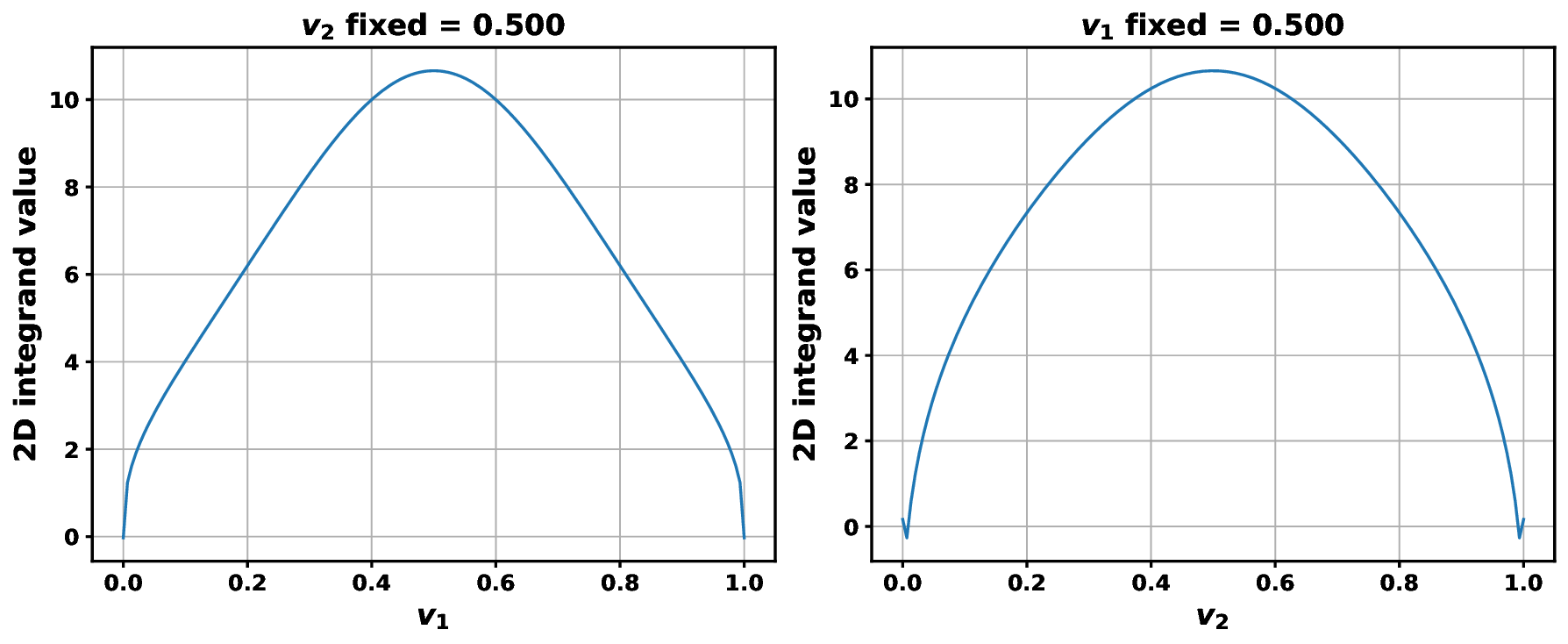}
        \caption{$c=1$}
        \label{fig:scale_1_Gauss}
    \end{subfigure}
    \hfill
    \begin{subfigure}[t]{0.48\textwidth}
        \centering
        \includegraphics[width=\linewidth]{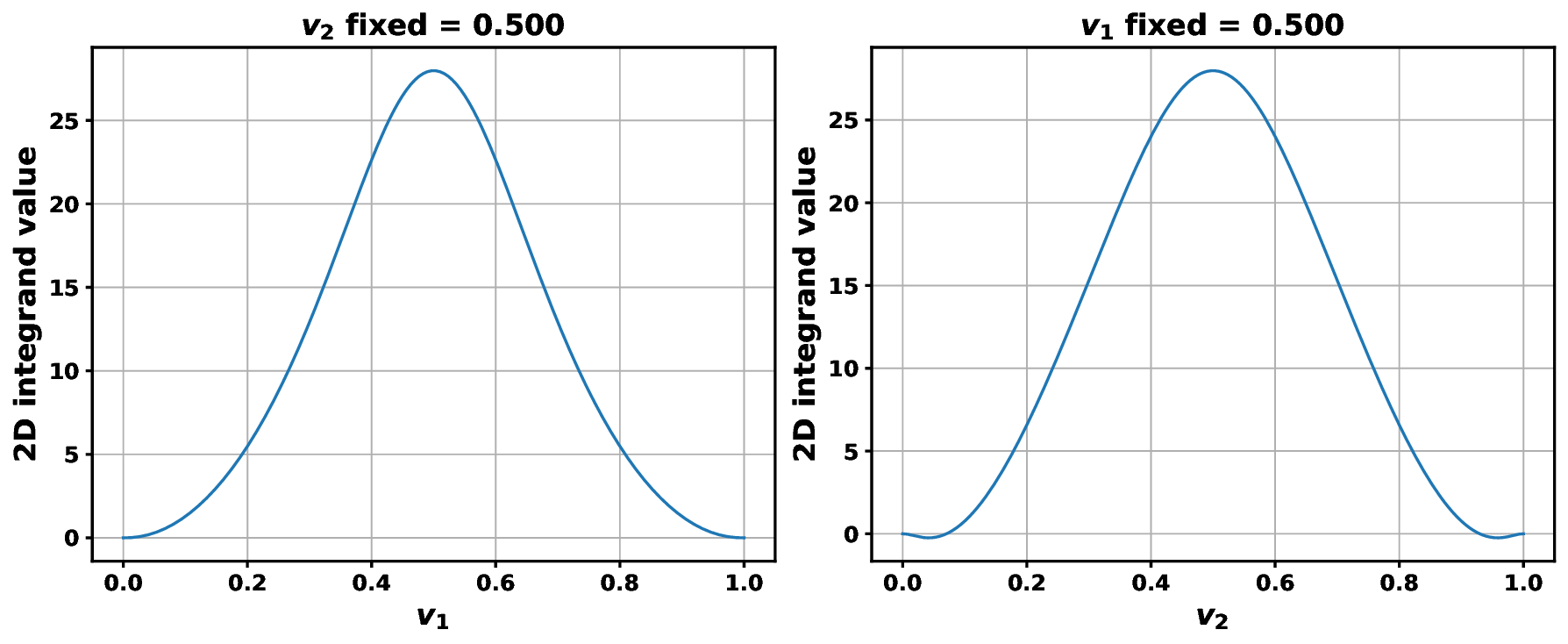}
         \caption{$c = 2.5$}
        \label{fig:scale_2.5_gauss}
    \end{subfigure}
    \vspace{-2mm}
    \caption{\small 1D slice of the transformed integrand component ${\tilde h_{2,(3,4)}^{(0)}}$ with $\mathbf{K}_{2,(3,4)}^{(0)} = \brac{2.763,0.523}, \mathbf{m}_{2,(3,4)} =\brac{0.255,0.105}$ for the QPC loss and a $10$-dimensional Gaussian loss vector (parameter setting in Section \ref{subsec:10D_gaussian_qcl}).}
    \label{fig:diff_scale_Gauss_2D}
\end{figure}
\begin{remark}
We choose the domain transformation to mitigate boundary-induced oscillations by taking a scaling parameter $c>1$, while avoiding excessive concentration of the reference density that would lead to poor numerical conditioning of the transformed integrands. In practice, the parameter $c$ trades off two competing effects: values close to $1$ may leave residual oscillations near the boundary, whereas excessively large values of $c$ lead to overly concentrated (peaked) integrands due to strong decay. In our numerical experiments, we select $c$ from a moderate range, specifically $c \in [4,10]$, which yields stable and robust performance. A principled, theoretically optimal choice of $c$ is left for future work.
\end{remark}
With this choice of transformation for the component integrands, we next develop the multilevel extension of the Fourier-RQMC framework.
\subsection{Iteration-Indexed Multilevel Fourier–RQMC Approximation}\label{subsec:Multilevel Fou-RQMC}

Recall from Section \ref{sec:Problem Setting and Background} that the primal–dual  variable of the MSRM optimization problem is  $\mathbf z := (\mathbf m,\lambda)$. At the optimization step $j$, we denote the iterate by $\mathbf z^{(j)} := (\mathbf m^{(j)},\lambda^{(j)})$. In Algorithm \ref{alg:SLSQP_full}, the constrained optimizer requires repeated evaluations of the Fourier-based objective, gradient, and Hessian at successive iterates $\mathbf z^{(j)}$. Directly recomputing these quantities via RQMC at every iteration can be computationally expensive. This section introduces an iteration-indexed multilevel Fourier–RQMC construction that exploits the strong correlation between consecutive iterates $\mathbf z^{(j-1)}$ and $\mathbf z^{(j)}$. By expressing gradient and Hessian evaluations in a block form and estimating differences across iterations, we obtain a multilevel estimator with reduced variance and improved efficiency. For notational convenience in this section and in Section \ref{sec:num_analysis}, we introduce the following block-valued mappings. 
\begin{notation}\label{not:fou_rqmc_KKT}
For $\nu=0,1,2$, define the block-valued mappings $\mathcal{H}^{(\nu)}(\cdot;\mathbf z)$ by
\[
\mathcal{H}^{(0)}(\cdot;\mathbf z)
:= \sum_{k=1}^d m_k + \lambda\,\tilde h^{(0)}(\cdot;\mathbf{m}),
\qquad
\mathcal{H}^{(1)}(\cdot;\mathbf z)
:=
\begin{bmatrix}
\mathbf 1 - \lambda\,\tilde h^{(1)}(\cdot;\mathbf{m})\\[3pt]
-\,\tilde h^{(0)}(\cdot;\mathbf{m})
\end{bmatrix},
\]
and
\[
\mathcal{H}^{(2)}(\cdot;\mathbf z)
:=
\begin{bmatrix}
-\lambda\,\tilde h^{(2)}(\cdot;\mathbf{m})
&
\big(\tilde h^{(1)}(\cdot;\mathbf{m})\big)^{\!\top}\\[3pt]
\tilde h^{(1)}(\cdot;\mathbf{m}) & 0
\end{bmatrix}.
\]
Then the Fourier-RQMC approximations of
$\hat{\mathcal{L}}^{\mathrm{Fou}}(\mathbf z)$,
$\hat{\mathcal{L}}^{\mathrm{Fou}}_{\nabla_{\mathbf z}}(\mathbf z)$, and
$\hat{\mathcal{L}}^{\mathrm{Fou}}_{\nabla_{\mathbf z}^2}(\mathbf z)$
are denoted by
$I_{N,S_{\mathrm{shift}}}^{\mathrm{RQMC}}\!\big[\mathcal{H}^{(0)}(\cdot;\mathbf z)\big]$,
$I_{N,S_{\mathrm{shift}}}^{\mathrm{RQMC}}\!\big[\mathcal{H}^{(1)}(\cdot;\mathbf z)\big]$, and
$I_{N,S_{\mathrm{shift}}}^{\mathrm{RQMC}}\!\big[\mathcal{H}^{(2)}(\cdot;\mathbf z)\big]$, respectively.
Here $I_{N,S_{\mathrm{shift}}}^{\mathrm{RQMC}}[\cdot]$ is applied componentwise to vector-valued
$\tilde h^{(1)}(\cdot;\mathbf{m})$ and entrywise to matrix-valued $\tilde h^{(2)}(\cdot;\mathbf{m})$.
\label{nota:Fou_RQMC_KKT_represent}
\end{notation}
The mappings $\mathcal H^{(0)}$, $\mathcal H^{(1)}$, and $\mathcal H^{(2)}$ in Notation \ref{not:fou_rqmc_KKT}  correspond to the Fourier-based representations of the objective, first-order optimality conditions, and second-order conditions w.r.t. the primal–dual variable $\mathbf z$.

Within the single-level Fourier–RQMC framework, Algorithm \ref{alg:Single-level RQMC} evaluates at each optimization iteration $j\in\{1,\ldots,J\}$ the Fourier-based gradient $\widehat{\mathcal L}^{\mathrm{Fou}}_{\nabla_ \mathbf{z}}\paren{\mathbf{z}^{(j)}}$. This quantity is approximated using a RQMC estimator $I^{\mathrm{RQMC}}_{N,S_{\mathrm{shift}}}$, constructed from a fixed set of $N$ Sobol points and $S_{\mathrm{shift}}$ digital shifts. In particular, the component $\hat g^{(v),\mathrm{Fou}}(\mathbf{m}^{(j)})$ is approximated by
\begin{equation}
\begin{aligned}
I^{\mathrm{RQMC}}_{N,S_{\mathrm{shift}}}\!\left[\tilde h^{(\nu)}\paren{\,\cdot\,; \mathbf{m}^{(j)}}\right]
=
\sum_{k\in \mathcal I_{q_{\ell}}} \sum_{\mathbf{p}\in\mathcal I_k}
I^{\mathrm{RQMC}}_{N,S_{\mathrm{shift}}}\!\left[\tilde h^{(\nu)}_{k,p}\paren{\,\cdot\,; \mathbf{m}^{(j)}_{k,p}}\right].
\end{aligned}
\label{eq:single_RQMC_est}
\end{equation}
Evaluating  \eqref{eq:single_RQMC_est} at every optimization iteration can be computationally demanding, particularly when the integrand dimension $k$ is large, when many component integrands $\tilde h^{(v)}_{k,p}$ must be evaluated, or when the number of optimization iterations is itself substantial. From a numerical optimization perspective, however, the successive iterates $\mathbf{z}^{(j-1)}$ and $\mathbf{z}^{(j)}$ are typically strongly correlated, since the optimization algorithm evolves gradually toward a solution. This observation motivates the use of a control-variate strategy, whereby the estimator at iteration $\mathbf{z}^{(j-1)}$ is exploited to reduce the variance of the estimator at the current iteration $\mathbf{z}^{(j)}$.

To formalize this idea, we express the Fourier-based gradient at iteration j as the telescoping sum
\begin{equation*}
\begin{aligned}
\hat{\mathcal{L}}_{\nabla_{\mathbf{z}}}^{\mathrm{Fou}}\paren{\mathbf{z}^{(j)}} &= \hat{\mathcal{L}}_{\nabla_{\mathbf{z}}}^{\mathrm{Fou}}\paren{\mathbf{z}^{(1)}} + \sum_{j=2}^{J} \left[\hat{\mathcal{L}}_{\nabla_{\mathbf{z}}}^{\mathrm{Fou}}\paren{\mathbf{z}^{(j)}} - \hat{\mathcal{L}}_{\nabla_{\mathbf{z}}}^{\mathrm{Fou}}\paren{\mathbf{z}^{(j-1)}}\right] \\
    & = \hat{\mathcal{L}}_{\nabla_{\mathbf{z}}}^{\mathrm{Fou}}\paren{\mathbf{z}^{(1)}} + \sum_{j=2}^{J} \Delta \hat{\mathcal{L}}_{\nabla_{\mathbf{z}}}^{\mathrm{Fou}}\paren{\mathbf{z}^{(j)},\mathbf{z}^{(j-1)}}
\end{aligned}  
\end{equation*}
where $\Delta \hat{\mathcal{L}}_{\nabla_{\mathbf{z}}}^{\mathrm{Fou}}\paren{\mathbf{z}^{(j)},\mathbf{z}^{(j-1)}} := \hat{\mathcal{L}}_{\nabla_{\mathbf{z}}}^{\mathrm{Fou}}\paren{\mathbf{z}^{(j)}} - \hat{\mathcal{L}}_{\nabla_{\mathbf{z}}}^{\mathrm{Fou}}\paren{\mathbf{z}^{(j-1)}}$. The increments in this decomposition capture differences between consecutive optimization iterates and typically exhibit smaller variance and improved regularity compared to the original estimator. This structure enables an iteration-indexed multilevel RQMC construction, conceptually related to Multilevel Monte Carlo (MLMC) methods \cite{giles_multilevel_2015,bayer_multilevel_2024}, with the crucial distinction that here the “levels” correspond to optimization iterations rather than to discretization levels of a stochastic differential equation.

We first estimate \( \hat{\mathcal{L}}_{\nabla_{\mathbf{z}}}^{\mathrm{Fou}}\paren{\mathbf{z}^{(1)}} \) by its RQMC approximation 
\({I}_{N,S_{\mathrm{shift}}}^{\mathrm{RQMC}} \brac{\mathcal{H}^{(1)}\paren{\cdot;\mathbf{z}_{1}}}\) using \eqref{eq:single_RQMC_est} 
with \(N\) Sobol points and \(S_{\mathrm{shift}}\) digital shifts. 
At each subsequent iteration \( j \in \{2,\dots,J\} \), we evaluate the incremental difference $\Delta \hat{\mathcal{L}}_{\nabla_{\mathbf{z}}}^{\mathrm{Fou}}\paren{\mathbf{z}^{(j)},\mathbf{z}^{(j-1)}}$ by ${I}_{N_j,S_{\mathrm{shift}}}^{\mathrm{RQMC}} \brac{\Delta \mathcal{H}^{(1)}\paren{\cdot;\mathbf{z}^{(j)},\mathbf{z}^{(j-1)}} }$ using \(N_j\) Sobol points and the same number \(S_{\mathrm{shift}}\) of randomizations, but with \emph{independent} digital shifts (i.e., a fresh set of RQMC randomizations for each iteration),
whose components are given by 
\({I}_{N_j,S_{\mathrm{shift}}}^{\mathrm{RQMC}} \brac{\Delta \tilde{h}^{(\nu)}}\) 
\begin{equation*}
\begin{aligned}
{I}_{N_j,S_{\mathrm{shift}}}^{\mathrm{RQMC}} \brac{\Delta \tilde h^{(\nu)}\paren{\cdot; \mathbf{m}^{(j)}, \mathbf{m}^{(j-1)}}}
&:=
\sum_{k\in \mathcal{I}_{q_{\ell}}}\sum_{\mathbf{p} \in \mathcal{I}_k}
I_{N_j,S_{\mathrm{shift}}}^{\mathrm{RQMC}}
\!\left[
\Delta \tilde h_{k,p}^{(\nu)}\!\paren{\cdot;\mathbf{m}^{(j)}_{k,p},\mathbf{m}^{(j-1)}_{k,p}}
\right],
\end{aligned}
\end{equation*}
where 
\begin{equation*}
\begin{aligned}
    {I}_{N_j,S_{\mathrm{shift}}}^{\mathrm{RQMC}} \brac{\Delta \tilde h_{k,p}^{(\nu),{\mathrm{Fou}}}\paren{\cdot;\mathbf{m}^{(j)}, \mathbf{m}^{(j-1)}}}&:= \frac{1}{S_{\text{shift}}} \sum_{s=1}^{S_{\text{shift}}} \frac{1}{N_j} 
    \sum_{n=1}^{N_j} 
    \Delta \tilde h_{k,p}^{(\nu)}\!\paren{\mathbf{v}_n^{(s,j)};\mathbf{m}^{(j)}_{k,p},\mathbf{m}^{(j-1)}_{k,p}},
\end{aligned}
\end{equation*}
and the differences of the transformed integrands are defined as
\begin{equation*}
\begin{aligned}
\Delta \tilde h_{k,p}^{(\nu)}\!\paren{\mathbf{v};\mathbf{m}_{k,p}^{(j)},\mathbf{m}_{k,p}^{(j-1)}} 
&:= 
\tilde h_{k,p}^{(\nu)}\!\left(\mathbf{v};\mathbf{m}_{k,p}^{(j)}\right)
-
\tilde h_{k,p}^{(\nu)}\!\left(\mathbf{v};\mathbf{m}_{k,p}^{(j-1)}\right).
\end{aligned}
\end{equation*}
The advantage of the multilevel method is that the difference terms $\Delta \tilde h_{k,p}^{(\nu)}\!\left(\mathbf{v};\mathbf{m}^{(j)}_{k,p},\mathbf{m}^{(j-1)}_{k,p}\right)$
often have better regularity and smaller variability than the original integrands, 
$\tilde h_{k,p}^{(\nu)}\!\left(\mathbf{v};\mathbf{m}^{(j)}_{k,p}\right)$(see Figure \ref{fig:contraction_1D_gauss} for illustration), and then we can choose the number of Sobol points $N_j$ in a level-dependent manner through the optimization process, which can reduce our computational time. This will be discussed in more detail in Section \ref{subsubsec:Multilevel Fou-RQMC comp complexity}.
\begin{figure}
    \centering
    \includegraphics[width=0.8\linewidth]{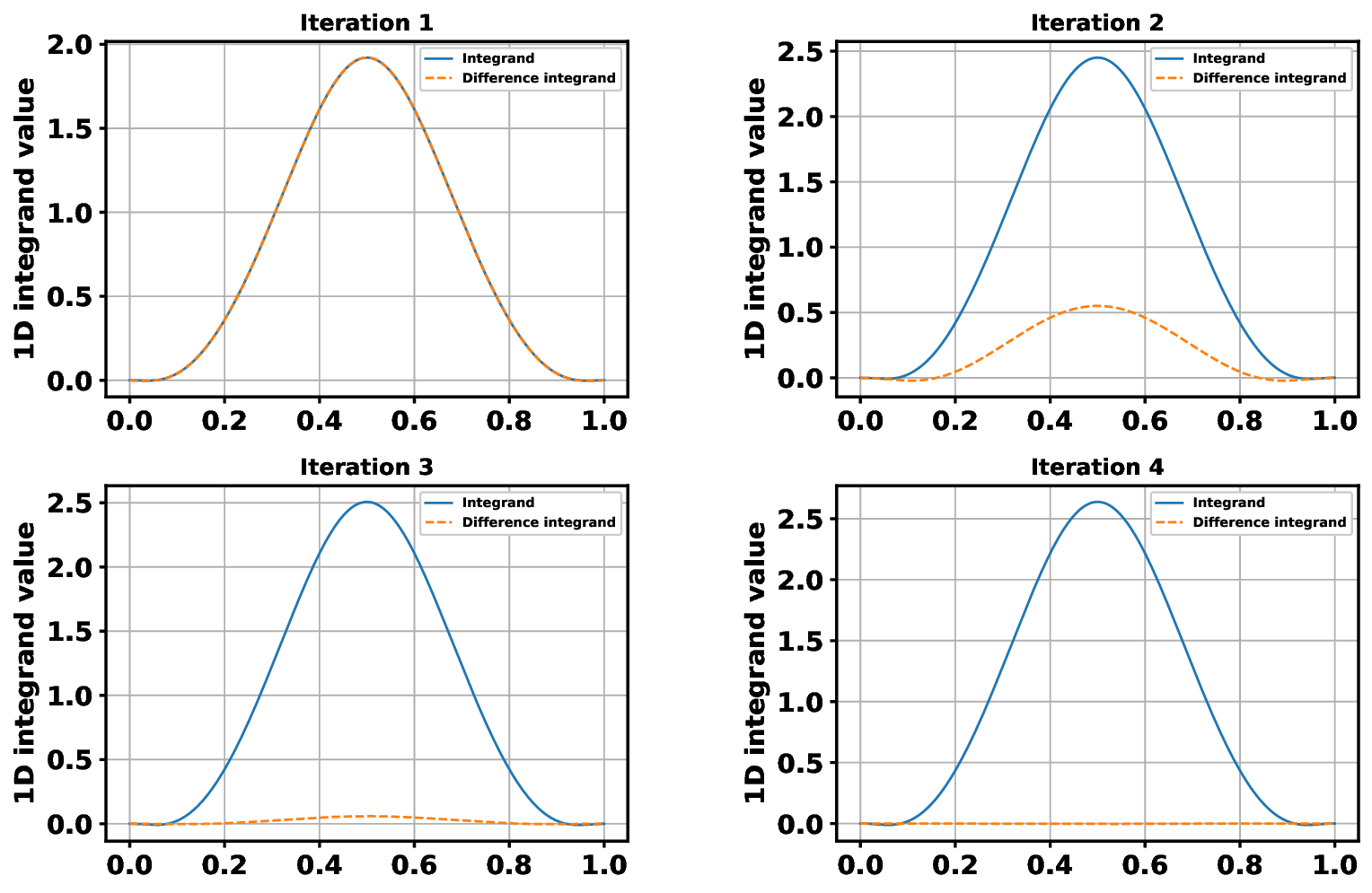}
    \caption{\small {Transformed integrand component $\tilde h_{1,1}^{(0)}$ (solid) and the corresponding difference integrand arising in the multilevel construction (dashed) across successive optimization iterations, for the QPC loss and 10D Gaussian loss vector (parameter setting in Section \ref{subsec:10D_gaussian_qcl}).}}
    \label{fig:contraction_1D_gauss}
\end{figure}
\begin{remark}[Damping for difference integrands]
Instead of determining two separate damping vectors
$\mathbf K^{(\nu,j-1)}_{k,p}$ and $\mathbf K^{(\nu,j)}_{k,p}$ for the integrands
$\tilde h^{(\nu)}_{k,p}\paren{\mathbf{v}; \mathbf m^{(j-1)}_{k,p}}$ and
$\tilde h^{(\nu)}_{k,p}\paren{\mathbf{v}; \mathbf m^{(j)}_{k,p}}$,
we apply Algorithm~\ref{alg:Fou_opti_damp} directly to the difference integrand
\[
\Delta \tilde h^{(\nu)}_{k,p}\paren{\mathbf{v}; \mathbf m^{(j)}_{k,p}, \mathbf m^{(j-1)}_{k,p}}
=
\tilde h^{(\nu)}_{k,p}\paren{\mathbf{v}; \mathbf m^{(j)}_{k,p}}
-
\tilde h^{(\nu)}_{k,p}\paren{\mathbf{v}; \mathbf m^{(j-1)}_{k,p}},
\]
and compute a single damping vector $\mathbf K^{(\nu,j)}_{k,p}$ for this term.

This choice is justified because the difference integrand typically inherits an admissible analyticity strip given by the intersection of the strips of the two terms and, in practice, often exhibits milder oscillations and boundary growth due to cancellation. In particular, cancellations between consecutive optimization iterates reduce oscillatory behavior and boundary growth. As a result, the difference integrands exhibit improved regularity and smaller variability, which is advantageous both for contour selection and for variance reduction in the multilevel estimator.
\end{remark}
A concise overview of the multilevel Fourier–RQMC algorithm is provided in Algorithm \ref{alg:RQMC_Fou_multi}.
\begin{algorithm}[H]
\small
\caption{Multilevel Fourier-RQMC at optimization step $j$}
\label{alg:RQMC_Fou_multi}
\begin{algorithmic}[1]
\Require
Allocation $\mathbf m^{(j)}$; baseline Sobol size $N_1 = N\in\mathbb N$; level sizes $\{N_l\}_{l=2}^j$;
number of shifts $S_{\mathrm{shift}}$; for  $k \in \mathcal{I}_{q_{\ell}}$; index set $\mathcal I_k$, for $\mathbf{p} \in \mathcal{I}_k$: component \emph{difference} integrands
$\Delta h_{k,p}^{(\nu,l)}\paren{\mathbf{v}; \mathbf{m}_{k,p}^{(l)},\mathbf{m}_{k,p}^{(l-1)}}$ for levels $l=2,\dots,j$ with corresponding optimal damping vectors $\mathbf{K}_{k,p}^{(\nu,l)}$; RQMC estimates at \emph{level-1} 
${I}_{N_1,S_{\mathrm{shift}}}^{\mathrm{RQMC}} \brac{\tilde h^{(\nu)}\paren{\cdot;\mathbf{m}^{(1)}}}$ from Algorithm~\ref{alg:Single-level RQMC}.

\State For each $(k,p)$, apply the same transform \eqref{eq:domain_transform_MN} or \eqref{eq:domain_transform_MNIG} used at the level $1$ to obtain $\Delta \tilde h_{k,p}^{(\nu,l)}\paren{\mathbf{v}; \mathbf{m}_{k,p}^{(l)},\mathbf{m}_{k,p}^{(l-1)}}$.
\State \textbf{(Coupled base nets across levels)} For $k \in \mathcal{I}_{q_{\ell}}$:
\begin{itemize}
  \item Generate a single \emph{unshifted} base-2 digital net
        $\{\mathbf u^{(s)}_n\}_{n=1}^{N_1}\subset[0,1]^{k}$. Reuse this same \emph{unshifted} net for \emph{all} levels $l=1,\dots,j$ and all $\mathbf{p}\in\mathcal I_k$.
\end{itemize}
\State Initialize the multi-level estimators from level $2$ to $j$, ${I}_{N_{2:l},S_{\mathrm{shift}}}^{\mathrm{RQMC}} \brac{\Delta \tilde h^{(\nu)}} \gets 0$ 
\For{$l=2,\dots,j$} 
  \State Draw an \emph{independent} digital shift $\boldsymbol\Delta^{(s,l)}\in[0,1]^{k}$ (with a new seed, independent across $k$ and $l$).
  \State Form the shifted points $\mathbf v^{(s,l)}_n=\mathbf u^{(s)}_n\oplus\boldsymbol\Delta^{(s,l)}$, $n=1,\dots,N_l$.
  \State Initialize the level-$l$ increments: 
$
{I}_{N_l,S_{\mathrm{shift}}}^{\mathrm{RQMC}} \brac{\Delta \tilde h^{(\nu)}\paren{\cdot;\mathbf{m}^{(l)}, \mathbf{m}^{(l-1)}}} \gets 0.  
  $
  \For{$k \in \mathcal{I}_{q_{\ell}}$}
    \For{$\mathbf{p}\in\mathcal I_k$}
      \State Compute
      \(
I^{\mathrm{RQMC}}_{N_l,S_{\mathrm{shift}}}\!\left[\Delta \tilde h_{k,p}^{(\nu)}\paren{\cdot;\mathbf{m}_{k,p}^{(l)},\mathbf{m}_{k,p}^{(l-1)}}\right]
      \)
      using $\{\mathbf v^{(s,l)}_n\}_{n=1}^{N_l}$.
      \State Accumulate
      $
        {I}_{N_l,S_{\mathrm{shift}}}^{\mathrm{RQMC}} \brac{\Delta \tilde h^{(\nu)}\paren{\cdot;\mathbf{m}^{(l)}, \mathbf{m}^{(l-1)}}}
      \mathrel{+}= I^{\mathrm{RQMC}}_{N_l,S_{\mathrm{shift}}}\!\left[\Delta \tilde h_{k,p}^{(\nu)}\paren{\cdot;\mathbf{m}_{k,p}^{(l)},\mathbf{m}_{k,p}^{(l-1)}}\right].
      $
    \EndFor
  \EndFor
  \State Update the multi-level estimators:
  \(
    {I}_{N_{2:l},S_{\mathrm{shift}}}^{\mathrm{RQMC}} \brac{\Delta \tilde h^{(\nu)}} \mathrel{+}=  {I}_{N_j,S_{\mathrm{shift}}}^{\mathrm{RQMC}} \brac{\Delta \tilde h^{(\nu)}\paren{\cdot;\mathbf{m}^{(l)}, \mathbf{m}^{(l-1)}}}.
  \)
\EndFor
\State \textbf{Output at step $j$:}
\(
   {I}_{N_j,S_{\mathrm{shift}}}^{\mathrm{RQMC}} \brac{\tilde h^{(\nu)}\paren{\cdot;\mathbf{m}^{(j)}}}
  =
  {I}_{N_1,S_{\mathrm{shift}}}^{\mathrm{RQMC}} \brac{\tilde h^{(\nu)}\paren{\cdot;\mathbf{m}^{(1)}}}
  +{I}_{N_{2:l},S_{\mathrm{shift}}}^{\mathrm{RQMC}} \brac{\Delta \tilde h^{(\nu)}}.  
\)
\end{algorithmic}
\end{algorithm}
\enlargethispage{2\baselineskip} 

\section{Error and Complexity Analysis for Fourier–RQMC Methods}\label{sec:num_analysis}
This section analyzes the error and computational complexity of the proposed single-level and multilevel Fourier–RQMC schemes. We decompose the total numerical error into (i) an optimization error due to the SQP solver (Section~\ref{subsec:opt_err}) and (ii) a quadrature-induced surrogate error due to Fourier–RQMC approximation of the expectation terms (Section~\ref{subsec:stat_err_analysis}). Throughout, we assume the KKT system \eqref{eq:KKT_system}  admits a unique solution $\mathbf z^\ast=(\mathbf m^\ast,\lambda^\ast)$, ensuring the target allocation is well defined. A convergence analysis of the SAA method for the MSRM problem is presented in Appendix~\ref{appendix:converge_analysis_SAA}, providing a benchmark for comparison with our  Fourier–RQMC surrogates. We now introduce additional notation that will be used in Sections~\ref{subsec:opt_err}--\ref{subsec:computational_complex_Fou_RQMC}.
\begin{notation}\
\begin{itemize}
\item Let $M \subset \mathbb{R}^d$ be a nonempty compact set. For each component $(k,p)$, define
$M_{k,p} := P_{k,p} M \subset \mathbb{R}^k$, with $P_{k,p}$ is defined in Notation \ref{not:component_integrands}.
Moreover, for a given $\bar\lambda > 0$, set
$\mathcal{Z} := M \times (0,\bar\lambda] \subset \mathbb{R}^{d+1}$.
   \item  $\mathbf{z}^{*} = \paren{\mathbf{m^*},\lambda^*}$ denotes the exact unique solution obtained using the 
    true Lagrangian  
    $\mathcal{L}(\mathbf{z})\paren{\hat{\mathcal{L}}^{\mathrm{Fou}}(\mathbf{z})}$.
    \item 
${\mathbf{z}}_{N,S_{\mathrm{shift}}}^{\mathrm{RQMC},*} := \paren{{\mathbf{m}}_{N,S_{\mathrm{shift}}}^{\mathrm{RQMC},*},{\lambda}_{N,S_{\mathrm{shift}}}^{\mathrm{RQMC},*}}$ denotes the solution obtained by 
    replacing the true Lagrangian 
    with its  Fourier-RQMC approximation $I_{N,S_{\mathrm{shift}}}^{\mathrm{RQMC}}
        \!\left( 
            \mathcal{H}^{(0)}(\cdot;\mathbf z)
        \right).$
    \item ${\mathbf{z}}_{N,S_{\mathrm{shift}}}^{(\mathrm{RQMC},j)}:= \paren{{\mathbf{m}}_{N,S_{\mathrm{shift}}}^{(\mathrm{RQMC},j)}, {\lambda}_{N,S_{\mathrm{shift}}}^{(\mathrm{RQMC},j)}} $ be the $j$-th iterate solution, which is returned by the numerical optimization solver (SLSQP) applied to the Fourier-RQMC problem $I_{N,S_{\mathrm{shift}}}^{\mathrm{RQMC}}
        \!\left( 
            \mathcal{H}^{(0)}(\cdot;\mathbf z)
        \right).$
\end{itemize}
\label{nota:decompose_error}
\end{notation}
By the triangle inequality, we obtain the following decomposition of the total error
\begin{equation}
\bigl\|{\mathbf{z}}_{N,S_{\mathrm{shift}}}^{(\mathrm{RQMC},j)}-\mathbf z^*\bigr\|
\;\le\;
\underbrace{\bigl\| {\mathbf{z}}_{N,S_{\mathrm{shift}}}^{(\mathrm{RQMC},j)}-{\mathbf{z}}_{N,S_{\mathrm{shift}}}^{\mathrm{RQMC},*}\bigr\|}_{\varepsilon_{\mathrm{opt}}(j)}
\;+\;
\underbrace{\bigl\|{\mathbf{z}}_{N,S_{\mathrm{shift}}}^{\mathrm{RQMC},*}-\mathbf z^*\bigr\|}_{\varepsilon_{\mathrm{stat}}^{\mathrm{RQMC}}(N)}.
\label{eq:terminal_error_decompose}
\end{equation}
where
\begin{itemize}
\item ${\varepsilon_{\mathrm{opt}}(j)}$ is optimization error, set by the deterministic solver (SLSQP) on the Fourier-RQMC surrogate.
\item ${\varepsilon_{\mathrm{stat}}^{\mathrm{RQMC}}(N)}$ is a statistical error when approximating the exact system with the Fourier-RQMC surrogate.
\end{itemize}
\begin{remark}
    We emphasize that the optimization error $\varepsilon_{\mathrm{opt}}(j)$ is governed solely by the convergence properties of the SQP algorithm applied to the  Fourier–RQMC surrogate problem and is therefore insensitive to whether a single-level or multilevel Fourier–RQMC estimator is employed. In contrast, the statistical error $\varepsilon^{\mathrm{RQMC}}_{\mathrm{stat}}(N)$ depends explicitly on the structure of the underlying Fourier–RQMC estimator, and it is at this level that the distinction between single-level and multilevel constructions becomes essential for variance reduction and complexity improvements.
\label{rema:separate_opt_stat_err}
\end{remark}
We now analyze these two error contributions separately, starting with the optimization error $\varepsilon_{\mathrm{opt}}(j)$.
\subsection{Optimization error}\label{subsec:opt_err}
We perform numerical optimization based on the Fourier–RQMC surrogates. To derive convergence rates for the associated numerical optimization scheme implemented via SQP,  we introduce in Appendix \ref{appendix:additional_opt_err}  a set of regularity assumptions, adapted from \cite{boggs_sequential_1995} and tailored to the Fourier–RQMC surrogate optimization problem, and obtain the following result from \cite[Theorem 3.4]{boggs_sequential_1995}.
\begin{theorem}
Suppose that Assumption \ref{assump:regularity_Fou_RQMC} holds, and let the sequence
$\{{\mathbf{z}}_{N,S_{\mathrm{shift}}}^{(\mathrm{RQMC},j)}\}_{j\ge 0}$ be generated by the SQP algorithm.
Assume further that, for almost every realization of the RQMC shifts, the SQP iterates satisfy ${\mathbf{z}}_{N,S_{\mathrm{shift}}}^{(\mathrm{RQMC},j)} \to {\mathbf{z}}_{N,S_{\mathrm{shift}}}^{\mathrm{RQMC},*}$ as $j\to\infty$. Then the convergence is superlinear: there exists a sequence of constants
$\{\eta_j\}_{j\ge 0}$ with $\eta_j>0$ and $\eta_j \to 0$ such that
\begin{equation}
\norm{\mathbf{z}_{N,S_\mathrm{shift}}^{(\mathrm{RQMC},j+1)}-{\mathbf{z}}_{N,S_{\mathrm{shift}}}^{\mathrm{RQMC},*}}
\leq
\eta_j \,
\norm{{\mathbf{z}}_{N,S_{\mathrm{shift}}}^{(\mathrm{RQMC},j)}-{\mathbf{z}}_{N,S_{\mathrm{shift}}}^{\mathrm{RQMC},*}}.
\end{equation}
\label{theorem:superlinear_converge_SQP}
\end{theorem}
\begin{proof}
    The detailed proof of Theorem \ref{theorem:superlinear_converge_SQP} is in \cite{boggs_sequential_1995}.
\end{proof}
\begin{remark}
Under the conditions of Theorem~\ref{theorem:superlinear_converge_SQP}, there exists an iteration index
$J_{\mathrm{loc}}$ such that, for all $j \ge J_{\mathrm{loc}}$, ${\mathbf{z}}_{N,S_{\mathrm{shift}}}^{(\mathrm{RQMC},j)}$ exhibit the corresponding superlinear contraction. We refer to this regime as the \emph{local convergence stage}.
\label{rema:local_convergence_stage}
\end{remark}
Using Theorem~\ref{theorem:superlinear_converge_SQP} together with
Remark~\ref{rema:local_convergence_stage}, the optimization error, measured w.r.t. the Fourier-RQMC surrogate solution, $\mathbf{z}_{N,S_{\mathrm{shift}}}^{\mathrm{RQMC},*}$, 
$\varepsilon_{\mathrm{opt}}(j)$ admits the superlinear bound
\begin{equation}
    \varepsilon_{\mathrm{opt}}(j)
    = \mathcal{O}\paren{\|e_{J_{\mathrm{loc}}}\|^{\,p^{\,j-J_{\mathrm{loc}}}}} ,
    \quad 1<p<2,
    \qquad j \ge J_{\mathrm{loc}},
    \label{eq:opti_err_SAA_slinear}
\end{equation}
where
\(
\|e_{J_{\mathrm{loc}}}\|
:=
\norm{{\mathbf{z}}_{N,S_{\mathrm{shift}}}^{(\mathrm{RQMC},J_{\mathrm{loc}})}-{\mathbf{z}}_{N,S_{\mathrm{shift}}}^{\mathrm{RQMC},*}}.
\)

We next analyze the statistical error $\varepsilon_{\mathrm{stat}}^{\mathrm{RQMC}}(N)$ induced by the Fourier-RQMC approximation of the expectation terms.

\subsection{Statistical Error and Asymptotic Analysis}\label{subsec:stat_err_analysis}
In order to bound the statistical error of the Fourier-RQMC solution, we first establish a uniform strong law of large numbers (USLLN) for the estimators $ I_{N,\bar S}^{\mathrm{RQMC}}\!\left[\tilde h_{k,p}^{(\nu)}(\cdot;\mathbf m_{k,p})\right]$, $\nu\in\{0,1,2\}$,  $k \in \mathcal{I}_{q_{\ell}}$, and $\mathbf{p} \in \mathcal{I}_k$, in the regime $N\to\infty$ for fixed $S_{\mathrm{shift}}=\bar{S}$.

As noted in \cite{owen_strong_2021}, almost sure convergence in $N$ need not hold for arbitrary randomizations of low-discrepancy nets. Therefore, throughout the remainder of this section, we work with Sobol point sets randomized via \emph{nested uniform scrambling} \cite{owen_randomly_1995}. For numerical experiments (Section ~\ref{subsec:sing-RQMC}), we instead employ digitally shifted Sobol sequences.
Moreover, in what follows, we work under Assumptions~\ref{ass:A7} for $\hat g_{k,p}^{(\nu),\mathrm{Fou}}$ and~\ref{ass:A10} for $\tilde h_{k,p}^{(\nu)}$, with $\nu\in\{0,1,2\}$. 

The following lemma provides uniform convergence of the Fourier–RQMC estimators over the decision set; this is the key input for consistency of the surrogate solution.
\begin{lemma}[Uniform convergence of Fourier-RQMC estimators]
\label{lem:uniform_conv_RQMC_Fou}
Let $\{\mathbf v_n\}_{n=1}^N$ be a Sobol sequence in $[0,1]^k$ with a uniformly bounded gain coefficient.\footnote{See \cite[Theorem~1]{owen_scrambling_1998}.}
Fix $S_{\mathrm{shift}}=\bar S$, and for each $s=1,\dots,S_{\mathrm{shift}}$ let
$\{\mathbf v_n^{(s)}\}_{n=1}^N$ be obtained by applying {nested uniform scrambling} to $\{\mathbf v_n\}_{n=1}^N$ as in \cite{owen_randomly_1995}.
Then, for each $\nu\in\{0,1,2\}$,
the RQMC estimator $I_{N,\bar S}^{\mathrm{RQMC}}\!\left[\tilde h_{k,p}^{(\nu)}(\cdot;\mathbf m_{k,p})\right]$
satisfies a USLLN on $M_{k,p}$,
\begin{equation}
\sup_{\mathbf m_{k,p}\in M_{k,p}}
\left\|
I_{N,\bar S}^{\mathrm{RQMC}}\!\left[\tilde h_{k,p}^{(\nu)}(\cdot;\mathbf m_{k,p})\right]
-\hat g_{k,p}^{(\nu),\mathrm{Fou}}(\mathbf m_{k,p})
\right\|
\xrightarrow[\;N\to\infty\;]{\mathrm{a.s.}} 0.
\label{eq:uniform_SLLN_rqmc_fou}
\end{equation}
\end{lemma}

\begin{proof}
    The detailed proof is presented in Appendix \ref{Appendix:proof_uniform_conv_RQMC}
\end{proof}
To pass from uniform convergence of the estimators to convergence of the corresponding solution, a stability condition on the underlying F.O.C. system \eqref{eq:KKT_system} is required,  which we formalize via the notion of \emph{strong regularity} \cite{robinson_strongly_1980} in Definition \ref{def:strong_regular_solution}.
\begin{definition}[Strong regularity of optimal solution]
    Suppose that our Fourier-based representation for the true Lagrangian gradient $\hat{\mathcal{L}}_{\nabla_{\mathbf{z}}}^{\mathrm{Fou}}(\mathbf{z})$  is continuously differentiable. We say that a solution $\mathbf{z}^*$ is strongly regular if there exist neighborhoods $U_1$ and $U_2$ of  $\mathbf{0}_{\mathcal{Z}}$ and $\mathbf{z^*}$, such that for every $\delta_{\mathbf{z}} \in U_1$ the linearized  equation
    \begin{equation}
        \delta_{\mathbf{z}} + \hat{\mathcal{L}}_{\nabla_{\mathbf{z}}}^{\mathrm{Fou}}(\mathbf{z^*}) + \hat{\mathcal{L}}_{\nabla_{\mathbf{z}}^2}^{\mathrm{Fou}}(\mathbf{z^*})\paren{\mathbf{z}-\mathbf{z^*}} = 0
    \label{eq:strong_regular_sol}
    \end{equation}
has a unique solution 
 in $U_2$, denoted $\tilde {\mathbf{z}} = \tilde {\mathbf{z}} (\delta_\mathbf{z}) $, and $\tilde {\mathbf{z}}(.)$ is Lipschitz continuous on $U_1$.
\label{def:strong_regular_solution}
\end{definition}
Next, we state Theorem~\ref{theorem:consistentcy-RQMC-Fou}, which establishes the consistency of the Fourier–RQMC solution w.r.t.$N$.
\begin{theorem}[Consistency of solution from Fourier-RQMC problem]\label{theorem:consistentcy-RQMC-Fou}
Fix $S_{\mathrm{shift}} = \bar S$, and let $\curly{\mathbf{v}_n^{(s)}}_{n=1}^{N}$ be constructed as in  Lemma \ref{lem:uniform_conv_RQMC_Fou}. Suppose that Assumption \ref{assump:regularity_exact_KKT} holds and that the exact solution $\mathbf z^\ast$ is strongly regular in the sense of Definition \ref{def:strong_regular_solution}. \footnote{Under Assumption \ref{assump:regularity_exact_KKT}, strong regularity of $\mathbf z^\ast$ follows from \cite[Proposition~16]{sun_strong_2006}.} Then, as $N \to \infty$, the Fourier-RQMC problem admits a  (locally) unique solution
${\mathbf{z}}_{N,\bar S}^{\mathrm{RQMC},*} \in \mathcal{Z}$ , and
\[
{\mathbf{z}}_{N,\bar S}^{\mathrm{RQMC},*}
\xrightarrow[]{\mathrm{a.s.}}
\mathbf{z}^* .
\]
\end{theorem}
\begin{proof}
The detailed proof is presented in Appendix \ref{appendix:proof_prop_consistency}.
\end{proof}
In RQMC, the CLT for the estimators is obtained by letting the number of i.i.d shifts $S_\mathrm{shift} \to \infty$. In the Fourier-RQMC setting for the MSRM problem, the CLT over shifts will describe the fluctuations of $\mathbf{z}_{N,S_{\mathrm{shift}}}^{\mathrm{RQMC},*}$ around $\mathbf{z}_{N,\infty}^{\mathrm{RQMC},*}$ not the exact solution $\mathbf{z}^*$. To express a limiting distribution centered at $\mathbf z^\ast$, we therefore consider the joint regime in which $S_\mathrm{shift} \to \infty,N \to \infty$, and impose the following assumptions.
\begin{assumption}\
\begin{enumerate}[label=(\roman*), ref=\roman*]
\item \label{ass:i} \(\sqrt{S_{\mathrm{shift}}} N^{r}\,\norm{\mathbf{z}_{N,\infty}^{\mathrm{RQMC},*} - \mathbf{z}^*} \xrightarrow[]{\mathbb{P}} 0\), as \(S_{\mathrm{shift}} \to \infty, N\to \infty\) . 
\item \label{ass:ii} There exists a positive semidefinite matrix $\boldsymbol{H}(\mathbf{z^*})$ such that: \begin{equation*} \lim_{N \to \infty} N^{2r} \mathrm{Var}_S\!\paren{ I_{N}^{\mathrm{RQMC}} \brac{ \mathcal{H}^{(1)}\paren{\mathbf{v}_n^{(s)},\mathbf{z}^{*}} }} = \boldsymbol H\!\left( \mathbf{z}^{*} \right). 
\end{equation*} 
\end{enumerate} 
with $r$ is defined in \eqref{eq:stat_error_RQMC_rate}.
\label{assump:joint_grow_N_S}
\end{assumption}
Assumption \eqref{ass:i} requires $N$ to grow sufficiently fast relative to $S_\mathrm{shift}$, so the term  $\mathbf{z}_{N,\infty}^{\mathrm{RQMC},*} - \mathbf{z}^*$ is negligible at $\sqrt{S_\mathrm{shift}} N^r$ scale. Joint growth conditions of this type for many randomization schemes are discussed in detail in \cite{nakayama_sufficient_2024}. Assumption \eqref{ass:ii} allows us to identify the asymptotic covariance and recenter the CLT at $\mathbf{z}^*$. 

We now state Theorem~\ref{theorem:efficiency-RQMC-Fou}, which describes the asymptotic behavior of $\mathbf{z}_{N,S_{\mathrm{shift}}}^{\mathrm{RQMC},*}$.

\begin{theorem}[CLT for the Fourier-RQMC solution]
\label{theorem:efficiency-RQMC-Fou}
Let $\curly{\mathbf{v}_n^{(s)}}_{n=1}^{N}$ be constructed as in Lemma \ref{lem:uniform_conv_RQMC_Fou}. Suppose that Assumption \ref{assump:joint_grow_N_S} holds. Then, as \(S_{\mathrm{shift}} \to \infty, N \to \infty\),
\begin{equation}
    \sqrt{S_{\mathrm{shift}}}N^{r}\,
    \paren{
    {\mathbf{z}}_{N,S_{\mathrm{shift}}}^{\mathrm{RQMC},*} - \mathbf{z}^{*}
    }
    \xrightarrow{\mathrm{law}}
    \mathcal{N}\!\left( \mathbf{0}, \boldsymbol V\paren{\mathbf{z}^*} \right).
    \label{eq:CLT_solution_m_1}
\end{equation}
, and the sandwich covariance matrix \(\boldsymbol V\paren{\mathbf{z}^*}\) is given by
\begin{equation}
\boldsymbol V\paren{\mathbf{z}^*}
:=
\paren{
\hat{\mathcal{L}}_{\nabla_{\mathbf{z}}^2}^{\mathrm{Fou}}\big( \mathbf{z}^{*}\big)
}^{-1}
\,\boldsymbol H\!\left( \mathbf{z}^{*} \right)\,
\paren{
\hat{\mathcal{L}}_{\nabla_{\mathbf{z}}^2}^{\mathrm{Fou}}\big( \mathbf{z}^{*}\big)
}^{-1}.
\label{eq:CLT_solution_m_2}
\end{equation}
\end{theorem}
\begin{proof}
To prove this theorem, we also need to establish  the consistency for the solution w.r.t. $S_{\mathrm{shift}}$, which is mentioned in Proposition~\ref{prop:consistentcy-RQMC-Fou-S}. The detailed proof is provided in Appendix~\ref{appendix:proof_theorem_efficiency}.
\end{proof}
Using Theorem~\ref{theorem:efficiency-RQMC-Fou} together with \eqref{eq:stat_error_RQMC_rate}, the statistical error satisfies
\begin{equation}
\varepsilon_{\mathrm{stat}}^{\mathrm{RQMC}}(N)
   = \mathcal{O}\paren{N^{-r}}   .
\label{eq:statistical_error_RQMC_sol}
\end{equation}
\begin{remark}[Estimating $\boldsymbol V$ under digital shift randomization]
\label{rema:replace_inverse_hessian_RQMC}
In the numerical experiments of Section \ref{subsec:sing-RQMC}, we employ digital shift randomization to compute RQMC estimators. For this randomization, a USLLN w.r.t. $N$ does not generally hold for Fourier–RQMC estimators; see  \cite{owen_strong_2021}. Nevertheless, the statistical error and asymptotic variance can still be characterized via a CLT by letting $S_{\mathrm{shift}}\to\infty$. In this setting, the covariance matrix  $\boldsymbol V$ in \eqref{eq:CLT_solution_m_2} is obtained by  replacing $\mathbf{z}^*$ with $\mathbf{z}_{N,\infty}^{\mathrm{RQMC},*}$,
$\hat{\mathcal{L}}_{\nabla_{\mathbf{z}}^2}^{\mathrm{Fou}}\big(\mathbf{z}^{*}\big)$
with
\(
I_{N,S_{\mathrm{shift}}}^{\mathrm{RQMC}}
\brac{
\mathcal{H}^{(2)}
\big(\cdot;\mathbf{z}_{N,\infty}^{\mathrm{RQMC},*}\big)
}
\), and $\boldsymbol{H}(\mathbf{z^*})$ with $\boldsymbol{H}_{N,\infty}^{\mathrm{RQMC}}\paren{\cdot;\mathbf{z}_{N,\infty}^{\mathrm{RQMC},*}}$.
This substitution is justified because a USLLN does hold for Fourier-RQMC estimators w.r.t. $S_{\mathrm{shift}}$; see Lemma \ref{lem:uniform_conv_RQMC_Fou_S} and Proposition \ref{prop:consistentcy-RQMC-Fou-S}.
\end{remark}
We conclude with remarks on how to compute the statistical error for single-level RQMC in practice.
\begin{remark}[Estimation of $\boldsymbol H(\mathbf z^*)$]
The true variance w.r.t. the random-shift measure is generally unknown and must be estimated numerically.  
A natural estimator is the sample variance of the RQMC estimator,
\(
\mathrm{Var}\!\left(
I_{N,S_{\mathrm{shift}}}^{\mathrm{RQMC}}
\!\left[
\mathcal{H}^{(1)}(\cdot;\mathbf z^*)
\right]
\right).
\)
This estimator is well defined since, under Assumption~\ref{ass:A10}, the transformed integrands admit finite second moments. Then, by the LLN applied to the i.i.d shifts $s$, we have
\[
\mathrm{Var}\!\left(
I_{N,S_{\mathrm{shift}}}^{\mathrm{RQMC}}
\!\left[
\mathcal{H}^{(1)}(\cdot;\mathbf z^*)
\right]
\right)
\;\xrightarrow[]{}\;
\mathrm{Var}_S\!\left(
I_{N}^{\mathrm{RQMC}}
\!\left[
\mathcal{H}^{(1)}
\bigl(\mathbf v_n^{(s)}, \mathbf z^*\bigr)
\right]
\right),
\], as $S_{\mathrm{shift}}\to\infty$.
\label{rema:var_sing_rqmc_estim}
\end{remark}
\begin{remark}[Statistical error of the single-level Fourier--RQMC solution]\label{rema:compute_stat_err_sol_sing}
We construct a practical plug-in estimator of the asymptotic covariance matrix $\boldsymbol{V}(\mathbf{z}^*)$ from two components:\footnote{If $\mathbf z^*$ is unavailable, we replace it by the solution returned at the last step of the optimization process $\mathbf z_{J,N,S_{\mathrm{shift}}}^{\mathrm{RQMC}}$.}
\begin{enumerate}
\item \emph{Gradient variance.}  
The variance of the RQMC estimator of the gradient block, $\mathrm{Var}\!\left(
I^{\mathrm{RQMC}}_{N,S_{\mathrm{shift}}}
\big[\mathcal H^{(1)}(\cdot;\mathbf z^\ast)\big]
\right)$ which is estimated using the sample variance over the random digital shifts as described in Remark~\ref{rema:var_sing_rqmc_estim}, with componentwise contributions $\mathrm{Var}\!\big(
I^{\mathrm{RQMC}}_{N,S_{\mathrm{shift}}}
[\tilde h^{(\nu)}(\cdot;\mathbf{m}^\ast)]
\big)$.

\item \emph{Hessian approximation.}  
A Fourier--RQMC approximation of the Hessian term
$\widehat{\mathcal L}^{\mathrm{Fou}}_{\nabla_z^2}(\mathbf z^\ast)$, computed via
\(
I^{\mathrm{RQMC}}_{N,S_{\mathrm{shift}}}
[\mathcal H^{(2)}(\cdot;\mathbf z^\ast)].
\)
\end{enumerate}
For the first component, independence of the randomized digital nets across interaction orders
$k \in \mathcal I_{q_\ell}$ implies that variances add across $k$, yielding
\begin{equation*}
\mathrm{Var}\!\left(
I^{\mathrm{RQMC}}_{N,S_{\mathrm{shift}}}
[\tilde h^{(\nu)}(\cdot;\mathbf m^\ast)]
\right)
=
\sum_{k\in \mathcal I_{q_\ell}}
\mathrm{Var}\!\left(
\sum_{\mathbf p\in\mathcal I_k}
I^{\mathrm{RQMC}}_{N,S_{\mathrm{shift}}}
[\tilde h^{(\nu)}_{k,p}(\cdot;\mathbf m^\ast_{k,p})]
\right).
\end{equation*}

For a fixed interaction order $k$, all components $\mathbf p\in\mathcal I_k$ share the same digital net and the same random shifts, and are therefore correlated. Consequently, the variance within each $k$ expands as
\begin{align*}
\mathrm{Var}\!\left(
\sum_{\mathbf p\in\mathcal I_k}
I^{\mathrm{RQMC}}_{N,S_{\mathrm{shift}}}
[\tilde h^{(\nu)}_{k,p}(\cdot;\mathbf m^\ast_{k,p})]
\right)
&=
\sum_{p\in\mathcal I_k}
\mathrm{Var}\!\left(
I^{\mathrm{RQMC}}_{N,S_{\mathrm{shift}}}
[\tilde h^{(\nu)}_{k,p}(\cdot;\mathbf m^\ast_{k,p})]
\right) \notag \\
&\quad
+ 2 \!\!\!
\sum_{\substack{\mathbf p, \mathbf t\in\mathcal I_k\\ \mathbf p<\mathbf t}}
\mathrm{Cov}\!\left(
I^{\mathrm{RQMC}}_{N,S_{\mathrm{shift}}}
[\tilde h^{(\nu)}_{k, p}],
\,
I^{\mathrm{RQMC}}_{N,S_{\mathrm{shift}}}
[\tilde h^{(\nu)}_{k, t}]
\right).
\end{align*}
Combining the gradient variance estimator in~(i) with the Hessian approximation in~(ii) yields the plug-in covariance estimator
$V^{\mathrm{RQMC,sing}}_{N,S_{\mathrm{shift}}}(\cdot;\mathbf z^\ast)$.
The resulting statistical error of the single-level Fourier--RQMC solution is then estimated by
\begin{equation}
\varepsilon^{\mathrm{RQMC,sing}}_{N,S_{\mathrm{shift}}}(\mathbf z^\ast)
=
\frac{C_\alpha}{\sqrt{S_{\mathrm{shift}}}}
\sqrt{
\left\|
V^{\mathrm{RQMC,sing}}_{N,S_{\mathrm{shift}}}(\cdot;\mathbf z^\ast)
\right\|
},
\label{eq:stat_err_sol_sing_RQMC}
\end{equation}
where $C_\alpha$ is defined in~\eqref{eq:RMSE_RQMC}.  
All variances and covariances above are taken with respect to the random digital shifts, conditional on the underlying Sobol base nets.
\end{remark}
\subsection{Computational Complexity}\label{subsec:computational_complex_Fou_RQMC}
Combining the optimization error bound \eqref{eq:opti_err_SAA_slinear} with the statistical error rate \eqref{eq:statistical_error_RQMC_sol}, we now derive the computational complexity of the Fourier–RQMC methods for both single-level and multilevel constructions. 
\subsubsection{Single-level Fourier-RQMC}\label{subsubsec:Single Fou-RQMC comp complexity}
We first analyze the computational complexity of the single-level Fourier–RQMC scheme, summarized in Corollary~\ref{coro:complexity_sing_RQMC} below.
\begin{corollary}[Single-level Fourier-RQMC work complexity]\label{coro:complexity_sing_RQMC}
Consider the single-level Fourier-RQMC estimator defined in \eqref{eq:RQMC_estimator} with $N$ points and $S_{\mathrm{shift}}$ shifts. Let $\varepsilon$ denote the target total error and allocate the error budget by enforcing
$\varepsilon_{\mathrm{stat}}^{\mathrm{RQMC}}(N)\le \varepsilon/2$ and $\varepsilon_{\mathrm{opt}}(J)\le \varepsilon/2$.
Choose $N=N(\varepsilon)$ accordingly, and  $J=J(\varepsilon)$ denote the number of optimization iterations. Then the total work satisfies
\begin{equation}
W_{\mathrm{sing}}^{\mathrm{RQMC}}(\varepsilon)
= \mathcal O\!\left(\varepsilon^{-\tfrac{1}{r}}\,\log\!\log(1/\varepsilon)\right),
\label{eq:complexity_sing_RQMC_coro}
\end{equation}
where $r$ is defined in \eqref{eq:stat_error_RQMC_rate}.
\end{corollary}

\begin{proof}
Appendix \ref{appendix:proof_coro_sing_allocation} presents the proof.
\end{proof}
\begin{remark}[On constants in the single-level work complexity]
    The hidden constant in $W_{\mathrm{sing}}^{\mathrm{RQMC}}(\varepsilon)$ may depend on the decision dimension $d$ and on the number of component integration problems $N_{\mathrm{comp}}$; see Appendix \ref{appendix:proof_coro_sing_allocation} for the explicit cost decomposition.
\end{remark}
We next derive the total work complexity of the multilevel Fourier–RQMC estimator, highlighting its complexity improvements relative to the single-level scheme.
\subsubsection{Multilevel Fourier-RQMC}\label{subsubsec:Multilevel Fou-RQMC comp complexity}
When moving to the multilevel setting (Section~\ref{subsec:Multilevel Fou-RQMC}), we allow the sample size to vary across levels, i.e., we use $\{N_j\}_{j=1}^{J}$ together with $S_{\mathrm{shift}}$  shifts at each level $j$. As described in Section ~\ref{subsec:Multilevel Fou-RQMC}, the multilevel construction couples levels through a shared Sobol base net, while independent shifts are applied across levels. Following Remark \ref{rema:separate_opt_stat_err}, to determine $\{N_j\}_{j=1}^{J}$, we can solve a preliminary constrained optimization problem in which we minimize the total work complexity subject to achieving the prescribed statistical error at $\varepsilon_{\mathrm{stat}}^{\mathrm{RQMC}} \leq \varepsilon/2$ at final iteration $J$. The resulting work-optimal allocation is stated in Corollary~\ref{coro:multilevel_allocation}.

\begin{corollary}[Work-optimal multilevel allocation]
\label{coro:multilevel_allocation}
For each level $j$, let the smoothness parameter
\(
A_{\mathrm{mult},j}^{*}
\)
correspond to the most singular component among the level-$j$ difference
integrands $\Delta \tilde h_{k,p}^{(\nu,j)}$.
Assume that \(
A_{\mathrm{mult},j}^{*} = A_{\mathrm{sing}}^{*}
\) for all $j=1,\dots,J$. Then the work-optimal sample across levels is given by
\begin{equation}
    N_j \;\propto\;
    \boldsymbol D_j^{\tfrac{1}{2r+1}} \, c_j^{-\tfrac{1}{2r+1}},
    \label{eq:raw_N_j}
\end{equation}
and the corresponding multilevel computational work scales as
\begin{equation}
    W_{\mathrm{mult}}^{\mathrm{RQMC}}(\varepsilon)
    \;\propto\;
    S_1^{\tfrac{2r+1}{2r}} \, \varepsilon^{-\tfrac{1}{r}}.
    \label{eq:min_work_mult}
\end{equation}
where $r$ is defined in \eqref{eq:stat_error_RQMC_rate},
\(
\boldsymbol D_j :=
\norm{
\mathrm{Var}\!\left[
I_{N_j,S_{\mathrm{shift}}}^{\mathrm{RQMC}}
\brac{
\Delta \mathcal{H}^{(1)}
\paren{\cdot;\mathbf{z}^{(j)},\mathbf{z}^{(j-1)}}
}
\right]
}
\)\footnote{To simplify notation, throughout this section we denote the solution
returned by the Fourier-RQMC problem at iteration $j$ by
$\mathbf{z}_{N_j,S_{\mathrm{shift}}}^{(\mathrm{RQMC},j)} \equiv \mathbf{z}^{(j)}$
and $\mathbf{m}_{N_j,S_{\mathrm{shift}}}^{(\mathrm{RQMC},j)} \equiv \mathbf{m}^{(j)}$.},
\(
S_1 := \sum_{j=1}^{J} \boldsymbol D_{j}^{\frac{1}{1+2r}} c_j^{\frac{2r}{1+2r}},
\) and $c_j$ is the total cost at level $j$.
\end{corollary}
\begin{proof}
Appendix \ref{appendix:proof_coro_multi_allocation} presents the proof.
\end{proof}
\begin{remark}[Improved regularity at higher levels]\label{rem:better rates}
In practice, the difference integrands often exhibit improved smoothness as $j$ increases, which corresponds to larger effective rates $r_j$ at finer levels. Allowing level-dependent rates in the allocation problem can further improve the constants (and potentially the work) compared to the worst-case bound in Corollary \ref{coro:multilevel_allocation}; we keep the uniform-rate assumption here to obtain a closed-form allocation and  complexity estimate.
\end{remark}
In order to derive a closed-form expression for $\{N_j\}_{j=1}^J$ in Proposition~\ref{prop:choosing_Nj}, we impose an additional regularity assumption on the Fourier-RQMC difference surrogates.
\begin{assumption}\label{ass:diff_surrogate_Lipschitz}
For all $j\ge J_{\mathrm{loc}}$,
$I_{N_j,S_{\mathrm{shift}}}^{\mathrm{RQMC}}\!\big[\Delta \mathcal H^{(1)}(\cdot;\mathbf z^{(j)},\mathbf z^{(j-1)})\big]$
satisfy a mean-square Lipschitz property w.r.t. the iterates; that is, there exist a constant $L_H \geq 0$ such that
\begin{equation*}
\mathbb{E}\brac{\norm{{I}_{N_j,S_{\mathrm{shift}}}^{\mathrm{RQMC}} \brac{\Delta \mathcal{H}^{(1)}\paren{\cdot;\mathbf{z}^{(j)},\mathbf{z}^{(j-1)}}}}^2}\leq
L_H \,\mathbb{E}\brac{\norm{\mathbf{z}^{(j)} - \mathbf{z}^{(j-1)}}^2}
\end{equation*}
\end{assumption}
\begin{proposition}[Adaptive choice of the sample size $N_j$]
Let $J_{\mathrm{loc}}$ denote the iteration index at which the SQP iterates enter
their \emph{local convergence region}. Suppose that Assumption \ref{ass:diff_surrogate_Lipschitz} holds; then the sample size at iteration $j$ can be chosen as
\begin{equation}
\label{eq:Nj_rule}
N_j
=
\max\!\left\{
N_1\,\mathbf{1}_{\{j<J_{\mathrm{loc}}\}}
+
C_{\mathrm{loc},j-1}\,N_1\,\tilde\eta_{j-1}^{\frac{2}{2r+1}}
\,\mathbf{1}_{\{j\ge J_{\mathrm{loc}}\}},
\;
N_{\min}
\right\},
\end{equation}
where $N_1$ denotes the initial number of QMC points,
$N_{\min}$ is a prescribed minimal sample size,
$C_{\mathrm{loc},j-1}>0$ is a level-dependent constant, and
$\tilde\eta_{j-1} \to 0$ as $j\to\infty$.

If we make the local convergence rate constant
$\tilde\eta_{j-1}=\eta\in(0,1)$ for all $j\ge J_{\mathrm{loc}}$,
the allocation \eqref{eq:Nj_rule} reduces to
\begin{equation}
N_j
=
\max\!\left\{
N_1\,\mathbf{1}_{\{j<J_{\mathrm{loc}}\}}
+
C_{\mathrm{loc}}\,N_1\,\eta^{\frac{2(j-1-J_{\mathrm{loc}})}{2r+1}}
\,\mathbf{1}_{\{j\ge J_{\mathrm{loc}}\}},
\;
N_{\min}
\right\}.
\label{eq:Nj_same_eta}
\end{equation}
\label{prop:choosing_Nj}
\end{proposition}
\begin{proof}
Proof for Proposition \ref{prop:choosing_Nj} is presented in Appendix \ref{appendix:proof_prop_choosing_Nj}.
\end{proof}
\begin{remark}[Choice of Sobol points]
    Since Sobol sequences exhibit their best performance when the number of points is a power of two \cite{sobol_construction_2011}, in our numerical experiments we select,
    \[
        N_j \;\approx\; 2^{\left\lfloor \log_2 \!\paren{C_{\mathrm{loc}}\,N_1\,\eta^{\frac{2(j-1-J_{\mathrm{loc}})}{2r+1}}} \right\rfloor},
    \]
    i.e., the nearest power of two to $C_{\mathrm{loc}} N_1 \eta^{\frac{2(j-1)}{r+1}}$, and 
    likewise choose $N_{\min}$ as a power of two.
\label{remark: choice_Sobol_QMC}
\end{remark}
Under the multilevel sampling design described above, we now derive the asymptotic computational complexity.
\begin{proposition}[Computational Complexity for multilevel Fourier-RQMC] 
\label{prop:work_multi_level}
Assume that the per-iteration cost of multilevel Fourier-RQMC is level-independent, i.e.,
\(c_{j} \approx c\) for all \(j\) where $c$ denotes the per-iteration cost of the single-level RQMC.
Assume further \(\eta_j \approx \eta\) within the local convergence region \(j \ge J_{\mathrm{loc}}\),
then
\begin{equation}
\frac{W^{\mathrm{RQMC}}_{\mathrm{sing}}(\varepsilon)}{W^{\mathrm{RQMC}}_{\mathrm{mult}}(\varepsilon)}
 = \mathcal{O}(J).
\label{eq:W_sing_bound}
\end{equation}
Moreover, the total computational cost of the multilevel Fourier-RQMC method satisfies
\begin{equation}
W_{\mathrm{mult}}(\varepsilon)
=
\mathcal{O}\!\left(
\varepsilon^{-\tfrac{1}{r}}
\right),
\label{eq:W_mult_final}
\end{equation}
where \(r=1-A_{\mathrm{sing}}^*-\varsigma\).
\end{proposition}
\begin{proof}
Proof for Proposition \ref{prop:work_multi_level} is presented in Appendix \ref{appendix:proof_prop_work_multi_level}.
\end{proof}

\begin{remark}[On constants in the multilevel work complexity]
From Proposition~\ref{prop:work_multi_level}, the total multilevel computational work is, asymptotically, reduced by a factor proportional to the number of optimization iterations $J$ when compared with the total single-level work. Beyond this asymptotic gain, the leading constant in the multilevel work complexity can be further improved in practice because	the difference integrands typically exhibit substantially smaller variance than the original integrands, which reduces the variance prefactors $D_j$ entering the multilevel allocation and hence lowers the overall computational cost. These constant-level improvements are not captured by the worst-case bounds in Proposition~\ref{prop:work_multi_level} but are consistently observed in numerical experiments in Section \ref{sec:num_exp_results}.
\label{rema:improve_const_mult}
\end{remark}
\begin{remark}[Statistical error of multilevel Fourier-RQMC solution]\label{rema:compute_var_mult_RQMC}
In the multilevel setting, the variance $\mathrm{Var}\paren{I_{N,S_{\mathrm{shift}}}^{\mathrm{RQMC}}\brac{\tilde h^{(\nu)}\paren{\cdot;\mathbf{m}^{*}}}}$ is decomposed as
\begin{equation}
\begin{aligned}
      \mathrm{Var}\!\left(
I_{N,S_{\mathrm{shift}}}^{\mathrm{RQMC}}
\Big[
\tilde h^{(\nu)}(;\mathbf m^{*})
\Big]
\right)
&=
\mathrm{Var}\!\left(
I_{N_1,S_{\mathrm{shift}}}^{\mathrm{RQMC}}
\Big[
\tilde h^{(\nu)}\big(\cdot;\mathbf m^{(\mathrm{RQMC},1)}\big)
\Big]
\right)
\\
&\quad+
\sum_{j=2}^{J}
\mathrm{Var}\!\left(
I_{N_j,S_{\mathrm{shift}}}^{\mathrm{RQMC}}
\Big[
\Delta \tilde h^{(\nu)}\big(\cdot;\mathbf m^{(\mathrm{RQMC},j)},\mathbf m^{(\mathrm{RQMC},j-1)}\big)
\Big]
\right). 
\end{aligned}
 \label{eq:var_mult}
\end{equation}
Since independent shifts are generated at each level $j$, the level-wise variances are estimated separately, using the same  \emph{Gradient variance} estimator as in~Remark~\ref{rema:compute_stat_err_sol_sing}. For levels $j\ge 2$, the estimator is applied to the difference integrands $\Delta \tilde h^{(\nu)}$ (componentwise, $\Delta \tilde h_{k,p}^{(\nu)}$). Combining \eqref{eq:var_mult} with Remark~\ref{rema:compute_stat_err_sol_sing} yields an estimator $\boldsymbol V_{N,S_{\mathrm{shift}}}^{\mathrm{RQMC,mult}}(\cdot;\mathbf z^*)$ of $\boldsymbol V(\mathbf z^*)$, and the corresponding statistical error of the solution is
\begin{equation}
\varepsilon_{N,S_{\mathrm{shift}}}^{\mathrm{RQMC,mult}}(\mathbf z^*)
= \frac{C_\alpha}{\sqrt{S_{\mathrm{shift}}}}
\,\sqrt{\norm{\boldsymbol V_{N,S_{\mathrm{shift}}}^{\mathrm{RQMC,mult}}(\cdot;\mathbf z^*)}}.
\label{eq:stat_err_sol_mult_RQMC}
\end{equation}
with $C_\alpha$ is defined in~\eqref{eq:RMSE_RQMC}.
\end{remark}
\section{Numerical Experiments and Results}\label{sec:num_exp_results}
In this section, we evaluate the performance of the proposed Fourier–RQMC methods for the MSRM problem on three representative test cases: an exponential loss under a bivariate Gaussian model (see Section \ref{subsec:expo_loss}) and a QPC loss under both a 10D Gaussian model (see Section \ref{subsec:10D_gaussian_qcl})  and a 3D NIG model (see Section \ref{subsec:3D_NIG_qcl}). The  component-wise Fourier transforms for given loss functions are derived in Appendix \ref{Appendix: Fourier_transform_loss}. For the loss vector $\mathbf{X}$, Gaussian parameter sets are taken from \cite{armenti_multivariate_2018} and NIG parameter sets from \cite{kaakai_multivariate_2022}, and the corresponding component CFs are computed using the formulas in Appendix~\ref{appendix:loss_vector_distr}.

For the numerical implementation, the optimization error is determined by the numerical optimization solver (SLSQP).\footnote{In all experiments, we control the deterministic optimization error through the SLSQP stopping tolerance \texttt{ftol}, denoted $f_{\mathrm{tol}}$. Since \texttt{ftol} bounds several stopping criteria, we heuristically model the resulting optimization error as $\varepsilon_{\mathrm{opt}}\approx f_{\mathrm{tol}}^{\,2}$, consistent with the local superlinear regime in Theorem~\ref{theorem:superlinear_converge_SQP}.}
The statistical error of the solution is quantified as follows:
(i) for single-level Fourier-RQMC, we use \eqref{eq:stat_err_sol_sing_RQMC};
(ii) for multilevel Fourier-RQMC, we use \eqref{eq:stat_err_sol_mult_RQMC};
(iii) for SAA, we use \eqref{eq:stat_err_sol_SAA}.
In all cases, we report the statistical error using the maximum diagonal norm $\|\cdot\|_{\mathrm{diag},\infty}$ \footnote{For a matrix $\boldsymbol A\in\mathbb R^{d\times d}$, the maximum diagonal norm is defined as
\(
\|\boldsymbol A\|_{\mathrm{diag},\infty}
:= \max_{1\le i\le d} |A_{ii}|.
\)}
, and choose $C_\alpha = 1.96$ for 95\% confidence level.  For method comparison, we also report the relative statistical error
\begin{equation}
    \varepsilon_{\mathrm{stat,rel}}
:=\frac{\varepsilon_{\mathrm{stat}}}{\|\mathbf z^{\mathrm{ref}}\|_\infty},
\label{eq:relative_stat_err}
\end{equation}
where $\mathbf z^{\mathrm{ref}} := \paren{\mathbf{m}^{\mathrm{ref}}, \lambda^{\mathrm{ref}}}$ denotes a reference solution (available in closed form when possible, or otherwise approximated by SAA with $N=10^8$ samples, where the reference is taken as the final solution returned by the solver using an optimization tolerance $\varepsilon_{\mathrm{opt}}=10^{-6}$).

Reported computational times exclude the cost of estimating statistical errors and account solely for the runtime of the optimization procedures. All experiments were conducted using Python~3.13.2 on a MacBook Pro  with an Apple~M4~Pro. The code implementing the proposed methods is available on GitHub\footnote{\href{https://github.com/nnt2306/Single-and-Multi-level-Fourier-RQMC-for-MSRM}{https://github.com/nnt2306/Single-and-Multi-level-Fourier-RQMC-for-MSRM}}.
\subsection{Exponential Loss with Two-Dimensional Gaussian Loss Vector}\label{subsec:expo_loss}
In this section, we test the Fourier-RQMC method using the exponential loss function from \cite{kaakai_estimation_2024}, defined in \eqref{eq:multi_entropy}, with $\alpha=1$ and $\beta=1$. We consider a bivariate Gaussian loss vector with mean $\boldsymbol{\mu}=(0,0)^\top$ and $\boldsymbol{\Sigma}=\begin{pmatrix} 1 & \rho \\ \rho & 1 \end{pmatrix}$, where $\rho\in\{-0.5,\,0.5\}$. In this setting, the optimal allocation $\mathbf{m}^*$ admits a closed-form expression and is symmetric, i.e., $m_1^*=m_2^*$ (see \cite[Lemma~3.1]{kaakai_estimation_2024}). Table~\ref{tab:exact_and_ci} reports this closed-form solution together with the single-level Fourier-RQMC estimate and its $95\%$ confidence interval (CI).
\begin{table}[H]
\centering
\begin{tabular}{ccc}
\toprule
$\rho$ & $m_1^*=m_2^*$  & CI for $m_1(=m_2)$ \\
\midrule
$-0.5$ & $0.3868$ & $[0.38688,\,0.38690]$ \\
$0.5$  & $0.6364$ & $[0.63645,\,0.63647]$ \\
\bottomrule
\end{tabular}
\vspace{-2mm}
\caption{\small Exact optimal allocations and 95\% CI from single-level Fourier--RQMC ($N=2048$, $S_{\mathrm{shift}}=32$).}
\label{tab:exact_and_ci}
\end{table}
We see the convergence of the Fourier-RQMC solution at order $10^{-4}$ to the exact solution from Table \ref{tab:exact_and_ci}, and this is further supported by the convergence plot shown in Figure \ref{fig:conv-expo-grid}, where the local superlinear convergence of the optimization solver is observed.
\begin{figure}[H]
  \centering
  \begin{subfigure}{0.5\textwidth}
    \centering
    \includegraphics[width=\linewidth]{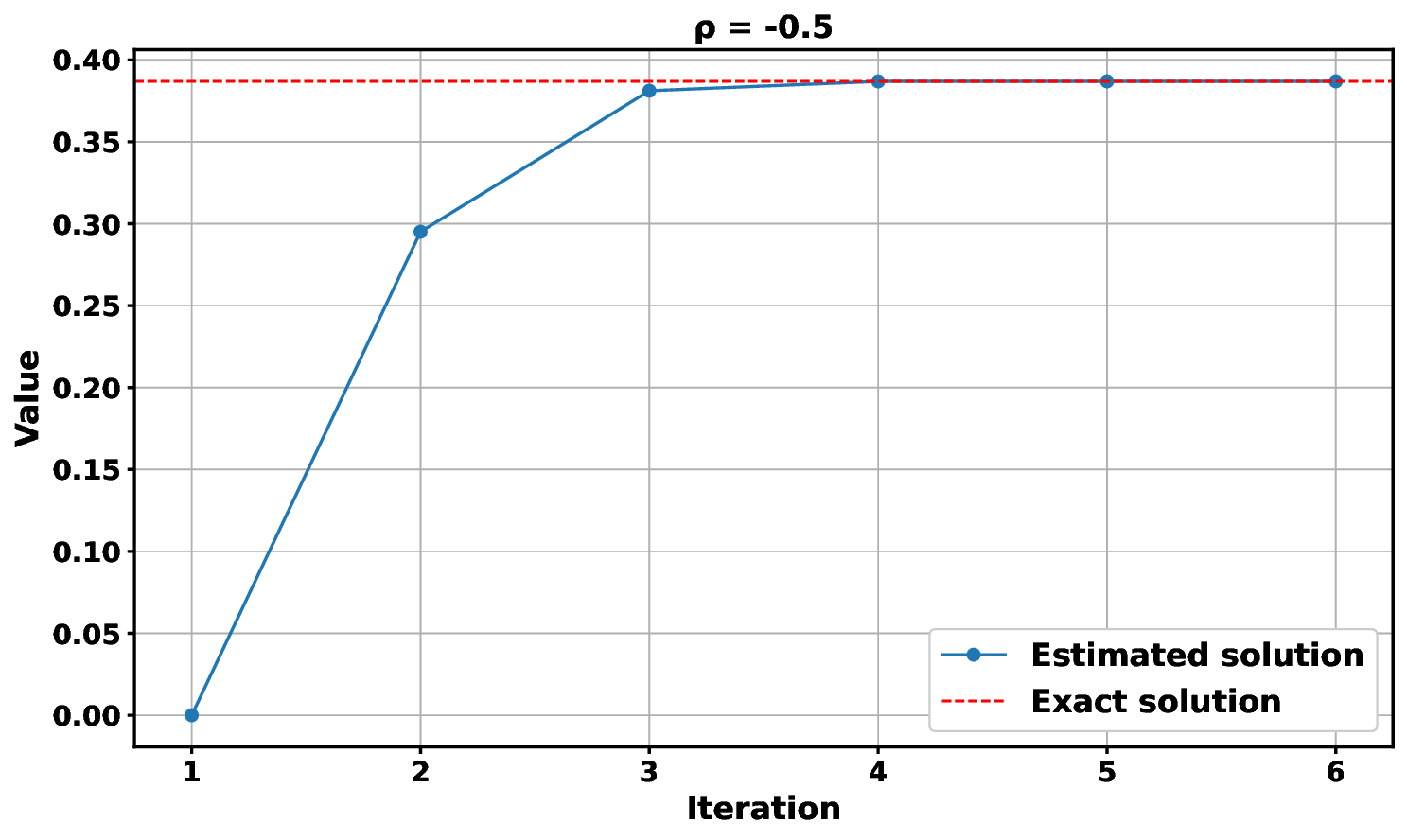}
  \end{subfigure}\hfill
  \begin{subfigure}{0.5\textwidth}
    \centering
    \includegraphics[width=\linewidth]{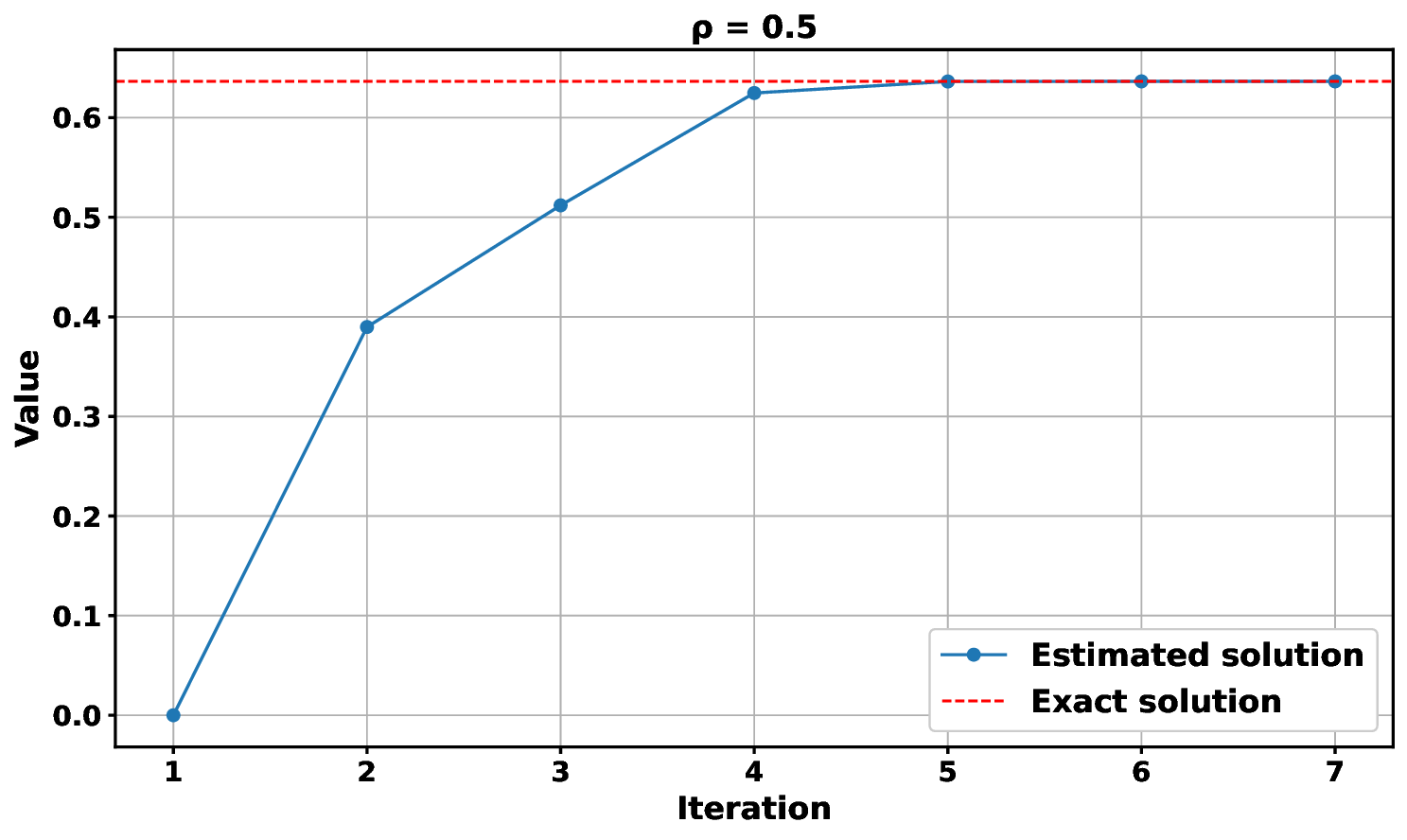}
  \end{subfigure}
  \caption{\small Exponential loss with a two-dimensional Gaussian loss vector: convergence of the single-level Fourier--RQMC iterate $m^{(j)}$ to the exact solution $m^*$ for $\rho=-0.5$ (left) and $\rho=0.5$ (right).}
  \label{fig:conv-expo-grid}
\end{figure}
Next, we choose the optimization tolerance $\varepsilon_{\mathrm{opt}}$ sufficiently small relative to $\varepsilon_{\mathrm{stat}}$, so that the observed error is dominated by the statistical component. This allows us to directly assess the rate of convergence of the statistical error as a function of the total sampling budget B for single-level Fourier–RQMC and  SAA. Figure~\ref{fig:err_sample_expo} shows that, for any prescribed relative statistical tolerance $\varepsilon_{\mathrm{stat,rel}}$, single-level Fourier–RQMC achieves the target accuracy with a substantially smaller budget than SAA, with the performance gap widening as $\varepsilon_{\mathrm{stat,rel}}$ decreases. The empirical convergence rates for Fourier–RQMC, estimated as $r = 1.49$ for $\rho=-0.5$ and $r = 1.29$ for $\rho=0.5$, are significantly higher than the
rate $r = 1/2$ observed for SAA. While these empirical rates are faster than the asymptotic rate in \eqref{eq:stat_error_RQMC_rate}, they are consistent with Remark~\ref{rema:non_asymptoic_err_decay} and can be attributed to the domain transformation introduced in Section~\ref{subsubsec:domain_trans_QMC}, which mitigates boundary-induced oscillations and enhances the effective smoothness of the integrand.
\begin{figure}[H]
  \centering
  \begin{subfigure}{0.5\textwidth}
    \centering
    \includegraphics[width=\linewidth]{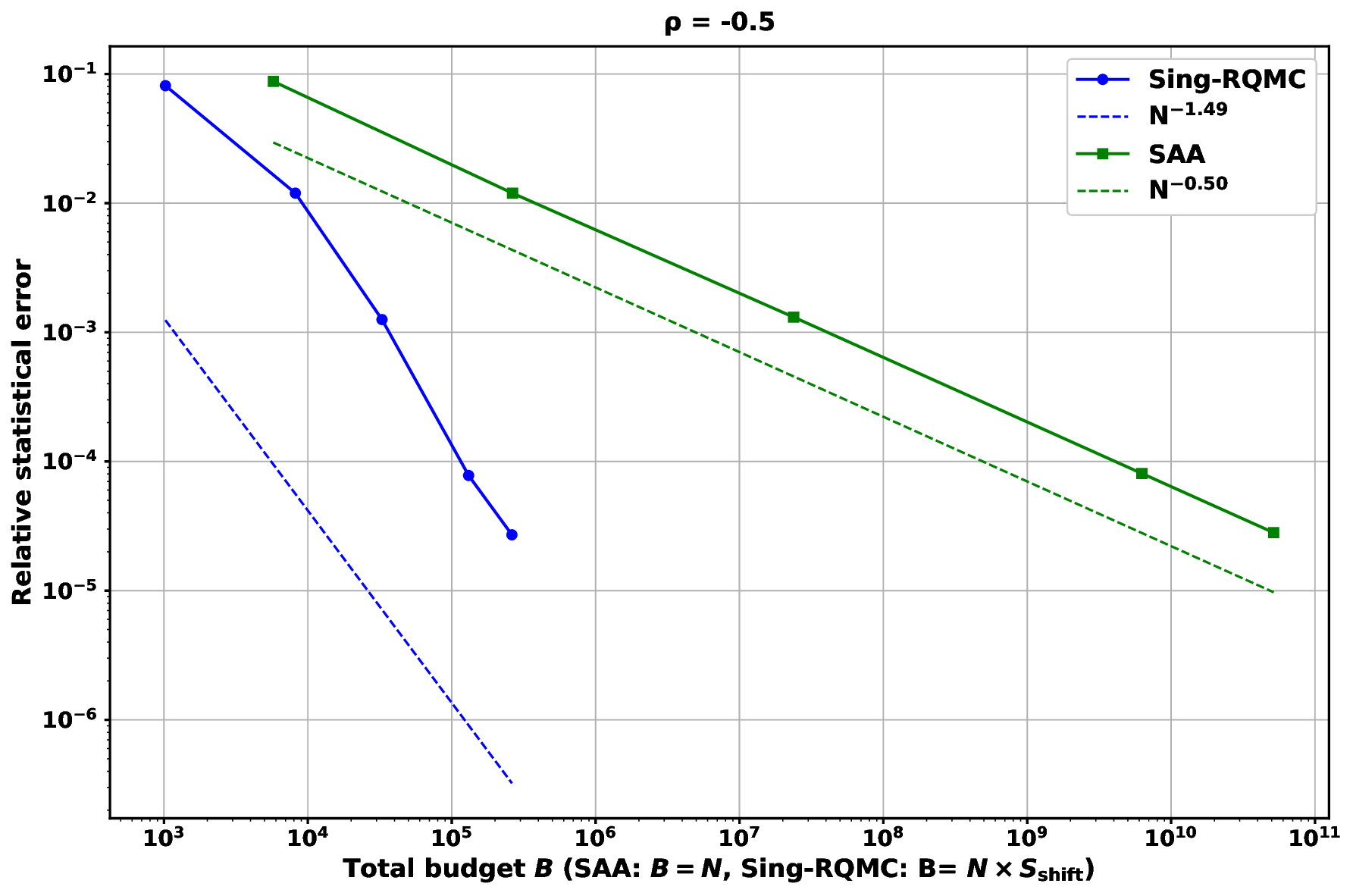}
  \end{subfigure}\hfill
  \begin{subfigure}{0.5\textwidth}
    \centering
    \includegraphics[width=\linewidth]{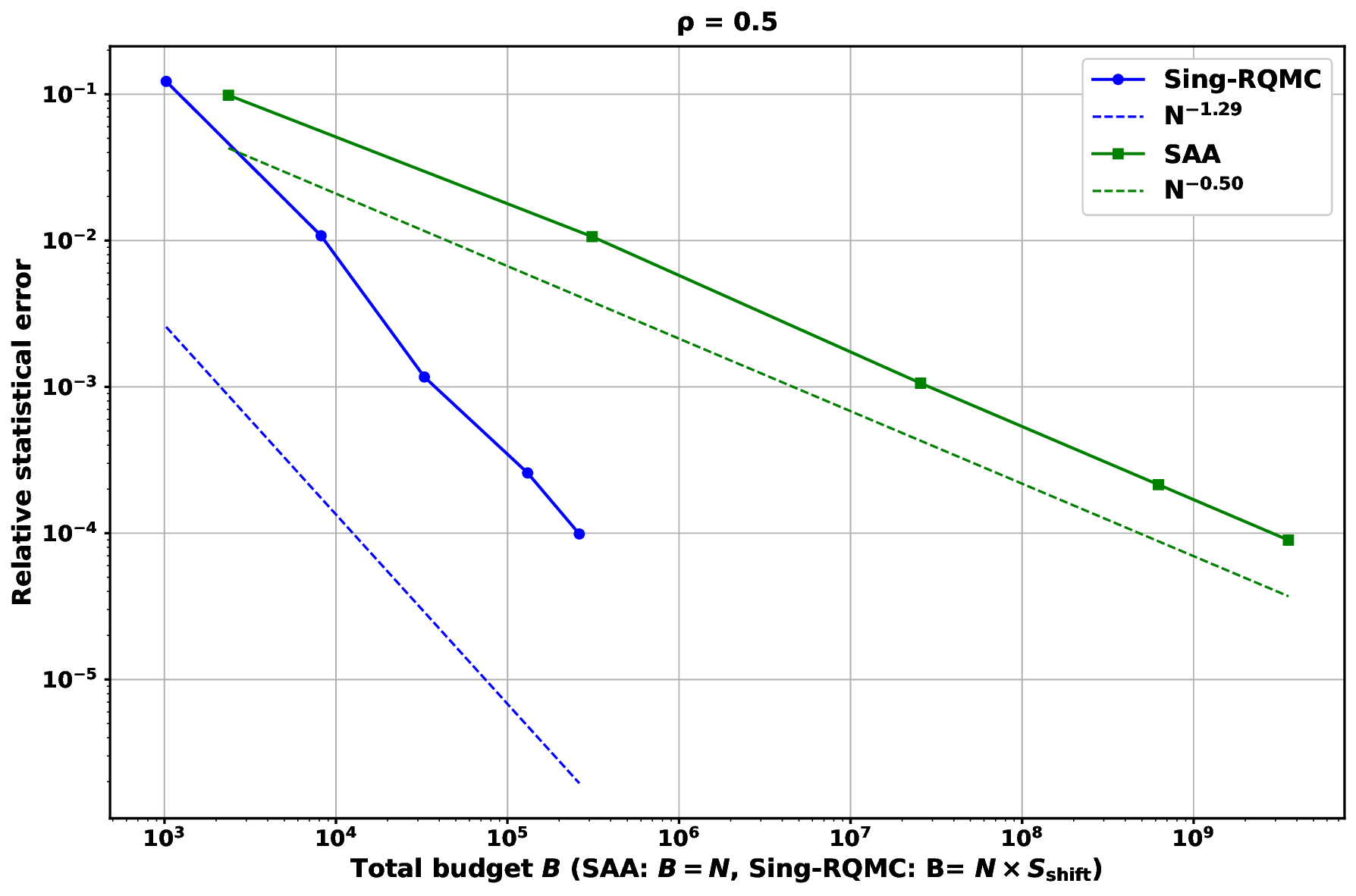}
  \end{subfigure}
  \caption{\small Exponential loss with a two-dimensional Gaussian loss vector: relative statistical error $\varepsilon_{\mathrm{stat,rel}}$ versus total sampling budget $B$ for SAA and single-level Fourier--RQMC, with $\rho=-0.5$ (left) and $\rho=0.5$ (right). Here $B_{\mathrm{SAA}}=N$ and $B_{\mathrm{RQMC}}=N S_{\mathrm{shift}}$.}
  \label{fig:err_sample_expo}
\end{figure}
To further assess the computational performance of the Fourier–RQMC methods, we compare the average  computational time required to reach a prescribed relative total error $\varepsilon_{\mathrm{rel}} $ against SAA and stochastic approximation (SA). For SA, we follow \cite{kaakai_estimation_2024} and employ the constrained Robbins–Monro scheme with parameters $(c,t,\gamma)=(2,10,0.7)$, combined with Polyak–Ruppert averaging to estimate the statistical error $\varepsilon_{\mathrm{stat}}$. The corresponding relative error $\varepsilon_{\mathrm{stat,rel}}$ is computed via \eqref{eq:relative_stat_err}. Figure~\ref{fig:error_vs_time_expo} reports the average runtime required to achieve a given $\varepsilon_{\mathrm{rel}}$, using the splitting strategy $\varepsilon_{\mathrm{opt}}=\varepsilon_{\mathrm{stat}}=\varepsilon/2$. \footnote{\textcolor{black}{$\varepsilon_{\mathrm{stat}}$ is not enforced \emph{a priori}. We first run the optimizer with a coarse surrogate, targeting $\varepsilon_{\mathrm{opt}}\le \varepsilon/2$, to obtain a rough solution estimate. We then refine the sample sizes so that $\varepsilon_{\mathrm{stat}}\le \varepsilon/2$.}} For SA, since no explicit optimization error is available, we directly report the runtime to reach $\varepsilon_{\mathrm{rel}}$.

Across both correlation settings, single-level Fourier–RQMC consistently outperforms both SAA and SA in terms of numerical complexity, exhibiting improved complexity rates that closely match the theoretical analysis of Section~\ref{sec:num_analysis}. In particular, for relative tolerances of order $10^{-4}$, Fourier–RQMC attains the target accuracy with approximately $10^{4}$ times fewer samples than SAA and up to $10^{6}$ times fewer samples than SA. Finally, the multilevel variant provides only a limited additional improvement in this setting, which is consistent with the fact that only a small number of iterations (approximately $2\text{--}3$) are spent in the local convergence regime.
\begin{figure}[H]
  \centering
  \begin{subfigure}{0.5\textwidth}
    \centering
    \includegraphics[width=\linewidth]{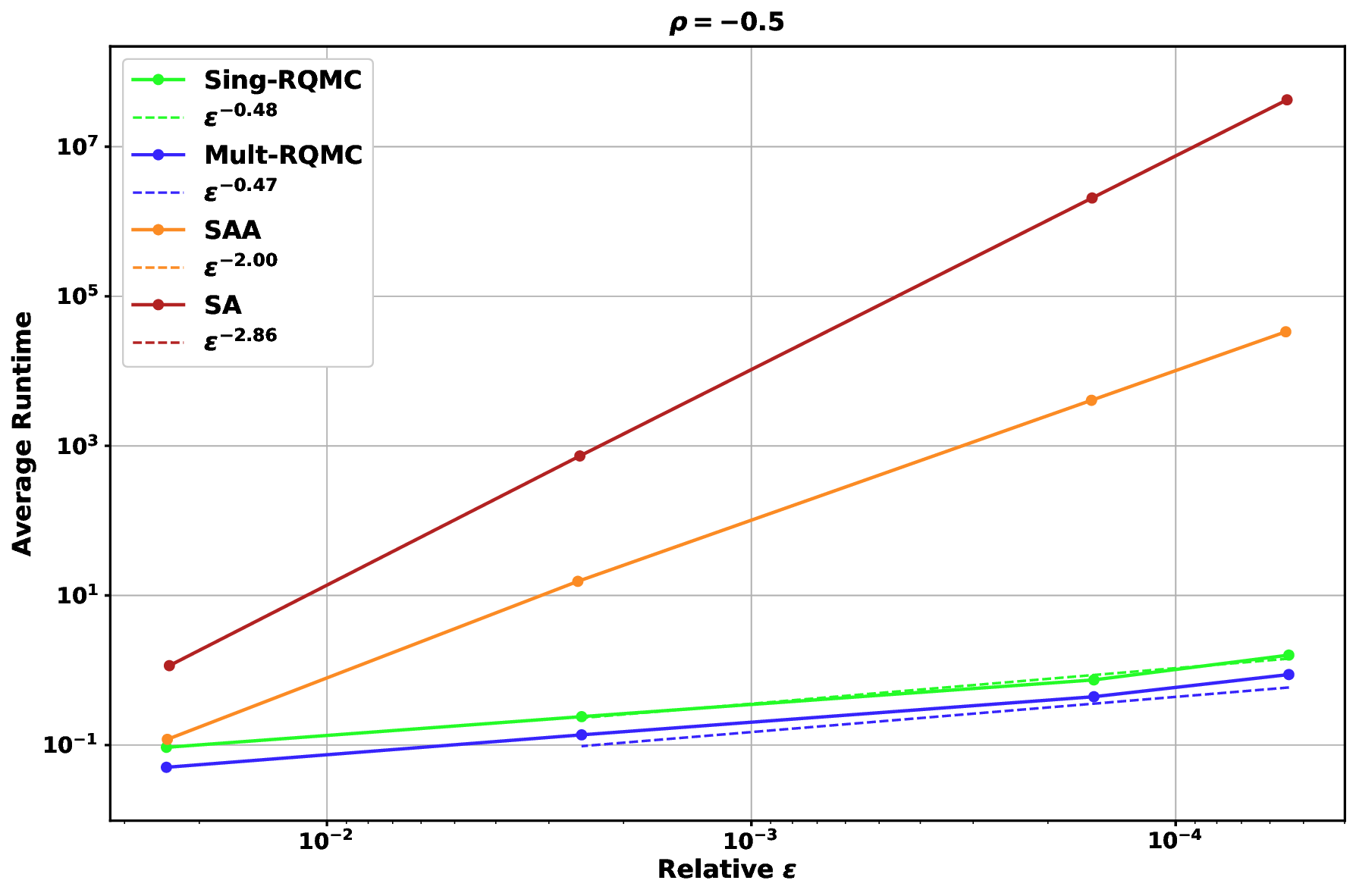}
  \end{subfigure}\hfill
  \begin{subfigure}{0.5\textwidth}
    \centering
    \includegraphics[width=\linewidth]{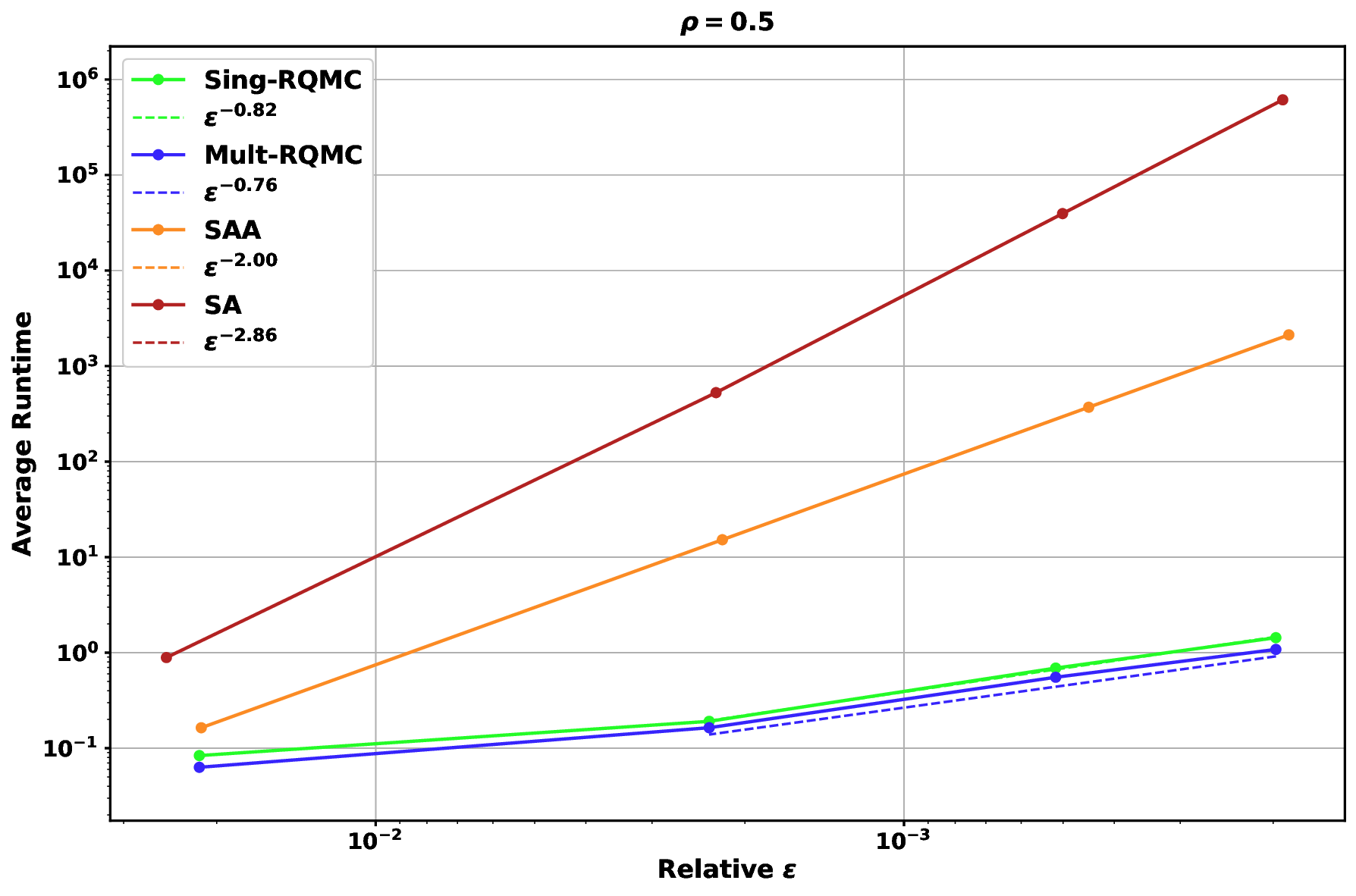}
  \end{subfigure}
  \caption{\small  Exponential loss with a two-dimensional Gaussian loss vector: average runtime (seconds) versus prescribed relative total tolerance $\varepsilon_{\mathrm{rel}}$ for $\rho=-0.5$ (left) and $\rho=0.5$ (right).}
  \label{fig:error_vs_time_expo}
\end{figure}
As shown in Figure~\ref{fig:error_vs_time_expo}, the SA method is consistently less efficient than the other approaches, in line with the findings of \cite{kaakai_estimation_2024}. We therefore omit SA as a baseline in the subsequent numerical experiments.

\subsection{QPC Loss with Ten-Dimensional Gaussian Loss Vector}\label{subsec:10D_gaussian_qcl}

We now consider the QPC loss defined in~\eqref{eq:multi_qcl} with $\alpha = 1$, applied to a 10-dimensional Gaussian loss vector as in \cite{armenti_multivariate_2018}, with mean $\boldsymbol\mu=(0,\dots,0)^\top$ and covariance matrix
\[
\boldsymbol \Sigma =
\begin{pmatrix}
2.11 & 0.37 & -0.42 & \cdots & -0.94 \\
0.37 & 1.78 & -0.45 & \cdots & -0.48 \\
\vdots & \vdots & \ddots & \vdots & \vdots \\
-0.94 & -0.48 & 0.45 & \cdots & 0.88
\end{pmatrix}.
\]
Since no closed-form solution is available for this loss setting, we compute a reference solution using SAA, as described at the beginning of Section~\ref{sec:num_exp_results}.

The right panel of Figure \ref{fig:error_vs_time_10D_MN_qcl} shows that single-level Fourier–RQMC remains up to two orders of magnitude more efficient than SAA for achieving a relative error of order $10^{-3}$. Moreover, the computational advantage of the multilevel estimator over the single-level Fourier–RQMC becomes more pronounced in this regime. For the prescribed tolerance $\varepsilon_{\mathrm{rel}}$, the optimizer typically requires approximately $8\text{--}12$ iterations and enters the local convergence regime after roughly $4\text{--}6$ iterations. Within this regime, the multilevel estimator achieves substantial variance reduction for the component integrands, leading to a lower per-iteration computational cost (see the left panel of  Figure~\ref{fig:error_vs_time_10D_MN_qcl}) and an improved asymptotic complexity compared with the single-level method, in agreement with the analysis in Section~\ref{subsec:computational_complex_Fou_RQMC}. Nevertheless, the overall computational gain of the multilevel approach remains limited, since independent random shifts must be generated at each level, whereas the single-level method draws random shifts only once at initialization (see Algorithms~\ref{alg:Single-level RQMC} and~\ref{alg:RQMC_Fou_multi}).

\begin{figure}
    \centering
    \begin{subfigure}{0.5\textwidth}
        \centering
        \includegraphics[width=\linewidth]{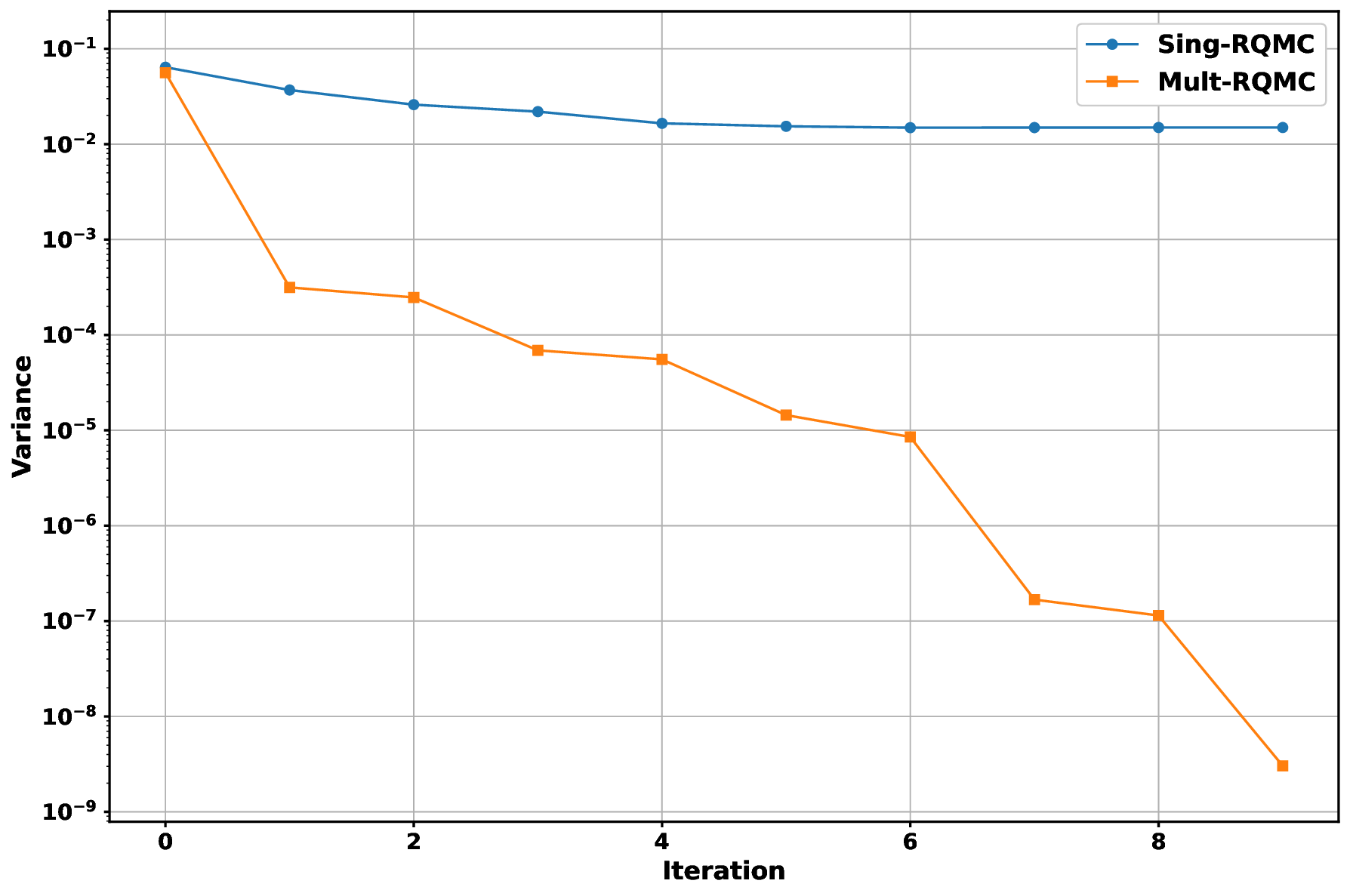}
    \end{subfigure}\hfill
    \begin{subfigure}{0.5\textwidth}
        \centering
        \includegraphics[width=\linewidth]{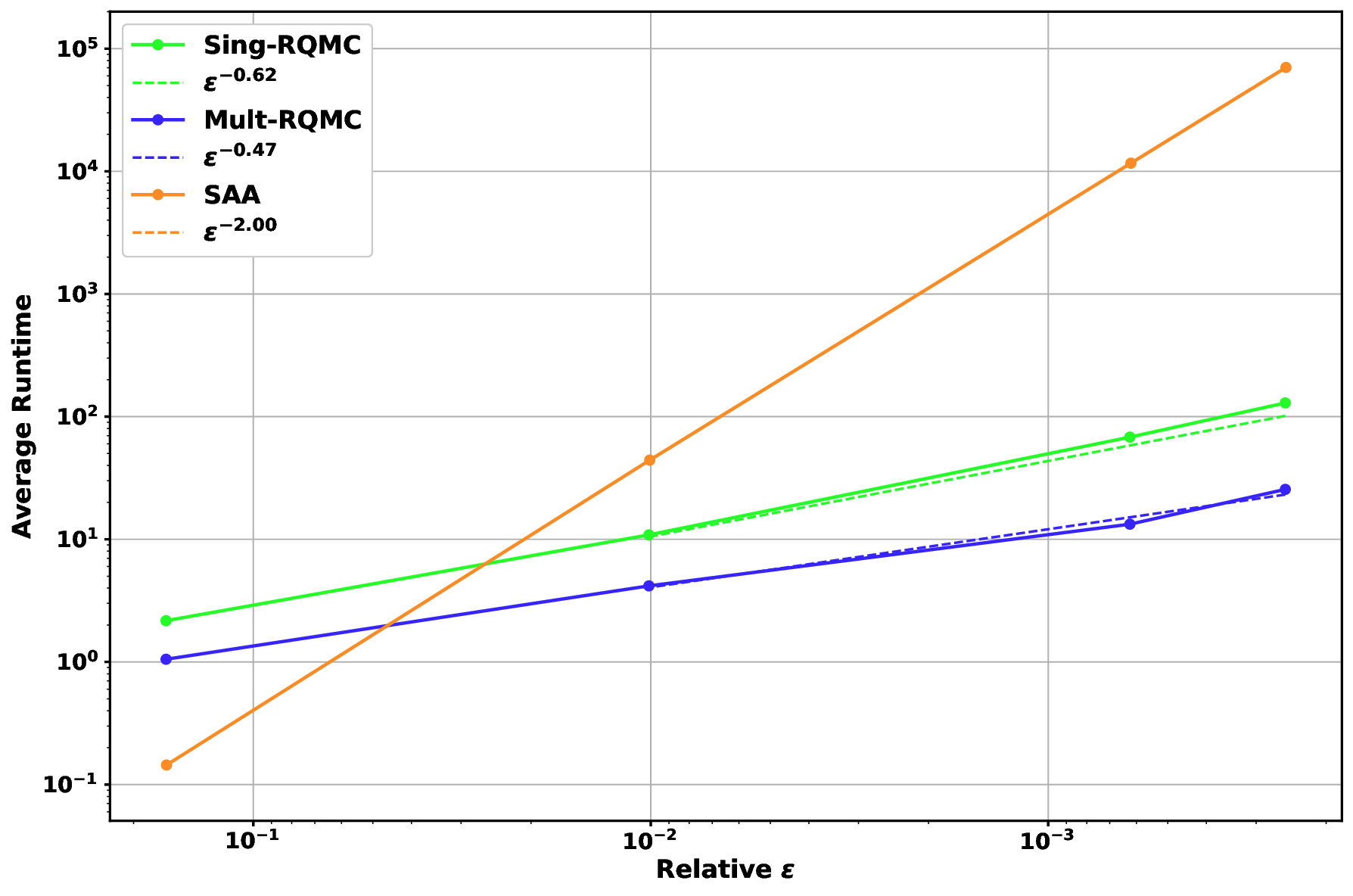}
    \end{subfigure}
    \caption{\small QPC loss with a ten-dimensional Gaussian loss vector: 
(Left):  Estimated variance across shifts of the  estimator of $m_5$ over optimization iterations.
 (Right): Average runtime in seconds  w.r.t. prescribed relative total tolerance $\varepsilon_{\mathrm{rel}}$ .
}
\label{fig:error_vs_time_10D_MN_qcl}
\end{figure}

\subsection{QPC Loss with Three-Dimensional NIG Loss Vector}\label{subsec:3D_NIG_qcl}
We next assess the performance of the Fourier–RQMC methods for the same QPC loss under a heavier-tailed distribution for the loss vector $\mathbf{X}$, namely a three-dimensional NIG distribution. Specifically, we consider $\mathbf{X} \sim \mathcal{NIG}(\alpha,\boldsymbol\beta,\delta,\boldsymbol\mu,\boldsymbol\Gamma)$, with parameters from \cite{kaakai_multivariate_2022}\[
\alpha = 365.78,\;
\delta = 0.373,\;
\boldsymbol\theta = (2,2,2)^\top,\; 
\boldsymbol\mu = (0.00084,0.00024,0.00055)^\top,\; 
\boldsymbol\beta = (-64.27,41.45,7.35)^\top
\]
and covariance matrix
\[
\boldsymbol\Gamma =
\begin{pmatrix}
2.338 & 1.796 & 2.080 \\
1.796 & 2.327 & 2.088 \\
2.080 & 2.088 & 2.555
\end{pmatrix}.
\]
In this NIG setting, boundary-induced oscillations in the transformed integrands become more pronounced, as discussed in Section~\ref{subsubsec:domain_trans_QMC}. As a consequence, the choice of the scaling parameter $c$ for single-level Fourier–RQMC becomes critical: inappropriate values can significantly degrade both numerical stability and optimization performance. In contrast, the multilevel Fourier–RQMC estimator partially mitigates these oscillatory effects by operating on the difference integrands, which tend to cancel boundary oscillations across successive optimization iterates. This behavior is illustrated in Figure~\ref{fig:contraction_grad_qcl_3D_nig}, where the difference integrands exhibit markedly reduced oscillations compared with the corresponding single-level integrands, leading to a more favorable variance structure and improved numerical robustness, implying more significant  computation gain of the multilevel method over the single-level one compared to the previous examples.

Compared with the Gaussian examples in Sections \ref{subsec:expo_loss}-\ref{subsec:10D_gaussian_qcl}, the component integrands under the NIG model are less smooth. Accordingly, the numerical complexity rate deteriorates   relative to the Gaussian case but remain close to the asymptotic rate predicted in \eqref{eq:stat_error_RQMC_rate}, namely $r \approx 1$. Despite this reduced smoothness, the right panel of  Figure~\ref{fig:error_vs_time_mnig_qcl}  shows that Fourier–RQMC methods continue to substantially outperform SAA in terms of computational complexity. For instance, at a relative tolerance of order $10^{-1}$, Fourier–RQMC methods achieves the prescribed accuracy with a computational cost approximately $10^{5}$ times smaller than that of SAA.

This significant gain can be attributed to two complementary effects. First, the variance  of the estimators is evaluated in Fourier space using RQMC, which appears to be significantly less sensitive to rare-event contributions of the gradient terms than MC estimators in physical space when using SAA, leading to smaller  constants. Second, the method requires the numerical inversion of the Hessian matrix $\nabla_{\mathbf z}^2 \mathcal L \paren{\hat{\mathcal{L}}_{\nabla_{\mathbf z}^2}^{\mathrm{Fou}}}$ at the final iteration $J$. When this Hessian is estimated in physical space using MC, it may become poorly conditioned, causing the associated error constants to grow excessively. In contrast, within the Fourier–RQMC framework, the Hessian is evaluated in Fourier space, where the transformed integrands exhibit increased smoothness. This additional regularity results in improved numerical conditioning of the Hessian matrix and, consequently, more stable error constants, as clearly observed in the left panel of Figure~\ref{fig:error_vs_time_mnig_qcl}.
\begin{figure}
    \centering
    \begin{subfigure}{0.5\textwidth}
        \centering
        \includegraphics[width=\linewidth]{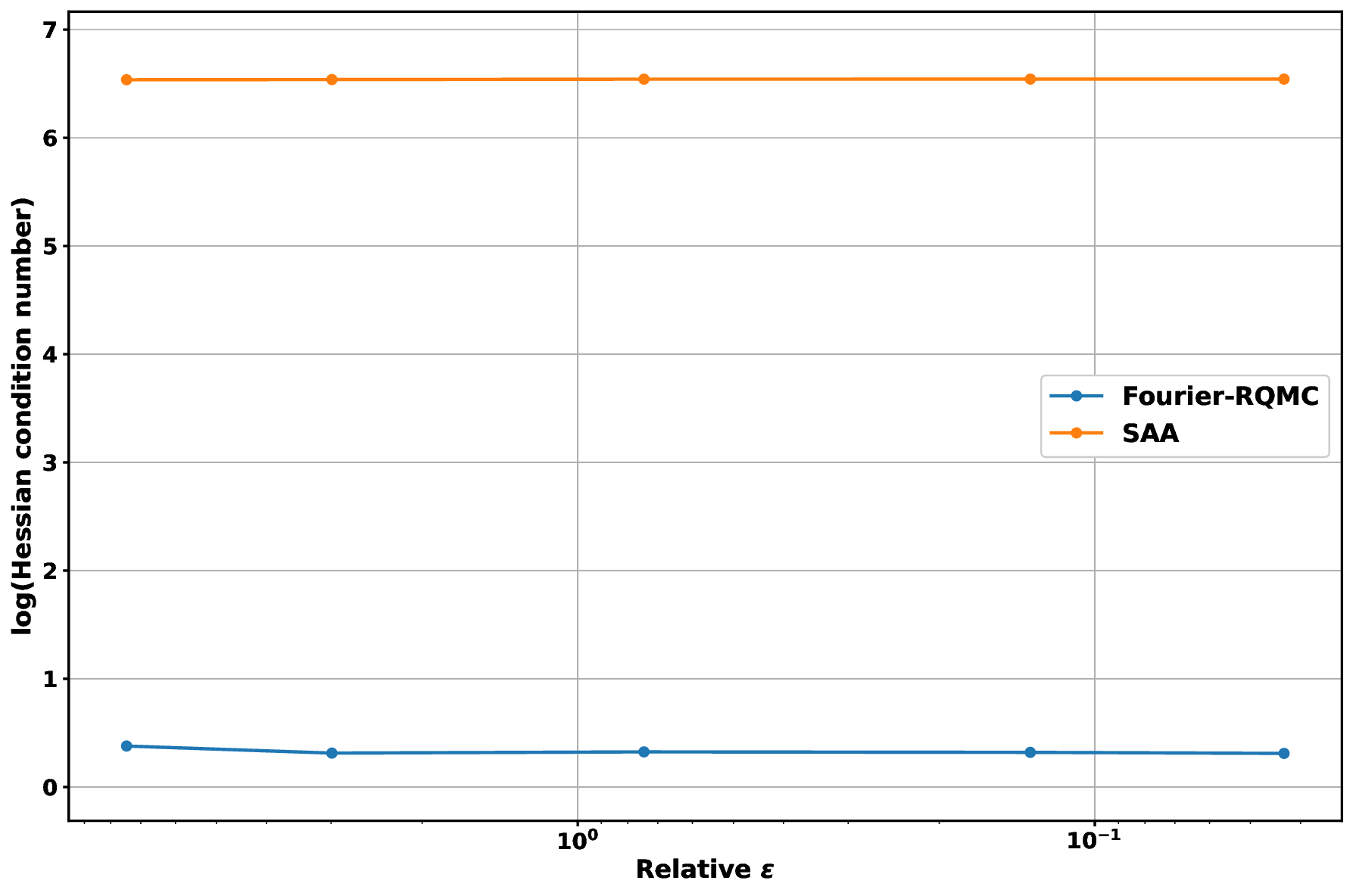}
    \end{subfigure}\hfill
    \begin{subfigure}{0.5\textwidth}
        \centering
        \includegraphics[width=\linewidth]{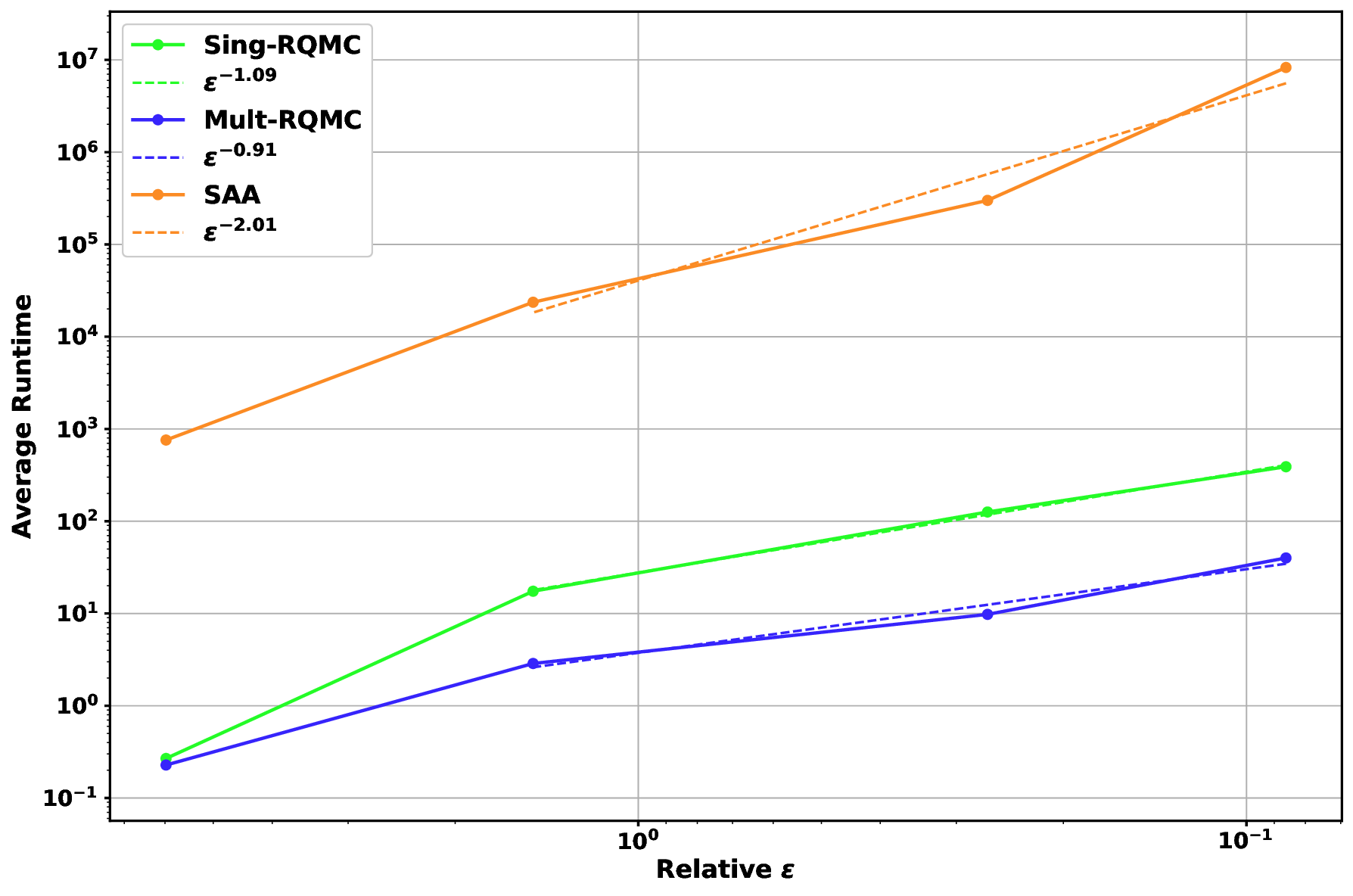}
    \end{subfigure}
    \caption{QPC loss with a three-dimensional NIG loss vector: (left) spectral condition number $\kappa(\nabla_z^2\widehat{\mathcal L}(z^{(J)}))$ at the final iterate $J$; (right) average runtime (seconds) versus prescribed relative total tolerance $\varepsilon_{\mathrm{rel}}$.}
\label{fig:error_vs_time_mnig_qcl}
\end{figure}
\begin{center}
\begin{center}

\includegraphics[width=0.75\textwidth]{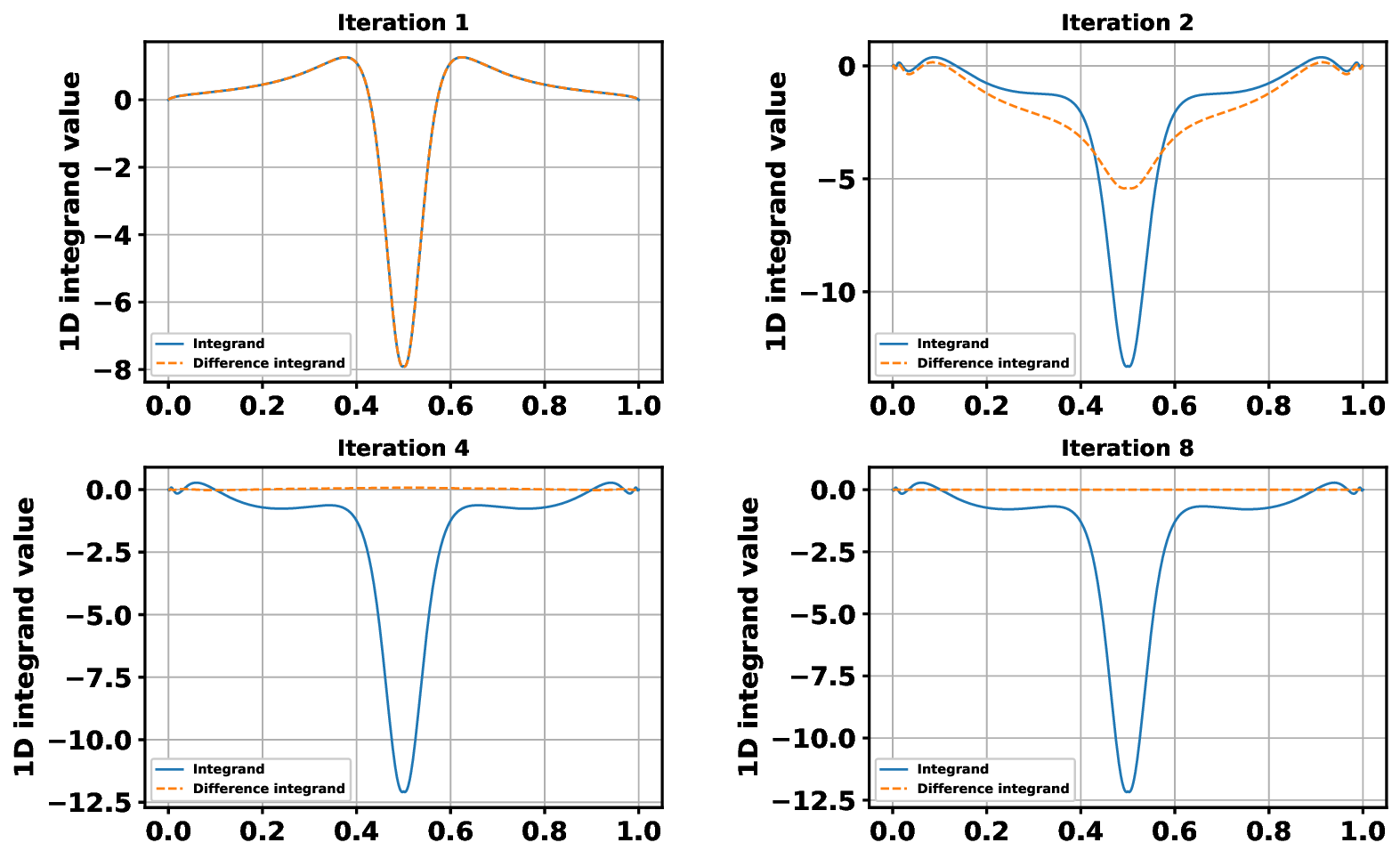}

\small (a) $c=4$

\vspace{1em}

\includegraphics[width=0.75\textwidth]{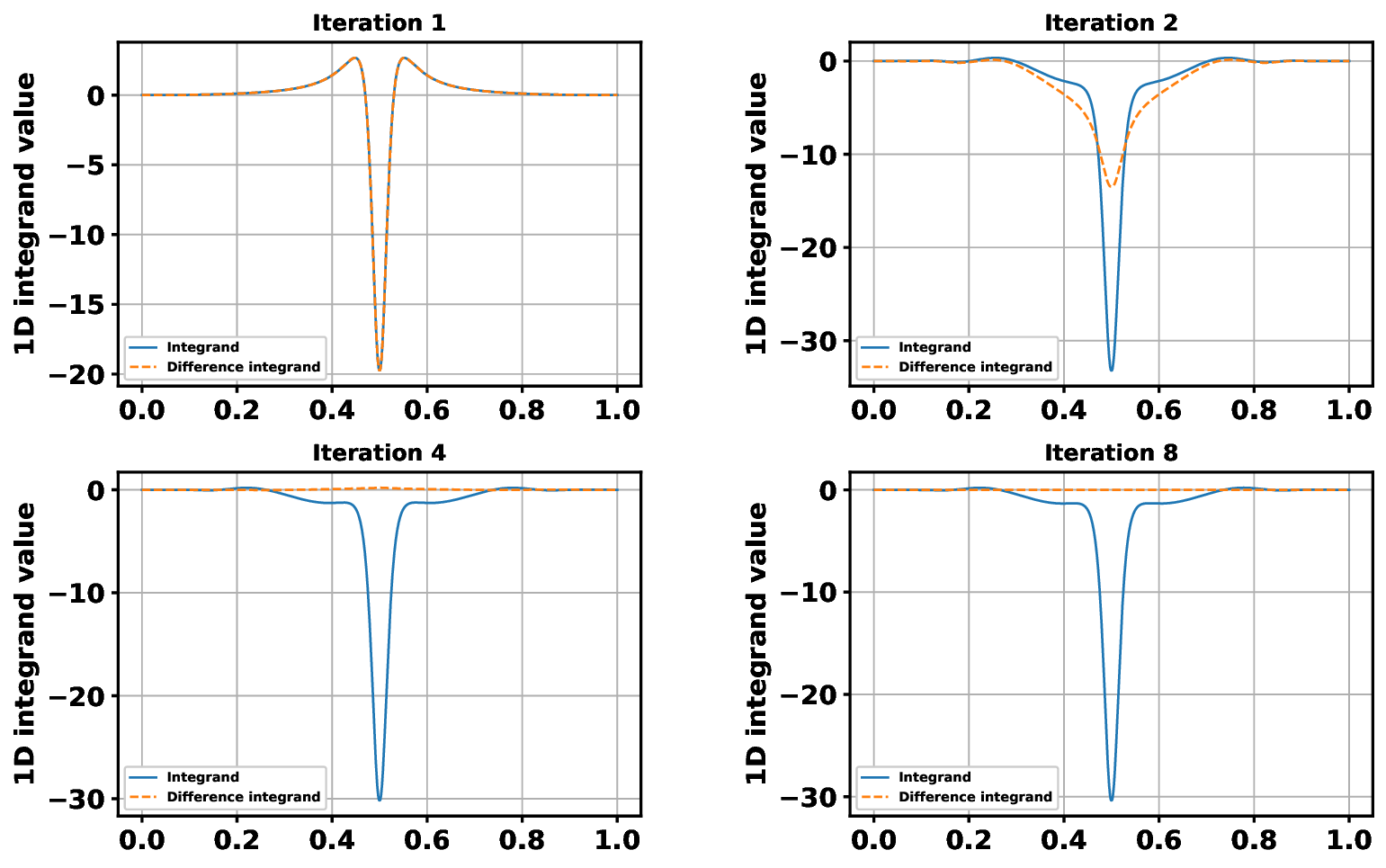}

\small (b) $c=10$

\captionof{figure}{\small QPC loss with a three-dimensional NIG loss vector: Transformed integrand component $\tilde h_{1,2}^{(1)}$ (solid) and the corresponding difference integrand arising in the multilevel construction (dashed) across successive optimization iterations with different scaling $c$.}
\label{fig:contraction_grad_qcl_3D_nig}
\end{center}

\end{center}




\bibliographystyle{plain}
\bibliography{references_mult_MSRM} 
\appendix
\section{Supplementary results for Sections \ref{sec:num_analysis}}\label{appendix:sup_err_analysis}
\subsection{Additional Assumptions for the Error Analysis}\label{appendix:additional_opt_err}
\begin{assumption}[Regularity conditions for the exact problem]\label{assump:convergence_SLSQP}
The point $\mathbf{m^*}$ is a (local) solution of equation \eqref{eq:deter_optimization_MSRM} at which the following conditions hold:
\begin{enumerate}[label=(C\arabic*), ref=C\arabic*]
    \item \label{ass:C1} For every $\mathbf{m_0} \in \mathbb{R}^d$, the mapping $\mathbf{m} \mapsto \nabla\ell(\mathbf{X-m})$ is differentiable at $\mathbf{m}_0$ a.s.

    \item\label{ass:C2} The strong second-order sufficient conditions (SSOSC) hold at the point $\mathbf{m}^*$. Let the tangent space at $\mathbf{m}^*$
\begin{equation*}
    \mathcal{K}(\mathbf{m}^*)
:= \curly{
\mathbf{a} \in \mathbb{R}^d \ :\ 
\nabla g(\mathbf{m}^*)^{\!\top} \mathbf{a} = 0 
}.
\end{equation*}
Then $\mathbf{a}^\top \nabla^2 g(\mathbf{m}^*) \mathbf{a} < 0,\quad  \forall  \mathbf{a}\in \mathcal{K}(\mathbf{m}^*) \setminus \left\{0\right\}.$ \footnote{$\nabla^2 \mathcal{L}(\mathbf{z}^*) = -\lambda^* \nabla^2 g(\mathbf{m}^*)$ as in \eqref{eq:hessian_KKT_system} with $\lambda^* >0$, so the normal  SSOSC in \cite{boggs_sequential_1995} become $\mathbf{a}^\top \nabla^2 g(\mathbf{m}^*) \mathbf{a} < 0$ }
\end{enumerate}
\label{assump:regularity_exact_KKT}
\end{assumption}
\begin{assumption}[Regularity conditions for the Fourier–RQMC problem]\
\begin{enumerate}[label=(C\arabic*), ref=C\arabic*]
    \setcounter{enumi}{2}
    \item \label{ass:C4}
The SSOSC holds at the point ${\mathbf{m}}_{N,S_{\mathrm{shift}}}^{\mathrm{RQMC},*}$ \textcolor{black}{a.s.} \footnote{a.s. in Assumption \ref{assump:regularity_Fou_RQMC} will be understood as for almost every realization of the RQMC shifts $S$.}Define the tangent space at ${\mathbf{m}}_{N,S_{\mathrm{shift}}}^{\mathrm{RQMC},*}$
\begin{equation*}
\mathcal{K}_{N,S_\mathrm{shift}}\paren{{\mathbf{m}}_{N,S_{\mathrm{shift}}}^{\mathrm{RQMC},*}}
:= \curly{
\mathbf{a} \in \mathbb{R}^d \ :\ 
I_{N,S_{\mathrm{shift}}}^{\mathrm{RQMC}}\brac{\tilde h^{(1)}\paren{{\cdot;\mathbf{m}}_{N,S_{\mathrm{shift}}}^{\mathrm{RQMC},*}}}^\top \mathbf{a} = 0 
}.
\end{equation*}
Then $\mathbf{a}^\top I_{N,S_{\mathrm{shift}}}^{\mathrm{RQMC}}\brac{\tilde h^{(2)}\paren{{\cdot;\mathbf{m}}_{N,S_{\mathrm{shift}}}^{\mathrm{RQMC},*}}}\mathbf{a} < 0,\quad  \forall  \mathbf{a}\in \mathcal{K}_{N,S_\mathrm{shift}}\paren{{\mathbf{m}}_{N,S_{\mathrm{shift}}}^{\mathrm{RQMC},*}}\setminus \left\{0\right\}.$
    \item \label{ass:C5}
    The Hessian matrix $I_{N,S_{\mathrm{shift}}}^{\mathrm{RQMC}}\brac{\tilde h^{(2)}\paren{{\cdot;\mathbf{m}}_{N,S_{\mathrm{shift}}}^{\mathrm{RQMC},*}}}$,
    is invertible a.s.
    \item \label{ass:C6}
For each iteration $j$, the BFGS approximation matrix $\mathbf{B}_{N,S_{\mathrm{shift}}}^{(\mathrm{RQMC},j)}\paren{{\cdot;\mathbf{m}}_{N,S_{\mathrm{shift}}}^{(\mathrm{RQMC},j)}}$
is $\mu_1$-strongly convex a.s. on the tangent space $\mathcal{K}_{N,S_\mathrm{shift}}\paren{{\mathbf{m}}_{N,S_{\mathrm{shift}}}^{(\mathrm{RQMC},j)}}$, i.e.,
\begin{equation*}
\mu_1 \,\boldsymbol I \;\preceq\; \mathbf{B}_{N,S_{\mathrm{shift}}}^{(\mathrm{RQMC},j)}\paren{{\cdot;\mathbf{m}}_{N,S_{\mathrm{shift}}}^{(\mathrm{RQMC},j)}} ,
\qquad  0<\mu_1  < \infty,
\label{eq:BFGS_mu_L_bounds}
\end{equation*}

\item \label{ass:C7}
For each iteration $j$, the approximation matrix
$\mathbf{B}_{N,S_{\mathrm{shift}}}^{(\mathrm{RQMC},j)}\paren{{\cdot;\mathbf{m}}_{N,S_{\mathrm{shift}}}^{(\mathrm{RQMC},j)}}$
is invertible a.s., and both it and its inverse are uniformly bounded a.s., i.e.,
\begin{equation*}
\Bigl\|\mathbf{B}_{N,S_{\mathrm{shift}}}^{(\mathrm{RQMC},j)}\paren{{\cdot;\mathbf{m}}_{N,S_{\mathrm{shift}}}^{(\mathrm{RQMC},j)}}\Bigr\|
\leq B_1,
\qquad
\Bigl\|\Bigl(\mathbf{B}_{N,S_{\mathrm{shift}}}^{(\mathrm{RQMC},j)}\paren{{\cdot;\mathbf{m}}_{N,S_{\mathrm{shift}}}^{(\mathrm{RQMC},j)}}\Bigr)^{-1}\Bigr\|
\leq B_2,
\label{eq:BFGS_bounds}
\end{equation*}
for some constants $B_1>0$ and $B_2>0$.
    \item \label{ass:C8}
    Let $P_j$ denote the projection matrix onto the tangent space $\mathcal{K}_{N,S_\mathrm{shift}}\paren{{\mathbf{m}}_{N,S_{\mathrm{shift}}}^{(\mathrm{RQMC},j)}}$.
    The approximation matrix $\mathbf{B}_{N,S_{\mathrm{shift}}}^{(\mathrm{RQMC},j)}\paren{{\cdot;\mathbf{m}}_{N,S_{\mathrm{shift}}}^{(\mathrm{RQMC},j)}}$ satisfies, a.s.,
    \begin{equation*}
        \lim_{j \to \infty}
        \frac{
        \norm{
        P_j \paren{
        \mathbf{B}_{N,S_{\mathrm{shift}}}^{(\mathrm{RQMC},j)}\paren{{\cdot;\mathbf{m}}_{N,S_{\mathrm{shift}}}^{(\mathrm{RQMC},j)}}
        - I_{N,S_{\mathrm{shift}}}^{\mathrm{RQMC}}\brac{\tilde h^{(2)} \paren{{\cdot;\mathbf{m}}_{N,S_{\mathrm{shift}}}^{\mathrm{RQMC},*}}}
        \paren{{\mathbf{m}}_{j+1,N,S_{\mathrm{shift}}}^{\mathrm{RQMC}}- {\mathbf{m}}_{N,S_{\mathrm{shift}}}^{(\mathrm{RQMC},j)}}
        }
        }
        }{
        \norm{{\mathbf{m}}_{j+1,N,S_{\mathrm{shift}}}^{\mathrm{RQMC}}- {\mathbf{m}}_{N,S_{\mathrm{shift}}}^{(\mathrm{RQMC},j)}}
        }
        = 0 .
    \end{equation*}
\end{enumerate}
\label{assump:regularity_Fou_RQMC}
\end{assumption}

\subsection{Proof for Lemma \ref{lem:uniform_conv_RQMC_Fou}}\label{Appendix:proof_uniform_conv_RQMC}
The argument is based on the uniform convergence proof for unconstrained SAA problems in
\cite[Proposition~7]{shapiro_monte_2003}. We adapt that reasoning by replacing the SAA estimators with our Fourier-RQMC
estimators and present the full proof for completeness. Recall that our Fourier-RQMC estimator with $\bar S$ independent shifts can be written as
\begin{equation}
I_{N,\bar S}^{\mathrm{RQMC}}\!\left[\tilde h_{k,p}^{(\nu)}(\cdot;\mathbf m_{k,p})\right]
= \frac{1}{\bar S}\sum_{s=1}^{\bar S}\frac{1}{N}\sum_{n=1}^N
\tilde h_{k,p}^{(\nu)}\!\left(\mathbf v_n^{(s)},\mathbf m_{k,p}\right),
\label{eq:RQMC_estimator_recall}
\end{equation}
and the corresponding target quantity $\hat g_{k,p}^{(\nu),\mathrm{Fou}}(\mathbf m_{k,p})$ satisfies Assumption \ref{ass:A7} is
\begin{equation}
\hat g_{k,p}^{(\nu),\mathrm{Fou}}(\mathbf m_{k,p})
= \mathbb{E}_\mathbf{V}\!\left[\tilde h_{k,p}^{(\nu)}(\mathbf V,\mathbf m_{k,p})\right],
\label{eq:target_integral_recall}
\end{equation}
where $\mathbf V$ denotes a generic random vector  distributed on $[0,1]^k$.

Fix $\bar{\mathbf m}_{k,p}\in M_{k,p}$. For $\eta>0$,
\[
U(\bar{\mathbf m}_{k,p},\eta)
:= \bigl\{\mathbf m_{k,p}\in M_{k,p}:\ \|\mathbf m_{k,p}-\bar{\mathbf m}_{k,p}\|\le \eta\bigr\}.
\]
Define 
$
\delta_{\bar{\mathbf m}_{k,p},\eta}^{(\nu)}(\mathbf v)
:= \sup_{\mathbf m_{k,p}\in U(\bar{\mathbf m}_{k,p},\eta)}
\norm{\tilde h_{k,p}^{(\nu)}(\mathbf v,\mathbf m_{k,p})-\tilde h_{k,p}^{(\nu)}(\mathbf v,\bar{\mathbf m}_{k,p})}.
$
From \eqref{eq:transform_integrand_v}, we have that $\tilde h_{k,p}^{(\nu)}(\mathbf v,\cdot)$ is continuous with $\mathbf m_{k,p}$ (they only depend on $\mathbf{m}_{k,p}$ through $\exp(.)$ part) for each fixed $\mathbf v$, we have
\begin{equation}
\delta_{\bar{\mathbf m}_{k,p},\eta}^{(\nu)}\paren{\mathbf{v}}\ \xrightarrow[]{}\ 0,
\qquad \text{as }\eta\downarrow 0,\ \ \forall \mathbf v\in[0,1]^k.
\label{eq:delta_pointwise}
\end{equation}
Moreover, from the representation of $\tilde h_{k,p}^{(\nu)}$ in~\eqref{eq:transform_integrand_v}, for $\mathbf m_{k,p}\in M_{k,p}$ along with the admissible damping vector $\mathbf K_{k,p}^{(\nu)}$ from  Assumption \ref{ass:A7}, we also have 
\[\sup_{\mathbf m_{k,p}\in M_{k,p}} \exp\!\langle \mathbf K_{k,p}^{(\nu)},\mathbf m_{k,p}\rangle \le C_M\] for some finite constant $C_M$. Consequently, there exists an integrable envelope $H$ such that $\norm{\tilde h_{k,p}^{(\nu)}(\mathbf v;\mathbf m_{k,p})}\le H(\mathbf v)$ for all $\mathbf m_{k,p}$ in a neighborhood of $\bar{\mathbf m}_{k,p}$. Hence,
\begin{equation*}
0\le \delta_{\bar{\mathbf m}_{k,p},\eta}^{(\nu)}(\mathbf v)\le 2H(\mathbf v),
\qquad \forall \mathbf v\in[0,1]^k.
\end{equation*}
Therefore, by the Dominated Convergence Theorem,
\begin{equation}
\lim_{\eta \to 0} \mathbb{E}_\mathbf{V}\!\left[\delta_{\bar{\mathbf m}_{k,p},\eta}^{(\nu)}\paren{\mathbf{V}}\right]
 = \mathbb{E}_\mathbf{V}\brac{\lim_{\eta \to 0}\delta_{\bar{\mathbf m}_{k,p},\eta}^{(\nu)}\paren{\mathbf{V}}}
\label{eq:rho_to_zero}
\end{equation}
For any $\mathbf m_{k,p}\in U(\bar{\mathbf m}_{k,p},\eta)$, by \eqref{eq:RQMC_estimator_recall},
\begin{equation}
\norm{
I_{N,\bar S}^{\mathrm{RQMC}}\!\left[\tilde h_{k,p}^{(\nu)}(\cdot;\mathbf m_{k,p})\right]
-
I_{N,\bar S}^{\mathrm{RQMC}}\!\left[\tilde h_{k,p}^{(\nu)}(\cdot;\bar{\mathbf m}_{k,p})\right]
}
\le
I_{N,\bar S}^{\mathrm{RQMC}}\!\left[\delta_{\bar{\mathbf m}_{k,p},\eta}^{(\nu)}\right].
\label{eq:local_fluctuation_bound}
\end{equation}
By Assumption~\ref{ass:A10}, from \cite[Theorem 4.2]{owen_strong_2021}, we have the SLLN for the nested uniform scrambling Sobol sequence (applied to the integrable function
$\delta_{\bar{\mathbf m}_{k,p},\eta}^{(\nu)}$), we have
\begin{equation}
I_{N,\bar S}^{\mathrm{RQMC}}\!\left[\delta_{\bar{\mathbf m}_{k,p},\eta}^{(\nu)}\right]
\xrightarrow[]{\mathrm{a.s.}}
\mathbb{E}_{\mathbf{V}}\!\left[\delta_{\bar{\mathbf m}_{k,p},\eta}^{(\nu)}\paren{\mathbf{V}}\right],
\qquad N\to\infty.
\label{eq:SLLN_delta}
\end{equation}
Combining \eqref{eq:delta_pointwise}--\eqref{eq:SLLN_delta}, for any $\epsilon>0$, there exists small enough
$\eta>0$ and  sufficient large $N_0$  such that for all $N\ge N_0$,
\begin{equation}
\sup_{\mathbf m_{k,p}\in U(\bar{\mathbf m}_{k,p},\eta)}
\norm{
I_{N,\bar S}^{\mathrm{RQMC}}\!\left[\tilde h_{k,p}^{(\nu)}(\cdot;\mathbf m_{k,p})\right]
-
I_{N,\bar S}^{\mathrm{RQMC}}\!\left[\tilde h_{k,p}^{(\nu)}(\cdot;\bar{\mathbf m}_{k,p})\right]
}
\le \epsilon
\quad\text{a.s.}
\label{eq:local_equicontinuity_rqmc}
\end{equation}
Since $M_{k,p}$ is compact, for the above $\eta>0$ there exist points
$\mathbf m_{k,p}^{(1)},\dots,\mathbf m_{k,p}^{(E)}\in M_{k,p}$ such that
\(
M_{k,p}\subset \bigcup_{i=1}^E U\paren{\mathbf m_{k,p}^{(i)},\eta}.
\)
Fix such a finite cover; for any $\mathbf m_{k,p}\in M_{k,p}$, choose an index $i$ such that
$\mathbf m_{k,p}\in U\paren{\mathbf m_{k,p}^{(i)},\eta}$. Then, by the triangle inequality,
\begin{equation}
 \begin{aligned}
\norm{
I_{N,\bar S}^{\mathrm{RQMC}}\!\left[\tilde h_{k,p}^{(\nu)}(\cdot;\mathbf m_{k,p})\right]
- \hat g_{k,p}^{(\nu),\mathrm{Fou}}(\mathbf m_{k,p})
}
&\le
\underbrace{\norm{
I_{N,\bar S}^{\mathrm{RQMC}}\!\left[\tilde h_{k,p}^{(\nu)}(\cdot;\mathbf m_{k,p})\right]
-
I_{N,\bar S}^{\mathrm{RQMC}}\!\left[\tilde h_{k,p}^{(\nu)}\!\paren{\cdot;\mathbf{m}_{k,p}^{(i)}}\right]
}}_{(\mathrm{I})}
\\[4pt]
&+
\underbrace{\norm{
I_{N,\bar S}^{\mathrm{RQMC}}\!\left[\tilde h_{k,p}^{(\nu)}\!\paren{\cdot;\mathbf{m}_{k,p}^{(i)}}\right]
- \hat g_{k,p}^{(\nu),\mathrm{Fou}}\!\paren{\mathbf{m}_{k,p}^{(i)}}
}}_{(\mathrm{II})}
\\[4pt]
&+
\underbrace{\norm{
\hat g_{k,p}^{(\nu),\mathrm{Fou}}\!\paren{\mathbf{m}_{k,p}^{(i)}}
- \hat g_{k,p}^{(\nu),\mathrm{Fou}}(\mathbf m_{k,p})
}}_{(\mathrm{III})}.
\end{aligned}
\label{eq:tri_decomp}
\end{equation}

\begin{itemize}
\item Term \((\mathrm{I})\) is controlled uniformly by \eqref{eq:local_equicontinuity_rqmc}:
\[
\sup_{\mathbf m_{k,p}\in M_{k,p}}(\mathrm{I})\le \epsilon \quad \text{a.s. for $N$ sufficient large.}
\]
\item Under Assumption~\ref{ass:A7}, $\hat g_{k,p}^{(\nu),\mathrm{Fou}}$ is continuous on $M_{k,p}$. Since $M_{k,p}$ is compact, $\hat g^{(\nu),\mathrm{Fou}}_{k,p}$ is uniformly continuous on $M_{k,p}$. Consequently, term~$(\mathrm{III})$ can be controlled deterministically. In particular, for any $\mathbf m_{k,p}\in U\paren{\mathbf m_{k,p}^{(i)},\eta}$,
\begin{equation}
 \sup_{\mathbf{m}_{k,p} \in M_{k,p}}\norm{\hat g_{k,p}^{(\nu),\mathrm{Fou}}\paren{\mathbf m_{k,p}}-\hat g_{k,p}^{(\nu),\mathrm{Fou}}\paren{\mathbf m_{k,p}^{(i)}}}  \leq  \epsilon
\label{eq:second_bound_rqmc}
\end{equation}

\item For Term \((\mathrm{II})\), note that it involves only finitely many points $\mathbf m_{k,p}^{(i)}$. Again by the SLLN for the nested uniform scrambling Sobol sequence  (Assumption~\ref{ass:A10} and \cite[Theorem 4.2]{owen_strong_2021}) for $\tilde h_{k,p}^{(\nu)}$,
\[
I_{N,\bar S}^{\mathrm{RQMC}}\!\left[\tilde h_{k,p}^{(\nu)}\paren{\cdot;\mathbf m_{k,p}^{(i)}}\right]
\xrightarrow[]{\mathrm{a.s.}}
\hat g_{k,p}^{(\nu),\mathrm{Fou}}\paren{\mathbf m_{k,p}^{(i)}},
\qquad N\to\infty,
\]
with $\hat g_{k,p}^{(\nu),\mathrm{Fou}}\paren{\mathbf m_{k,p}^{(i)}} = \mathbb{E}_\mathbf{V}\!\left[\tilde h_{k,p}^{(\nu)}\paren{\mathbf V,\mathbf m_{k,p}^{(i)}}\right]$ for each $i=1,\dots,E$. Taking the maximum over finitely many $i$ preserves almost sure convergence, hence
\(
\max_{1\le i\le E}(\mathrm{II}) \xrightarrow[]{\mathrm{a.s.}} 0
\). In particular, for $N$ sufficient large, $\sup_{\mathbf m_{k,p}\in M_{k,p}}(\mathrm{II})\le \epsilon$ almost surely.
\end{itemize}
Putting the three bounds into \eqref{eq:tri_decomp}, we obtain that for $N$ sufficient large,
\[
\sup_{\mathbf m_{k,p}\in M_{k,p}}
\norm{
I_{N,\bar S}^{\mathrm{RQMC}}\!\left[\tilde h_{k,p}^{(\nu)}(\cdot;\mathbf m_{k,p})\right]
- \hat g_{k,p}^{(\nu),\mathrm{Fou}}(\mathbf m_{k,p})
}
\le 3\epsilon
\qquad \text{a.s.}
\]
Since $\epsilon>0$ is arbitrary, the first convergence in \eqref{eq:uniform_SLLN_rqmc_fou} follows.

\subsection{Proof for Theorem \ref{theorem:consistentcy-RQMC-Fou}}\label{appendix:proof_prop_consistency}

The proof follows standard local stability and implicit-function arguments for stochastic generalized equations (see \cite{dontchev_implicit_2009}).
For $\eta>0$, define the neighborhood of the optimal solution $\mathbf{z}^*$ by
\begin{equation}
    U\paren{\mathbf{z}^*,\eta}
:= \bigl\{\mathbf{z} \in \mathcal{Z}:\ \norm{\mathbf{z}-\mathbf{z^*}}\le \eta\bigr\}.
\label{eq:neigbor_opt_sol}
\end{equation}
Recall that the aggregate Fourier-based integrands admit the finite decompositions. By Lemma \ref{lem:uniform_conv_RQMC_Fou}, the USLLN holds on each $M_{k,p} \subset \mathbb{R}^k$ for 
$I_{N,\bar S}^{\mathrm{RQMC}}\!\left[\tilde h^{(\nu)}_{k,p}(\cdot;\mathbf{m}_{k,p})\right]$.
Applying the triangle inequality yields, for $N$ sufficiently large, the corresponding USLLN for the aggregate integrands:
\begin{align}
\sup_{\mathbf m \in M}
\norm{
I_{N,\bar S}^{\mathrm{RQMC}}\!\left[\tilde h^{(\nu)}(\cdot;\mathbf{m})\right]
- \hat g^{(\nu),\mathrm{Fou}}(\mathbf m)
}
&\xrightarrow[]{\mathrm{a.s.}} 0.
\label{eq:uniform_SLLN_aggregate}
\end{align}
Next, by the strong regularity of solution $\mathbf{z}^*$, from  Definition~\ref{def:strong_regular_solution}, the Fourier representation of the true Hessian
$\hat{\mathcal{L}}_{\nabla^2_{\mathbf{z}}}^{\mathrm{Fou}}(\mathbf{z}^*)$ is invertible. Moreover, by continuity of $\hat{\mathcal{L}}_{\nabla^2_{\mathbf{z}}}^{\mathrm{Fou}}(\cdot)$, there exists $\eta>0$ and a constant $C_L>0$ such that
\begin{equation}
   \norm{\hat{\mathcal{L}}_{\nabla^2_{\mathbf{z}}}^{\mathrm{Fou}} (\mathbf{z})^{-1}} \leq C_L ,
   \qquad \forall\,\mathbf{z}\in U\paren{\mathbf{z}^*,\eta}.
   \label{eq:bound_inverse_norm}
\end{equation}
Define the map $T_N$ associated with the Fourier-RQMC problem by
\begin{equation}
    T_N(\mathbf{z})
    := \mathbf{z} - 
    \brac{\hat{\mathcal{L}}_{\nabla^2_{\mathbf{z}}}^{\mathrm{Fou}}(\mathbf{z})^{-1}}
    I_{N,\bar S}^{\mathrm{RQMC}}\brac{\mathcal{H}^{(1)} (\cdot;\mathbf{z})},
\label{eq:contraction_map_RQMC}
\end{equation}
By construction,
\(
I_{N,\bar S}^{\mathrm{RQMC}}\brac{\mathcal{H}^{(1)} (\cdot;\mathbf{z})}=0
\quad\Longleftrightarrow\quad
T_N(\mathbf{z})=\mathbf{z}.
\)

\medskip
\noindent\textbf{Step 1: $T_N$ is a contraction on $U\paren{\mathbf{z}^*,\eta}$.}

Let $\mathbf{z}_1,\mathbf{z}_2\in U\paren{\mathbf{z}^*,\eta}$.
Since $\hat{\mathcal{L}}_{\nabla_{\mathbf{z}}}^{\mathrm{Fou}}$ is continuously differentiable on $U\paren{\mathbf{z}^*,\eta}$, the mean value theorem yields
\begin{equation*}
    I_{N,\bar S}^{\mathrm{RQMC}}\brac{\mathcal{H}^{(1)}(\cdot;\mathbf{z}_1) }
    - I_{N,\bar S}^{\mathrm{RQMC}}\brac{\mathcal{H}^{(1)} (\cdot;\mathbf{z}_2)}
    =
    I_{N,\bar S}^{\mathrm{RQMC}}\brac{\mathcal{H}^{(2)}(\cdot;\bar{\mathbf{z}})} 
    \,(\mathbf{z}_1-\mathbf{z}_2),
\end{equation*}
for some $\bar{\mathbf{z}}$ on the line segment between $\mathbf{z}_1$ and $\mathbf{z}_2$.
Substituting into \eqref{eq:contraction_map_RQMC} gives
\begin{equation}
\begin{aligned}
    T_N\paren{\mathbf{z}_1}- T_N\paren{\mathbf{z}_2}
    &=
    \paren{\boldsymbol I
    - \,\hat{\mathcal{L}}_{\nabla^2_{\mathbf{z}}}^{\mathrm{Fou}} (\bar{\mathbf{z}})^{-1}\,
    I_{N,\bar S}^{\mathrm{RQMC}}\brac{\mathcal{H}^{(2)}(\cdot;\bar{\mathbf{z}})} }
    (\mathbf{z}_1-\mathbf{z}_2).
\end{aligned}
\end{equation}
From \eqref{eq:uniform_SLLN_aggregate} (restricted to $U$) we obtain the uniform convergence,
\begin{equation}
\sup_{\bar{\mathbf z} \in U}
\Bigl\|
I_{N,\bar S}^{\mathrm{RQMC}}\!\brac{\mathcal{H}^{(2)}(\cdot;\bar{\mathbf z})}
- \hat {\mathcal{L}}_{\nabla^2_{\mathbf{z}}}^{\mathrm{Fou}}(\bar{\mathbf z})
\Bigr\|
\xrightarrow[]{\mathrm{a.s.}} 0,
\label{eq:uniform_SLLN_hessian_z}
\end{equation}
and combining \eqref{eq:bound_inverse_norm} with \eqref{eq:uniform_SLLN_hessian_z} yields: for $N$ sufficient large, there exists $q\in(0,1)$ such that
\[
\sup_{\bar{\mathbf z}\in U}
\Bigl\|
\boldsymbol I
- \,\hat{\mathcal{L}}_{\nabla^2_{\mathbf{z}}}^{\mathrm{Fou}} (\bar{\mathbf z})^{-1}\,
I_{N,\bar S}^{\mathrm{RQMC}}\brac{\mathcal{H}^{(2)}(\cdot;\bar{\mathbf z})} 
\Bigr\|
\le q.
\]
Hence, $\norm{T_N(\mathbf{z}_1)-T_N(\mathbf{z}_2)}\le q\,\norm{\mathbf{z}_1-\mathbf{z}_2}$ for all $\mathbf{z}_1,\mathbf{z}_2\in U$, i.e., $T_N$ is a contraction on $U$.

\medskip
\noindent\textbf{Step 2: $T_N$ maps $U\paren{\mathbf{z}^*,\eta}$ into itself.}

By \eqref{eq:uniform_SLLN_aggregate} (restricted to $U$) we also have
\begin{equation}
\sup_{\mathbf z \in U}
\Bigl\|
I_{N,\bar S}^{\mathrm{RQMC}}\!\left[\mathcal{H}^{(1)}(\cdot;\mathbf{z})\right]
- \hat {\mathcal{L}}_{\nabla_{\mathbf{z}}}^{\mathrm{Fou}}(\mathbf z)
\Bigr\|
\xrightarrow[]{\mathrm{a.s.}} 0.
\label{eq:uniform_SLLN_grad_z}
\end{equation}
Since $\hat {\mathcal{L}}_{\nabla}^{\mathrm{Fou}}(\mathbf z^*) =0$, it follows from
\eqref{eq:contraction_map_RQMC} and \eqref{eq:uniform_SLLN_grad_z} that
$\norm{T_N(\mathbf{z}^*)- \mathbf{z}^*} \xrightarrow[]{\mathrm{a.s.}} 0$.
Therefore, for all sufficiently large $N$ , we can enforce
$\norm{T_N(\mathbf{z}^*)- \mathbf{z}^*} \le (1-q)\eta$, which implies that for any $\mathbf{z}\in U$,
\[
\norm{T_N(\mathbf{z})-\mathbf{z}^*}
\le
\norm{T_N(\mathbf{z})-T_N(\mathbf{z}^*)}
+\norm{T_N(\mathbf{z}^*)-\mathbf{z}^*}
\le
q\,\norm{\mathbf{z}-\mathbf{z}^*}+(1-q)\eta
\le \eta,
\]
i.e. $T_N(U)\subset U$.

By \textbf{Steps~1--2}, for all sufficiently large $N$ , $T_N$ is a Banach contraction on $U$ and maps $U$ into itself. Hence, by the Banach fixed point theorem, there exists a unique fixed point
${\mathbf{z}}_{N,\bar S}^{\mathrm{RQMC},*}\in U$ such that
$T_N\!\paren{{\mathbf{z}}_{N,\bar S}^{\mathrm{RQMC},*}}={\mathbf{z}}_{N,\bar S}^{\mathrm{RQMC},*}$, equivalently
$I_{N,\bar S}^{\mathrm{RQMC}}\!\brac{\mathcal{H}^{(1)}\paren{\cdot;{\mathbf{z}}_{N,\bar S}^{\mathrm{RQMC},*}}}=0$, Finally
\begin{align*}
\norm{{\mathbf{z}}_{N,\bar S}^{\mathrm{RQMC},*} - \mathbf{z}^*}
&= \norm{T_N\paren{{\mathbf{z}}_{N,\bar S}^{\mathrm{RQMC},*}}- T_N\paren{\mathbf{z}^*} + T_N\paren{\mathbf{z}^*} - \mathbf{z^*}} \\
&\leq q\, \norm{{\mathbf{z}}_{N,\bar S}^{\mathrm{RQMC},*} - \mathbf{z}^*}
    + \norm{T_N(\mathbf{z}^*)-\mathbf{z}^*}\\
&= q\, \norm{{\mathbf{z}}_{N,\bar S}^{\mathrm{RQMC},*} - \mathbf{z}^*}
    + \norm{ \,\hat{\mathcal{L}}_{\nabla^2_{\mathbf{z}}}^{\mathrm{Fou}} ({\mathbf{z^*}})^{-1}\,
I_{N,\bar S}^{\mathrm{RQMC}}\brac{\mathcal{H}^{(1)} (\cdot;{\mathbf{z^*}})}}\\
&\leq q\, \norm{{\mathbf{z}}_{N,\bar S}^{\mathrm{RQMC},*} - \mathbf{z}^*}
    +  \norm{\hat{\mathcal{L}}_{\nabla^2_{\mathbf{z}}}^{\mathrm{Fou}} ({\mathbf{z^*}})^{-1}}
      \norm{I_{N,\bar S}^{\mathrm{RQMC}}\brac{\mathcal{H}^{(1)} (\cdot;{\mathbf{z^*}})}} .
\end{align*}
Rearranging yields
\[
(1-q)\,\norm{{\mathbf{z}}_{N,\bar S}^{\mathrm{RQMC},*} - \mathbf{z}^*}
\le
 \norm{\hat{\mathcal{L}}_{\nabla^2_{\mathbf{z}}}^{\mathrm{Fou}} ({\mathbf{z^*}})^{-1}}
\norm{I_{N,\bar S}^{\mathrm{RQMC}}\brac{\mathcal{H}^{(1)} (\cdot;{\mathbf{z^*}})}},
\]
and \(
 \norm{\hat{\mathcal{L}}_{\nabla^2_{\mathbf{z}}}^{\mathrm{Fou}} ({\mathbf{z^*}})^{-1}}
\norm{I_{N,\bar S}^{\mathrm{RQMC}}\brac{\mathcal{H}^{(1)} (\cdot;{\mathbf{z^*}})}}
\xrightarrow[]{\mathrm{a.s.}} 0,
\)
we conclude that ${\mathbf{z}}_{N,\bar S}^{\mathrm{RQMC},*} \xrightarrow[]{\mathrm{a.s.}} \mathbf{z}^*$ as $N\to\infty$.

For the CLT in Theorem~\ref{theorem:efficiency-RQMC-Fou}, we work under the joint regime $S_{\mathrm{shift}}\to\infty$. We therefore provide uniform convergence of the Fourier-RQMC estimators in Lemma~\ref{lem:uniform_conv_RQMC_Fou_S}, and then deduce consistency of the optimizer in Proposition~\ref{prop:consistentcy-RQMC-Fou-S} with $S_{\mathrm{shift}}$.
\begin{lemma}[Uniform convergence of Fourier-RQMC estimators with $S_\mathrm{shift}$]
\label{lem:uniform_conv_RQMC_Fou_S}
 Fix $N = \bar N$, let $\{\mathbf{v}_n\}_{n=1}^{\bar N}$ be the Sobol sequence, and $\curly{\mathbf{v}_n^{(s)}}_{s=1}^{S_{\mathrm{shift}}}$ be the sequence obtained by applying suitable randomization (i.e., nested uniform scrambling, digital shifting) to
$\{\mathbf{v}_n\}_{n=1}^{\bar N}$.
Then the Fourier-RQMC estimators
$I_{\bar N,S_{\mathrm{shift}}}^{\mathrm{RQMC}}\!\left[\tilde h_{k,p}^{(\nu)}(\cdot;\mathbf{m}_{k,p})\right]$
satisfy a USLLN on $M_{k,p}$, that is,
\begin{equation}
\begin{aligned}
\sup_{\mathbf m_{k,p} \in M_{k,p}}
\norm{
I_{\bar N,S_{\mathrm{shift}}}^{\mathrm{RQMC}}\!\left[\tilde h_{k,p}^{(\nu)}(\cdot;\mathbf{m}_{k,p})\right]
- \hat g_{k,p}^{(\nu),\mathrm{Fou}}(\mathbf m_{k,p})
}
&\xrightarrow[]{\mathrm{a.s.}} 0, 
\end{aligned}
\label{eq:uniform_SLLN_rqmc_fou_S}
\end{equation}
as $S_{\mathrm{shift}} \to \infty$.
\end{lemma}
\begin{proof}
The argument follows the same lines as the proof of
Lemma~\ref{lem:uniform_conv_RQMC_Fou}.
The only difference is that, instead of taking expectations w.r.t. a generic random vector $\mathbf V \in [0,1]^k$, we take expectations w.r.t. the random shift index $S$. Since $\bar N$ and $\{\mathbf v_n^{(s)}\}_{s=1}^{S_{\mathrm{shift}}}$ are i.i.d.\,  the standard SLLN applies without using \cite[Theorem 4.2]{owen_strong_2021}.
Consequently,
\begin{align*}
I_{\bar N,S_{\mathrm{shift}}}^{\mathrm{RQMC}}
\!\left[\delta_{\bar{\mathbf m}_{k,p},\eta}^{(\nu)}\right]
&\xrightarrow[]{\mathrm{a.s.}}
\mathbb{E}_{S}\!\left[
I_{\bar N}^{\mathrm{RQMC}}
\!\left[\delta_{\bar{\mathbf m}_{k,p},\eta}^{(\nu)} \paren{\mathbf{v}_n^{(s)}}\right]
\right], \\[0.5em]
I_{\bar N,S_{\mathrm{shift}}}^{\mathrm{RQMC}}
\!\left[\tilde h^{(\nu)}\!\left(\cdot;\mathbf m_{k,p}\right)\right]
&\xrightarrow[]{\mathrm{a.s.}}
\mathbb{E}_{S}\!\left[
I_{\bar N}^{\mathrm{RQMC}}
\!\left[\tilde h^{(\nu)}\!\paren{\mathbf{v}_n^{(s)},\mathbf m_{k,p}}\right]
\right].
\end{align*}
as $S_{\mathrm{shift}} \to \infty$, and
\[
\mathbb{E}_{S}\!\left[
I_{\bar N}^{\mathrm{RQMC}}
\!\left[\tilde h^{(\nu)}\!\paren{\mathbf{v}_n^{(s)},\mathbf m_{k,p}}\right]
\right]
=
\hat g_{k,p}^{(\nu),\mathrm{Fou}}(\mathbf m_{k,p})
\].
\end{proof}

\begin{proposition}[Consistency of solution from Fourier-RQMC problem with $S_{\mathrm{shift}}$]\label{prop:consistentcy-RQMC-Fou-S}
Fix $N = \bar N$, let $\{\mathbf{v}_n\}_{n=1}^N$ be the Sobol sequence, and $\curly{\mathbf{v}_n^{(s)}}_{s=1}^{S_{\mathrm{shift}}}$ be the sequence obtained by applying suitable randomization (i.e., nested uniform scrambling, digital shifting) to
$\{\mathbf{v}_n\}_{n=1}^N$. Suppose that Assumption
   \eqref{ass:C4} holds.
Then, as $S_{\mathrm{shift}} \to \infty$, the Fourier-RQMC problem admits a unique solution
${\mathbf{z}}_{\bar N,S_{\mathrm{shift}}}^{\mathrm{RQMC},*} \in \mathcal{Z}$, and
\[
{\mathbf{z}}_{\bar N,S_{\mathrm{shift}}}^{\mathrm{RQMC},*}
\xrightarrow[]{\mathrm{a.s.}}
\mathbf z_{\bar N,\infty}^{\mathrm{RQMC},*} .
\]
\end{proposition}
\begin{proof}
The proof follows the same arguments as that of
Theorem~\ref{theorem:consistentcy-RQMC-Fou}.
The only modification is that we invoke
Lemma~\ref{lem:uniform_conv_RQMC_Fou_S} to obtain uniform convergence of the component integrands under random shifting. Moreover, by Assumption~\eqref{ass:C4}, the limiting solution
$\mathbf z_{\bar N,\infty}^{\mathrm{RQMC},*}$ is strongly regular.
As a result, the same Banach fixed-point argument applies to 
\(
I_{\bar N,\infty}^{\mathrm{RQMC}}\!\left[\mathcal{H}^{(0)}\right],
\)
which yields the desired consistency result.
\end{proof}

\subsection{Proof for Theorem \ref{theorem:efficiency-RQMC-Fou}}\label{appendix:proof_theorem_efficiency}
We decompose
\[
{\mathbf{z}}_{N,S_{\mathrm{shift}}}^{\mathrm{RQMC},*} - \mathbf{z}^*
=
\underbrace{{\mathbf{z}}_{N,S_{\mathrm{shift}}}^{\mathrm{RQMC},*} - {\mathbf{z}}_{N,\infty}^{\mathrm{RQMC},*}}_{(\mathrm{A})}
+
\underbrace{{\mathbf{z}}_{N,\infty}^{\mathrm{RQMC},*} - \mathbf{z}^*}_{(\mathrm{B})}.
\]

\paragraph{Term (A)}\mbox{}\\
For each fixed $N$,
Proposition~\ref{prop:consistentcy-RQMC-Fou-S} yields
\[
{\mathbf{z}}_{N,S_{\mathrm{shift}}}^{\mathrm{RQMC},*} \xrightarrow[]{a.s} \mathbf z_{N,\infty}^{\mathrm{RQMC},*},
\qquad S_{\mathrm{shift}}\to\infty.
\]
Moreover, by strong regularity at $\mathbf z_{N,\infty}^{\mathrm{RQMC},*}$ and the corresponding linearized equation (Definition~\ref{def:strong_regular_solution}), we have the expansion
\begin{equation*}
\begin{aligned}
\mathbf z_{N,S_{\mathrm{shift}}}^{\mathrm{RQMC},*}- \mathbf z_{N,\infty}^{\mathrm{RQMC},*}
&=
- \Big(I_{N,\infty}^{\mathrm{RQMC}}\!\big[\mathcal{H}^{(2)}(\,\cdot\,;\mathbf z_{N,\infty}^{\mathrm{RQMC},*})\big]\Big)^{-1} \\
&\quad \times \Big(
I_{N,\infty}^{\mathrm{RQMC}}\!\big[\mathcal{H}^{(1)}(\,\cdot\,;\mathbf z_{N,\infty}^{\mathrm{RQMC},*})\big]
- I_{N,S_{\mathrm{shift}}}^{\mathrm{RQMC}}\!\big[\mathcal{H}^{(1)}(\,\cdot\,;\mathbf z_{N,\infty}^{\mathrm{RQMC},*})\big]
\Big)
+ a_N .
\end{aligned}
\end{equation*}
where \(
a_N
:=
o_{\mathbb{P}}\!\paren{
\norm{I_{N,\infty}^{\mathrm{RQMC}}\brac{\mathcal{H}^{(1)}\big( \cdot;\mathbf z_{N,\infty}^{\mathrm{RQMC},*}\big)}
-I_{N,S_{\mathrm{shift}}}^{\mathrm{RQMC}}\brac{\mathcal{H}^{(1)}\big( \cdot;\mathbf z_{N,\infty}^{\mathrm{RQMC},*}\big)}}}
=
o_{\mathbb{P}}\!\paren{S_\mathrm{shift}^{-\tfrac{1}{2}}},
\)
the latter equality following from \eqref{eq:RMSE_RQMC}. Multiplying by $\sqrt{S_{\mathrm{shift}}}$ gives
\begin{equation}\label{eq:delta_taylor_MC_2}
\begin{aligned}
\sqrt{S_{\mathrm{shift}}}
\bigl(
\mathbf z_{N,S_{\mathrm{shift}}}^{\mathrm{RQMC},*}
-
\mathbf z_{N,\infty}^{\mathrm{RQMC},*}
\bigr)
&=
-
\brac{I_{N,\infty}^{\mathrm{RQMC}}\brac{\mathcal{H}^{(2)}\big( \cdot;\mathbf z_{N,\infty}^{\mathrm{RQMC},*}\big)}}^{-1}
\\[1mm]
&\quad\times
\sqrt{S_{\mathrm{shift}}}
\paren{
I_{N,\infty}^{\mathrm{RQMC}}
\brac{
\mathcal{H}^{(1)}
\paren{\cdot;
\mathbf z_{N,\infty}^{\mathrm{RQMC},*}
}
}
-
I_{N,S_{\mathrm{shift}}}^{\mathrm{RQMC}}
\brac{
\mathcal{H}^{(1)}
\paren{\cdot;
\mathbf z_{N,\infty}^{\mathrm{RQMC},*}
}
}
}
\\
&\quad
+ o_{\mathbb{P}}(1).
\end{aligned}
\end{equation}
Since we use $S_\mathrm{shift}$ independent shifts, we can apply the multivariate CLT as $S_\mathrm{shift}\to\infty$,
\[
\sqrt{S_{\mathrm{shift}}}\!\paren{
I_{N,S_{\mathrm{shift}}}^{\mathrm{RQMC}}\brac{\mathcal{H}^{(1)}\big( \cdot;\mathbf z_{N,\infty}^{\mathrm{RQMC},*}\big)}
-I_{N,\infty}^{\mathrm{RQMC}}\brac{\mathcal{H}^{(1)}\big( \cdot;\mathbf z_{N,\infty}^{\mathrm{RQMC},*}\big)}}
\ \xrightarrow[]{\mathrm{law}}\
\mathcal{N}\!\left(\mathbf{0},\,\boldsymbol H_{N,\infty}^{\mathrm{RQMC}}\big(\mathbf z_{N,\infty}^{\mathrm{RQMC},*}\big)\right),
\]
where
\(
\boldsymbol H_{N,\infty}^{\mathrm{RQMC}}\big(\mathbf z_{N,\infty}^{\mathrm{RQMC},*}\big)
=
\mathrm{Var}_s\!\Big(I_{N}^{\mathrm{RQMC}}\brac{\mathcal{H}^{(1)}\big( \mathbf{v}_n^{(s)},\mathbf z_{N,\infty}^{\mathrm{RQMC},*}\big)}\Big).
\)
Combining this with \eqref{eq:delta_taylor_MC_2} and the invertibility of
\(I_{N,\infty}^{\mathrm{RQMC}}\brac{\mathcal{H}^{(2)}\big( \cdot;\mathbf z_{N,\infty}^{\mathrm{RQMC},*}\big)}\),
it yields
\[
\sqrt{S_{\mathrm{shift}}}\!\left({\mathbf{z}}_{N,S_{\mathrm{shift}}}^{\mathrm{RQMC},*} - \mathbf z_{N,\infty}^{\mathrm{RQMC},*} \right)
\ \xrightarrow[]{\mathrm{law}}\
\mathcal{N}\!\left(\mathbf{0},\,\boldsymbol V_{N,\infty}^{\mathrm{RQMC}}\big(\cdot; \mathbf z_{N,\infty}^{\mathrm{RQMC},*}\big)\right),
\]
with the sandwich covariance
\[
\boldsymbol V_{N,\infty}^{\mathrm{RQMC}}\big(\cdot; \mathbf z_{N,\infty}^{\mathrm{RQMC},*}\big)
:=
\paren{I_{N,\infty}^{\mathrm{RQMC}}\brac{\mathcal{H}^{(2)}\big( \cdot;\mathbf z_{N,\infty}^{\mathrm{RQMC},*}\big)}}^{-1}
\boldsymbol H_{N,\infty}^{\mathrm{RQMC}}\big( \mathbf z_{N,\infty}^{\mathrm{RQMC},*} \big)
\paren{I_{N,\infty}^{\mathrm{RQMC}}\brac{\mathcal{H}^{(2)}\big( \cdot;\mathbf z_{N,\infty}^{\mathrm{RQMC},*}\big)}}^{-1}.
\]

\paragraph{Term (B)}\mbox{}\\
 By Theorem~\ref{theorem:consistentcy-RQMC-Fou},
\[
{\mathbf{z}}_{N,\infty}^{\mathrm{RQMC},*}\xrightarrow[]{a.s}\mathbf z^*,
\qquad N\to\infty.
\]
for the joint regime, by Assumption \ref{assump:joint_grow_N_S} \eqref{ass:i},
we have that term (B) is negligible at the $\sqrt{S_{\mathrm{shift}}}N^{r}$scale. Hence,
\[
\sqrt{S_{\mathrm{shift}}}N^{r}\!\left({\mathbf{z}}_{N,S_{\mathrm{shift}}}^{\mathrm{RQMC},*}-\mathbf z^*\right)
=
\sqrt{S_{\mathrm{shift}}}N^{r}\!\left({\mathbf{z}}_{N,S_{\mathrm{shift}}}^{\mathrm{RQMC},*}-\mathbf z_{N,\infty}^{\mathrm{RQMC},*}\right)
+ o_{\mathbb{P}}(1),
\]
and the limiting distribution is governed by term (A). 

Finally, to identify the limiting covariance at $\mathbf z^*$, we use
\[
\begin{aligned}
\norm{
\boldsymbol N^{2r} \boldsymbol V_{N,\infty}^{\mathrm{RQMC}}\big(\cdot; \mathbf z_{N,\infty}^{\mathrm{RQMC},*}\big)
- \boldsymbol V\paren{\mathbf{z}^*}
}
\;\leq\;&\;
N^{2r}\norm{
\boldsymbol V_{N,\infty}^{\mathrm{RQMC}}\big(\cdot; \mathbf z_{N,\infty}^{\mathrm{RQMC},*}\big)
- \boldsymbol V_{N,\infty}^{\mathrm{RQMC}}\big(\cdot; \mathbf z^*\big)
} \\
&\;+\;
\norm{
N^{2r}\boldsymbol V_{N,\infty}^{\mathrm{RQMC}}\big(\cdot; \mathbf z^*\big)
- \boldsymbol V\paren{\mathbf{z}^*}
}.
\end{aligned}
\]
As $N \to \infty$, the first term converges to $0$ by continuity of $\boldsymbol V_{N,\infty}^{\mathrm{RQMC}}(\cdot)$ on a neighborhood of $\mathbf z^*$
together with ${\mathbf{z}}_{N,\infty}^{\mathrm{RQMC},*}\to\mathbf z^*$. The second term converges to $0$ by the
 USLLN with $N$ of 
$I_{N,\infty}^{\mathrm{RQMC}}\brac{\mathcal{H}^{(2)}} \to \hat{\mathcal{L}}_{\nabla_{\mathbf{z}}^2}^{\mathrm{Fou}}$ from \eqref{eq:uniform_SLLN_aggregate}, the uniform
invertibility of $\hat{\mathcal{L}}_{\nabla_{\mathbf{z}}^2}^{\mathrm{Fou}}$ on that neighborhood due to the strong regularity of $\mathbf{z}^*$ and  $N^{2r} H_{N,\infty}\big( \mathbf z^{*} \big) \xrightarrow{} \boldsymbol{H}(\mathbf{z}^*)$,  from Assumption \ref{assump:joint_grow_N_S} \eqref{ass:ii}.
Therefore,
\(
\boldsymbol V_{N,\infty}^{\mathrm{RQMC}}\big(\cdot;\mathbf z_{N,\infty}^{\mathrm{RQMC},*}\big)\to \boldsymbol V(\mathbf z^*).
\)
By Slutsky's theorem, as $S_{\mathrm{shift}}\to\infty$ and $N\to\infty$, we have
\[
\sqrt{S_{\mathrm{shift}}}N^r\!\left({\mathbf{z}}_{N,S_{\mathrm{shift}}}^{\mathrm{RQMC},*}-\mathbf z^*\right)
\ \xrightarrow[]{\mathrm{law}}\
\mathcal{N}\!\left(\mathbf 0,\ \boldsymbol V\paren{\mathbf z^*}\right).
\]

\subsection{Proof for Corrolary \ref{coro:complexity_sing_RQMC}}\label{appendix:proof_coro_sing_allocation}
For a fixed $N$ and $S_{\mathrm{shift}}$ digital shifts, the total work decomposes into a one-off sampling cost and the cumulative cost of $J$ SQP iterations:
\begin{equation}
    W_{\mathrm{sing}}^{\mathrm{RQMC}}(N,J)
= C_{\mathrm{draw}}(N)+J\,C_{\mathrm{iter}}(N,d).
\label{eq:work_sing_RQMC}
\end{equation}

\begin{itemize}
    \item \emph{Sampling cost.} For each $k\in \mathcal{I}_{q_{\ell}}$ we generate a single fixed RQMC design. The associated cost is
    \begin{equation}
      C_{\mathrm{draw}}(N)
    :=
    \mathcal{O}\paren{N S_{\mathrm{shift}}\,c_{\mathrm{draw}}},
    \qquad
    c_{\mathrm{draw}}
    :=
    \max_{k \in \mathcal{I}_{q_{\ell}}} c_{\mathrm{draw},k},  
    \label{eq:cost_draw_sing_RQMC}
    \end{equation}
    where $c_{\mathrm{draw},k}$ denotes the cost of drawing one sample at the dependence level $k$.

    \item \emph{Per-iteration cost.}
    Each SQP iteration incurs the following costs:
    \begin{itemize}
        \item \emph{Function and gradient evaluation.}
        Evaluating the aggregate integrands $\tilde h^{(\nu)}$ involves summing over all components, hence
        \begin{equation}
           C_{\mathrm{eval}}(N)
        :=
        \mathcal{O}\paren{N S_{\mathrm{shift}}\,c_{\mathrm{eval}}},
        \qquad
        c_{\mathrm{eval}}
        \approx
        c_{\max}\,N_{\mathrm{comp}}, 
        \label{eq:cost_eval_sing_RQMC}
        \end{equation}
        where $c_{\max}$ denotes the maximum cost $c_{k,p}$ incurred in evaluating the component integrands $\tilde h_{k,p}^{(\nu)}$.
        \begin{equation}
           c_{\max}
        :=
        \max_{\nu \in \{0,1\}}
        \max_{k \in \mathcal{I}_{q_{\ell}}}\ \max_{\mathbf{p}\in\mathcal{I}_k} c_{k,p},
        \qquad
        N_{\mathrm{comp}}
        :=
        \sum_{k \in \mathcal{I}_{q_{\ell}}}\ \sum_{\mathbf{p}\in\mathcal{I}_k} 1 . 
        \label{eq:c_and_N_comp}
        \end{equation}

        \item \emph{BFGS update.}
        Forming the BFGS update costs $\mathcal{O}(d^2)$ (outer products and matrix-vector products).
        In addition, we still incur the evaluation cost $C_{\mathrm{eval}}(N)$ to compute the required gradient differences.

        \item \emph{QP solve.}
        With one active inequality constraint, solving the resulting dense QP costs $\mathcal{O}((d+1)^3)$.
    \end{itemize}
    Collecting the dominant terms
, the per-iteration cost can be summarized as
    \begin{equation}
        C_{\mathrm{iter}}(N,d)
        =
        \mathcal{O}\paren{N S_{\mathrm{shift}}\,c_{\mathrm{eval}} + (d+1)^3}
    \label{eq:per_iteration_cost_sing_RQMC}
    \end{equation}
\end{itemize}
By \eqref{eq:statistical_error_RQMC_sol}, choosing
$N=N(\varepsilon)$ such that $\varepsilon_{\mathrm{stat}}^{\mathrm{RQMC}}(N)\le \varepsilon/2$ yields
\[
N(\varepsilon)=\mathcal O\!\paren{\varepsilon^{-\tfrac{1}{r}}}.
\]
Moreover, by \eqref{eq:superlinear_err}, achieving $\varepsilon_{\mathrm{opt}}(j)\le \varepsilon/2$ requires
\[
J(\varepsilon)=\mathcal O\!\big(\log\log(1/\varepsilon)\big).
\]
Substituting the choices $N(\varepsilon)$ and $J(\varepsilon)$ into \eqref{eq:work_sing_RQMC} yields
\[
W_{\mathrm{sing}}^{\mathrm{RQMC}}(\varepsilon)
=\mathcal O\!\Big(\varepsilon^{-\tfrac{1}{r}}\log\log(1/\varepsilon)\Big),
\]
This concludes the proof.

\subsection{Proof for Corollary \ref{coro:multilevel_allocation}}
\label{appendix:proof_coro_multi_allocation}
Following Algorithm~\ref{alg:RQMC_Fou_multi}, at each level $j$ we generate a fresh set of
$S_{\mathrm{shift}}$ independent digital shifts
$\curly{\mathbf{v}_n^{(s,j)}}_{s=1}^{S_{\mathrm{shift}}}$.
The resulting per-iteration cost follows the same structure as the single-level case \eqref{eq:per_iteration_cost_sing_RQMC}
where the per-sample cost at level $j$ is 
\begin{equation}
    c_j := c_{\mathrm{eval}}^{(j)} + c_{\mathrm{draw}}^{(j)}.
    \label{eq:per_iter_multi_cost}
\end{equation}
The per-drawing cost at level $j$ is defined as
\(c_{\mathrm{draw}}^{(j)} := \max_{k \in \mathcal{I}_{q_{\ell}}} c_{\mathrm{draw},k}^{(j)}\),
while the per-evaluation cost at level $j$ satisfies
\(c_{\mathrm{eval}}^{(j)} \approx c_{\max}^{(j)}\,N_{\mathrm{comp}}\),
where
\(c_{\max}^{(j)} := \max_{\nu \in \{0,1\}} \max_{k \in \mathcal{I}_{q_{\ell}}} \max_{\mathbf{p} \in \mathcal{I}_k} c_{k,p}^{(j)}\)
denotes the maximum  cost incurred in evaluating the level-$j$ component difference integrands
\(\Delta \tilde h_{k,p}^{(\nu,j)}\). The MSE at level $J$ is computed as
\begin{equation}
    \mathrm{MSE}_{\mathrm{stat}}^{\mathrm{RQMC}}
    = \sum_{j=1}^{J} \frac{\boldsymbol D_j}{S_{\mathrm{shift}} N_j^{2r}}.
\label{eq:MSE_level_j}
\end{equation}
From \eqref{eq:per_iter_multi_cost} and \eqref{eq:MSE_level_j}, our constrained optimization problem is
\begin{equation}
    \min_{N_j \geq 1} \sum_{j=1}^{J} S_{\mathrm{shift}} c_j N_j
    \qquad \text{ s.t } \quad
    \sum_{j=1}^{J} \frac{\boldsymbol D_j}{S_{\mathrm{shift}} N_j^{2r}}
    \leq \frac{\varepsilon^2}{4}.
\label{eq:minize_N_j_work}
\end{equation}
The Lagrangian associated with \eqref{eq:minize_N_j_work}. with $\mu >0$ is:
\begin{equation*}
    \mathcal{L}(N_j, \mu)
    = \sum_{j=1}^{J} S_{\mathrm{shift}} c_j N_j
    +\mu \left(\sum_{j=1}^{J} \frac{4\boldsymbol D_j}{S_{\mathrm{shift}} N_j^{2r}} - \varepsilon^2 \right).
\end{equation*}
The F.O.C.\ gives:
\begin{equation}
\left\{
\begin{aligned}
   \frac{\partial \mathcal{L}}{\partial N_j}
   &= S_{\mathrm{shift}} c_j
      - \frac{8\mu r\boldsymbol D_j}{S_{\mathrm{shift}}N_j^{2r+1}}
   = 0, \\[4pt]
   \frac{\partial \mathcal{L}}{\partial \mu}
   &= \sum_{j=1}^{J}\frac{4\boldsymbol D_j}{S_{\mathrm{shift}} N_j^{2r}}
      - \varepsilon^2
   = 0. 
\end{aligned}
\label{eq:FOC_condition_mult}
\right.
\end{equation}
From the first equation in \eqref{eq:FOC_condition_mult}, we obtain
\begin{equation}
    N_j =  \left(\frac{8\mu r \boldsymbol D_j} {S^{2}_{\mathrm{shift}}c_j}\right)^\frac{1}{2r+1},
\label{eq:N_j_mult}
\end{equation}
which yields \eqref{eq:raw_N_j}. Replacing \eqref{eq:N_j_mult} into the second
equation in \eqref{eq:FOC_condition_mult}, we find
\begin{equation}
    4\paren{8r\mu}^{-\tfrac{2r}{2r+1}}
    S_{\mathrm{shift}}^{\tfrac{2r-1}{2r+1}} S_1
    = \varepsilon^2,
\label{eq:mu_multi}
\end{equation}
Combining \eqref{eq:N_j_mult} with \eqref{eq:mu_multi},
we have \eqref{eq:min_work_mult}.

\subsection{Proof for Proposition \ref{prop:choosing_Nj}}\label{appendix:proof_prop_choosing_Nj}
From Theorem~\ref{theorem:superlinear_converge_SQP}, there exists an index
$J_{\mathrm{loc}}$ such that for all $j\ge J_{\mathrm{loc}}$,
\begin{equation}
\norm{\mathbf{z}^{(j)}-{\mathbf{z}}^{*}}
\leq
\eta_{j-1} \,
\norm{{\mathbf{z}}^{(j-1)}-{\mathbf{z}}^{*}},
\qquad \text{with }\ \eta_{j-1}\to 0.
\label{eq:superlinear_err}
\end{equation}
Applying the reverse triangle inequality and the triangle inequality yields, for all $j\ge J_{\mathrm{loc}}$,
\begin{equation}
\paren{1-\eta_{j-1}} \norm{{\mathbf{z}}^{(j-1)}-{\mathbf{z}}^{*}}
\leq
\norm{\mathbf{z}^{(j)} - \mathbf{z}^{(j-1)}}
\leq
\paren{1+\eta_{j-1}} \norm{{\mathbf{z}}^{(j-1)}-{\mathbf{z}}^{*}} .
\label{eq:increment_vs_error}
\end{equation}
Consequently, combining \eqref{eq:increment_vs_error} at indices $j$ and $j+1$ with \eqref{eq:superlinear_err} gives
\begin{equation}
\begin{aligned}
\frac{\norm{\mathbf{z}^{(j+1)} - \mathbf{z}^{(j)}}}{\norm{\mathbf{z}^{(j)} - \mathbf{z}^{(j-1)}}}
&\leq
\frac{\paren{1+\eta_{j}} \norm{{\mathbf{z}}^{(j)}-{\mathbf{z}}^{*}}}{\paren{1-\eta_{j-1}} \norm{{\mathbf{z}}^{(j-1)}-{\mathbf{z}}^{*}}} \\
&\leq
\underbrace{\frac{1+\eta_{j}}{1-\eta_{j-1}} \eta_{j-1}}_{:= \tilde \eta_{j-1}} .
\end{aligned}
\label{eq:step_ratio_bound}
\end{equation}
Since $\eta_{j-1}\to 0$ as $J\to\infty$, we also have $\tilde \eta_{j-1}\to 0$. For each $j\ge J_{\mathrm{loc}}$ there exists a finite constant
$C_{j-1}$ such that
\begin{equation}
\norm{\mathbf{z}^{(j)} - \mathbf{z}^{(j-1)}} \leq C_{j-1} \tilde \eta_{j-1}.
\label{eq:contraction_sol_bound}
\end{equation}
Combine \eqref{eq:contraction_sol_bound} with Assumption \ref{ass:diff_surrogate_Lipschitz}, for all $j\ge J_{\mathrm{loc}}$, we obtain
\begin{equation}
\begin{aligned}
\mathbb{E}\brac{\norm{{I}_{N_j,S_{\mathrm{shift}}}^{\mathrm{RQMC}} \brac{\Delta \mathcal{H}^{(1)}\paren{\cdot;\mathbf{z}^{(j)},\mathbf{z}^{(j-1)}}}}^2}
&\leq
C_{H,j-1} \eta_{j-1}^2.
\end{aligned}
\label{eq:second_moment_bound}
\end{equation}
where $C_{H,j-1}:= L_H C_{j-1}^2$. Moreover, for a random vector $\mathbf{X}\in\mathbb{R}^d$,
\[
\text{Var}\brac{\mathbf{X}}
=
\mathbb{E}\brac{\mathbf{X}\mathbf{X}^\top} - \mathbb{E}\brac{\mathbf{X}}\mathbb{E}\brac{\mathbf{X}^\top}
\preceq
\mathbb{E}\brac{\mathbf{X}\mathbf{X}^\top}.
\]
Using Jensen's inequality for the matrix norm, we have
\[
\norm{\text{Var}\brac{\mathbf{X}}}
\leq
\norm{\mathbb{E}\brac{\mathbf{X}\mathbf{X}^\top}}
\leq
\mathbb{E}\brac{\norm{\mathbf{X}\mathbf{X}^\top}}
=
\mathbb{E}\brac{\norm{\mathbf{X}}^2}.
\]
Applying this to
$\mathbf{X}={I}_{N_j,S_{\mathrm{shift}}}^{\mathrm{RQMC}} \brac{\Delta \mathcal{H}^{(1)}\paren{\cdot;\mathbf{z}^{(j)},\mathbf{z}^{(j-1)}}}$
and using \eqref{eq:second_moment_bound} gives
\begin{equation}
\boldsymbol{D}_j
=\norm{\text{Var}\brac{{I}_{N_j,S_{\mathrm{shift}}}^{\mathrm{RQMC}} \brac{\Delta \mathcal{H}^{(1)}\paren{\cdot;\mathbf{z}^{(j)},\mathbf{z}^{(j-1)}}}}}
\leq
\mathbb{E}\brac{\norm{{I}_{N_j,S_{\mathrm{shift}}}^{\mathrm{RQMC}} \brac{\Delta \mathcal{H}^{(1)}\paren{\cdot;\mathbf{z}^{(j)},\mathbf{z}^{(j-1)}}}}^2}
\leq
C_{H,j-1} \eta_{j-1}^2.
\label{eq:var_contract_mult}
\end{equation}
Using this variance contraction property of $\boldsymbol{D}_j$, we substitute it into \eqref{eq:raw_N_j} to obtain the expression in
\eqref{eq:N_j_mult}, with constant $C_{\mathrm{loc},j-1}$. Moreover, by fixing
$\eta_{j-1} = \eta$, we recover the formula in \eqref{eq:Nj_same_eta} with constant $C_{\mathrm{loc}}$. This
completes the proof.

\subsection{Proof for Proposition \ref{prop:work_multi_level}}\label{appendix:proof_prop_work_multi_level}
In the single-level RQMC setting, the uniform number of points
\(N_{\mathrm{sing}}\) is used for all iterations \(j\), and is derived from
\eqref{eq:statistical_error_RQMC_sol} as
\begin{equation}
    N_{\mathrm{sing}}
    =
    \left( \frac{\boldsymbol V_{\max}}{S_{\mathrm{shift}} \, \varepsilon^2} \right)^{\tfrac{1}{2r}},
\end{equation}
where \(\boldsymbol V_{\max} := \max_{1 \leq j \leq J} \boldsymbol V_j\),
\(\boldsymbol V_j := \norm{\mathrm{Var}\!\paren{ {I}_{N_1,S_{\mathrm{shift}}}^{\mathrm{RQMC}} \brac{\mathcal{H}^{(1)}\paren{\cdot;\mathbf{z}^{(j)}}}}}\),
and 
\[\boldsymbol V_1 = \boldsymbol D_1
= \norm{\mathrm{Var}\!\paren{ I_{N_1,S_{\mathrm{shift}}}^{\mathrm{RQMC}}
\brac{\mathcal H^{(1)}\paren{\cdot;\mathbf z^{(1)}}}}}.\]
The total work across all \(J\) iterations is
\begin{equation}
    W^{\text{RQMC}}_{\mathrm{sing}}(\varepsilon)
    = \sum_{j=1}^J S_{\mathrm{shift}} \, c \, N_{\mathrm{sing}}
    \approx J \, S_{\mathrm{shift}} \, \boldsymbol V_{\max}^{\tfrac{1}{2r}} \, \varepsilon^{-\tfrac{1}{r}}.
\label{eq:W_single}
\end{equation}
From \eqref{eq:min_work_mult}, we obtain
\begin{equation}
    \frac{W^{\text{RQMC}}_{\mathrm{sing}}(\varepsilon)}{W^{\text{RQMC}}_{\mathrm{mult}}(\varepsilon)}
    \approx \frac{J \, \boldsymbol V_{\max}^{\tfrac{1}{2r}}}{ \left(V_1^{\tfrac{1}{2r+1}} + \sum_{j=2}^{J_\mathrm{loc}-1} \boldsymbol D_j^{\tfrac{1}{2r+1}} + \sum_{j=J_\mathrm{loc}}^{J}\boldsymbol D_j^{\tfrac{1}{2r+1}}\right)^{\tfrac{2r+1}{2r}}}.
\end{equation}
 Using the contraction property for $\boldsymbol D_j$ from Proposition \ref{prop:choosing_Nj}, we obtain
\begin{equation}
\sum_{j=J_{\mathrm{loc}}}^{J} 
\boldsymbol D_j^{\tfrac{1}{2r+1}}
=
C_{\mathrm{loc}}^{\tfrac{1}{2r+1}}
\sum_{j=J_{\mathrm{loc}}}^{J} 
\!\left(\eta^{\tfrac{1}{2r+1}}\right)^{\!2j-2}
=
C_{\mathrm{loc}}^{\tfrac{1}{2r+1}}
\frac{
\eta^{\tfrac{2J_{\mathrm{loc}}-2}{2r+1}}
-
\eta^{\tfrac{2J}{2r+1}}
}{
1 - \eta^{\tfrac{2}{2r+1}}
}.
\label{eq:sum_A_convergence}
\end{equation}
Hence,
\[
\sum_{j=J_{\mathrm{loc}}}^{J} \boldsymbol D_j^{\tfrac{1}{2r+1}}
\xrightarrow[J \to \infty]{}
C_{\mathrm{loc}}^{\tfrac{1}{2r+1}}
\frac{
\eta^{\tfrac{2J_{\mathrm{loc}}-2}{2r+1}}
}{
1 - \eta^{\tfrac{2}{2r+1}}
}
=
\mathcal{O}(1),
\]
which implies \eqref{eq:W_sing_bound}. Finally, the multilevel work complexity \eqref{eq:W_mult_final} follows from Corollary \ref{eq:complexity_sing_RQMC_coro}.
\section{Supplementary results for Section \ref{subsec:opti_damp}}\label{Appendix:supp_optimal_damping}
\subsection{Damping rule and convexity properties}

\subsubsection{Proof for Corollary \ref{coro:optimal_damping}}\label{appendix:proof_optimal_damping}
Following the construction in Appendix~\ref{Appendix: Fourier_transform_loss}, the loss components
$\ell^{(\nu)}_{k,p}$ (for the exponential loss and QPC loss, excluding the linear term $\ell(x)=x$) are nonnegative,
and the marginal densities $f_{\mathbf X_{k,p}}$ are also nonnegative.
Therefore, by \cite[Proposition~3.2]{bayer_optimal_2023}, the associated Fourier factors satisfy the ridge property, i.e.,

\begin{align*}
   \abs{\hat{\ell}^{(\nu)}_{k,p}\paren{\mathbf{u} + \mathrm{i}\mathbf{K}_{k,p}^{(\nu)}}}  &\leq \abs{\hat{\ell}^{(\nu)}_{k,p}\paren{\mathrm{i}\mathbf{K}_{k,p}^{(\nu)}}}, \quad \forall \mathbf{u} \in \mathbb{R}^k, \, \mathbf{K}_{k,p}^{(\nu)} \in \delta_{\ell_{k,p}}^{(\nu)},\\
  \abs{ \Phi_{\mathbf{X}_{k,p}}\paren{\mathbf{u} + \mathrm{i}\mathbf{K}_{k,p}^{(\nu)}}}
    &\le
 \abs{\Phi_{\mathbf{X}_{k,p}}\paren{\mathrm{i}\mathbf{K}_{k,p}^{(\nu)}}},
  \qquad 
  \forall\, \mathbf{u} \in \mathbb{R}^k,\ \mathbf{K}_{k,p}^{(\nu)} \in \delta_{X_{k,p}}. 
\end{align*}
Combining these bounds with the definition of $h^{(\nu)}_{k,p}$ in \eqref{eq:loss_component_integrands} yields (B.1) and shows that  the supremum over $\mathbf u$ is attained at $\mathbf u=0$.
\begin{equation}\label{eq:f_bound}
\begin{aligned}
   \abs{h_{k,p}^{(\nu)}\paren{\mathbf{u},\mathbf{m}_{k,p},{\mathbf{K}_{k,p}^{(\nu)}},\boldsymbol \Theta_{k,p}}}
    &\leq (2\pi)^{-k} e^{\langle \mathbf{K}_{k,p}^{(\nu)},\mathbf{m}_{k,p} \rangle} \left| e^{-\mathrm{i}\langle \mathbf{u},\mathbf{m}_{k,p} \rangle} \right| \left| \mathbf{\Phi}_{\mathbf{X}_{k,p}}\paren{\mathrm{i}\mathbf{K}_{k,p}^{(\nu)}} \right|\left| \hat{\ell}^{(\nu)}_{k,p}\paren{\mathrm{i}\mathbf{K}_{k,p}^{(\nu)}} \right|\\ 
    &= \left|h_{k,p}^{(\nu)}\paren{\mathbf{0}_{\mathbb{R}^k}; \mathbf{m}_{k,p} ,\mathbf{K}_{k,p}^{(\nu)},\boldsymbol\Theta_{k,p}} \right|, \quad \forall \mathbf{u} \in \mathbb{R}^k, \, \mathbf{K}_{k,p}^{(\nu)} \in \delta_{K_{k,p}}^{(\nu)}. 
\end{aligned}
\end{equation}
Also, the quantity $\abs{h_{k,p}^{(\nu)}\paren{\mathbf{0}_{\mathbb{R}^k};\mathbf m_{k,p},\mathbf K_{k,p}^{(\nu)},\boldsymbol \Theta_{k,p}}}$ is strictly positive, so taking logarithms is valid and directly yields \eqref{eq:K_m_relate}. Since $\log(\cdot)$ is strictly increasing on $(0,\infty)$, the logarithmic transformation preserves the minimizer. This completes the proof.
\subsubsection{\texorpdfstring{On the Convexity of $\upsilon_{k,p}^{(\nu)}$ in in \eqref{eq:K_m_relate}}{Convexity of upsilon(k,p) (nu)}}
We need an additional assumption below to prove Proposition \ref{prop:psi_convex}
 \begin{enumerate}[label=(A\arabic*), ref=A\arabic*]
    \setcounter{enumi}{6}
    \item \label{ass:A9} For all $\mathbf{K}_{k,p}^{(\nu)} \in \delta^{(\nu)}_{K_{k,p}}$, we assume
\begin{equation*}
\int_{\mathbb{R}^k} \|\mathbf{x}\|^2 
e^{-\langle \mathbf{K}_{k,p}^{(\nu)}, \mathbf{x} \rangle}
f_{\mathbf{X}_{k,p}}(\mathbf{x}) \, d\mathbf{x}
< \infty, \quad \int_{\mathbb{R}^k} \|\mathbf{x}\|^2 
e^{\langle \mathbf{K}_{k,p}^{(\nu)}, \mathbf{x} \rangle}
\ell_{k,p}^{(\nu)}(\mathbf{x}) \, d\mathbf{x}
< \infty.
\end{equation*}
Moreover, for every compact set $C \subset \delta_{K_{k,p}}^{(\nu)}$ there exist
integrable functions $\varphi^{\mathbf{X}}_{k,p},\varphi^{\ell}_{k,p}:\R^k\to(0,\infty)$ such that,
for all $\mathbf K_{k,p}^{(\nu)}\in C$ and all $\mathbf x\in\R^k$,
\[
\|\mathbf x\|^2 e^{-\langle \mathbf K_{k,p}^{(\nu)},\mathbf x\rangle}\, f_{\mathbf X_{k,p}}(\mathbf x)
\le \varphi^{\mathbf{X}}_{k,p}(\mathbf x),
\qquad
\|\mathbf x\|^2 e^{\langle \mathbf K_{k,p}^{(\nu)},\mathbf x\rangle}\, \ell_{k,p}^{(\nu)}(\mathbf x)
\le \varphi^{\ell}_{k,p}(\mathbf x).
\]

\end{enumerate} 
\begin{proposition}\label{prop:psi_convex}
Suppose Assumptions \ref{ass:A7} and \eqref{ass:A9} hold. Define the normalized weights (Esscher transforms) as:
\[
w_{\boldsymbol \Phi_{k,p}}\paren{\mathbf x;\mathbf K_{k,p}^{(\nu)}} \;:=\; \frac{e^{-\langle \mathbf K_{k,p}^{(\nu)},\mathbf x\rangle} f_{\mathbf{X}_{k,p}}(\mathbf x)}{\int_{\mathbb{R}^k} e^{-\langle \mathbf K_{k,p}^{(\nu)},\mathbf x\rangle} f_{\mathbf{X}_{k,p}}(\mathbf x) \mathrm{d}\mathbf{x}},
\qquad
w_{\ell_{k,p}^{(\nu)}}\paren{\mathbf x;\mathbf K_{k,p}^{(\nu)}} \;:=\; \frac{e^{\langle \mathbf K_{k,p}^{(\nu)},\mathbf x\rangle} \ell_{k,p}^{(\nu)}(\mathbf x)}{\int_{\mathbb{R}^k} e^{\langle \mathbf K_{k,p}^{(\nu)},\mathbf x\rangle} \ell_{k,p}^{(\nu)}(\mathbf x) \mathrm{d}\mathbf{x}}
.\]
Then the Hessian of $\upsilon_{k,p}^{(\nu)}\!\left(\mathbf{m}_{k,p},\mathbf{K}_{k,p}^{(\nu)},\boldsymbol \Theta_{k,p}\right)$ (defined in \eqref{eq:K_m_relate}) is expressed as: \footnote{$\mathbb{E}_{w}\brac{\mathbf X},\operatorname{Cov}_{w}\brac{\mathbf X}$ denotes the expectation and covariance matrix of $\mathbf{X}$, respectively, under the given probability density $w$.}
\begin{equation}
\nabla^2_{\mathbf K_{k,p}^{(\nu)}} \upsilon_{k,p}^{(\nu)}\!\left(\mathbf{m}_{k,p},\mathbf{K}_{k,p}^{(\nu)},\boldsymbol \Theta_{k,p}\right)\;=\;\operatorname{Cov}_{w_{\boldsymbol \Phi_{k,p}}}\brac{\mathbf{X}_{k,p}}\;+\;\operatorname{Cov}_{w_{ \ell_{k,p}^{(\nu)}}}\brac{\mathbf{X}_{k,p}}\;\succeq\;0.
\label{eq:hessian_cov_sum}
\end{equation}
Moreover, $\upsilon_{k,p}^{(\nu)}(\mathbf m_{k,p},\mathbf K_{k,p}^{(\nu)},\boldsymbol \Theta_{k,p})$ is strictly convex on $\delta_{K_{k,p}}^{(\nu)}$ if, for every
$\mathbf K_{k,p}^{(\nu)}\in\delta_{K_{k,p}}^{(\nu)}$, at least one of the two covariance matrices
$\operatorname{Cov}_{w_{\boldsymbol \Phi_{k,p}}}[\mathbf X_{k,p}]$ or
$\operatorname{Cov}_{w_{\ell_{k,p}^{(\nu)}}}[\mathbf X_{k,p}]$
${\succ 0}$. If, in addition, there exists a compact set
$C \subset \delta_{K_{k,p}}^{(\nu)}$ and a constant $\mu_{k,p}>0$ such that \footnote{$\lambda_{\min}(\boldsymbol A)$ denotes the smallest eigenvalue of a square matrix $\boldsymbol A$.}
\[
\lambda_{\min}\!\Big(
  \nabla^2_{\mathbf K_{k,p}^{(\nu)}} \,\upsilon_{k,p}^{(\nu)}(\mathbf m_{k,p},\mathbf K_{k,p}^{(\nu)},\boldsymbol \Theta_{k,p})
\Big)
\;\ge\; \mu_{k,p},
\qquad \forall\,\mathbf K_{k,p}^{(\nu)}\in C,
\]
 then $\upsilon_{k,p}^{(\nu)}(\mathbf m_{k,p},\mathbf K_{k,p}^{(\nu)},\boldsymbol \Theta_{k,p})$
is $\mu_{k,p}$-strongly convex on $C$
\end{proposition}
\begin{proof}
   From Assumption
\ref{ass:A7}, the integrals
\[
\int_{\mathbb{R}^k} e^{-\langle \mathbf K_{k,p},\mathbf x\rangle}
      f_{\mathbf X_{k,p}}(\mathbf x)\,d\mathbf x,
\qquad
\int_{\mathbb{R}^k} e^{\langle \mathbf K_{k,p},\mathbf x\rangle}
      \ell_{k,p}^{(\nu)}(\mathbf x)\,d\mathbf x
\]
are finite for all $\mathbf{K}_{k,p}^{(\nu)} \in \delta_{K_{k,p}}^{(\nu)}$, combining with Assumption \eqref{ass:A9}, the normalized Esscher weights $w_{\boldsymbol\Phi_{k,p}}\paren{\mathbf x;\mathbf K_{k,p}^{(\nu)}}$ and
$w_{\ell_{k,p}^{(\nu)}}\paren{\mathbf x;\mathbf K_{k,p}^{(\nu)}}$ are well-defined probability densities on $\mathbb R^k$. Also, Assumption~\eqref{ass:A9} provides the uniform integrability needed to apply the dominated convergence theorem. Hence, the differentiation w.r.t. $\mathbf K_{k,p}^{(\nu)}$ can be passed through the integrals, and the first- and second-order derivatives are given by

\begin{equation*}
\nabla_{\mathbf{K}_{k,p}^{(\nu)}}\mathbf{\Phi}_{\mathbf{X}_{k,p}}\paren{\mathrm{i}\mathbf K_{k,p}^{(\nu)}} :=- \int_{\mathbb{R}^k}  \mathbf{x} e^{-\langle \mathbf K_{k,p}^{(\nu)},\mathbf x\rangle}f_{\mathbf{X}_{k,p}}(\mathbf x) \mathrm{d} \mathbf{x},\quad \nabla_{\mathbf{K}_{k,p}^{(\nu)}}\hat{\ell}^{(\nu)}_{k,p}\paren{\mathrm{i}\mathbf K_{k,p}^{(\nu)}} := \int_{\mathbb{R}^k}  \mathbf{x} e^{\langle \mathbf K_{k,p}^{(\nu)},\mathbf{x}\rangle}\ell_{k,p}^{(\nu)}(\mathbf x)\mathrm{d} \mathbf{x}
\end{equation*}

\begin{equation*}
\nabla^2_{\mathbf{K}_{k,p}^{(\nu)}}\mathbf{\Phi}_{\mathbf{X}_{k,p}}\paren{\mathrm{i}\mathbf K_{k,p}^{(\nu)}} := \int_{\mathbb{R}^k}  \mathbf{x}\mathbf{x}^T e^{-\langle \mathbf K_{k,p}^{(\nu)},\mathbf x\rangle}f_{\mathbf{X}_{k,p}}(\mathbf x) \mathrm{d} \mathbf{x}, \quad
\nabla^2_{\mathbf{K}_{k,p}^{(\nu)}}\hat{\ell}^{(\nu)}_{k,p}\paren{\mathrm{i}\mathbf K_{k,p}^{(\nu)}} := \int_{\mathbb{R}^k}  \mathbf{x}\mathbf{x}^T e^{\langle \mathbf K_{k,p}^{(\nu)},\mathbf{x}\rangle}\ell_{k,p}^{(\nu)}(\mathbf x)\mathrm{d} \mathbf{x}
\end{equation*}
By the nonnegativity of the loss components $\ell_{k,p}^{(\nu)}$ established in Appendix~\ref{Appendix: Fourier_transform_loss}, together with the nonnegativity of the marginal densities $f_{\mathbf{X}_{k,p}}$, we obtain
\begin{equation*}
    \ln \abs{\Phi_{\mathbf X_{k,p}}\paren{\mathrm i\,\mathbf K_{k,p}^{(\nu)}}} = \ln \Phi_{\mathbf X_{k,p}}\paren{\mathrm i\,\mathbf K_{k,p}^{(\nu)}}, \quad \ln \abs{\ln \widehat{\ell}^{(\nu)}_{k,p}\paren{\mathrm i\,\mathbf K_{k,p}^{(\nu)}}} = \ln \widehat{\ell}^{(\nu)}_{k,p}\paren{\mathrm i\,\mathbf K_{k,p}^{(\nu)}}
\end{equation*}
Now using the chain rule, we obtain
\begin{equation}
\begin{aligned}
\nabla^2_{\mathbf K_{k,p}^{(\nu)}} \ln \Phi_{\mathbf X_{k,p}}\paren{\mathrm i\,\mathbf K_{k,p}^{(\nu)}}
&=
\frac{
    \nabla^2_{\mathbf K_{k,p}^{(\nu)}} \Phi_{\mathbf X_{k,p}}\paren{\mathrm i\,\mathbf K_{k,p}^{(\nu)}}
}{
    \Phi_{\mathbf X_{k,p}}\paren{\mathrm i\,\mathbf K_{k,p}^{(\nu)}}
}
-
\frac{
    \nabla_{\mathbf K_{k,p}^{(\nu)}} \Phi_{\mathbf X_{k,p}}\paren{\mathrm i\,\mathbf K_{k,p}^{(\nu)}}
    \,
    \nabla_{\mathbf K_{k,p}^{(\nu)}} \Phi_{\mathbf X_{k,p}}\paren{\mathrm i\,\mathbf K_{k,p}^{(\nu)}}^{\top}
}{
    \Phi_{\mathbf X_{k,p}}^{2}\paren{\mathrm i\,\mathbf K_{k,p}^{(\nu)}}
},
\\[2mm]
\nabla^2_{\mathbf K_{k,p}^{(\nu)}} \ln \widehat{\ell}^{(\nu)}_{k,p}\paren{\mathrm i\,\mathbf K_{k,p}^{(\nu)}}
&=
\frac{
    \nabla^2_{\mathbf K_{k,p}^{(\nu)}} \widehat{\ell}^{(\nu)}_{k,p}\paren{\mathrm i\,\mathbf K_{k,p}^{(\nu)}}
}{
    \widehat{\ell}^{(\nu)}_{k,p}\paren{\mathrm i\,\mathbf K_{k,p}^{(\nu)}}
}
-
\frac{
    \nabla_{\mathbf K_{k,p}^{(\nu)}} \widehat{\ell}^{(\nu)}_{k,p}\paren{\mathrm i\,\mathbf K_{k,p}^{(\nu)}}
    \,
    \nabla_{\mathbf K_{k,p}^{(\nu)}} \widehat{\ell}^{(\nu)}_{k,p}\paren{\mathrm i\,\mathbf K_{k,p}^{(\nu)}}^{\top}
}{
    \paren{\widehat{\ell}^{(\nu)}_{k,p}}^{2}\paren{\mathrm i\,\mathbf K_{k,p}^{(\nu)}}
}.
\end{aligned}
\label{equation:second_deri_fourier_1}
\end{equation}
 We can rewrite \eqref{equation:second_deri_fourier_1} in terms of the Esscher transform as:
\begin{equation}
    \begin{aligned}
       \nabla^2_{\mathbf K_{k,p}^{(\nu)}} \ln \Phi_{\mathbf X_{k,p}}\paren{\mathrm{i}\mathbf K_{k,p}^{(\nu)}} &= \mathbb{E}_{w_{\boldsymbol \Phi_{k,p}}}\brac{\mathbf{X}_{k,p}\mathbf{X}_{k,p}^\top}- \mathbb{E}_{w_{\boldsymbol \Phi_{k,p}}}\brac{\mathbf{X}_{k,p}}\mathbb{E}_{w_{\boldsymbol \Phi_{k,p}}}\brac{\mathbf{X}_{k,p}}^\top = \operatorname{Cov}_{w_{\boldsymbol \Phi_{k,p}}}\brac{\mathbf{X}_{k,p}}\\
       \nabla^2_{\mathbf K_{k,p}^{(\nu)}} \ln \widehat{\ell}^{(\nu)}_{k,p}\paren{\mathrm{i}\mathbf K_{k,p}^{(\nu)}}  &= \mathbb{E}_{w_{\boldsymbol \Phi_{k,p}}}\brac{\mathbf{X}_{k,p}\mathbf{X}_{k,p}^\top}- \mathbb{E}_{w_{\ell_{k,p}^{(\nu)}}}\brac{\mathbf{X}_{k,p}}\mathbb{E}_{w_{\ell_{k,p}^{(\nu)}}}\brac{\mathbf{X}_{k,p}}^\top = \operatorname{Cov}_{w_{\ell_{k,p}^{(\nu)}}}\brac{\mathbf{X}_{k,p}}
    \end{aligned}
\end{equation}
Since $\operatorname{Cov}_{w}[\mathbf X] \succeq 0$ for any probability density $w$, \eqref{eq:hessian_cov_sum} follows immediately. Strict convexity holds whenever, for every $\mathbf K_{k,p}^{(\nu)}$, at least one of the two covariance terms is $\succ 0$. The strong convexity statement is exactly the uniform curvature bound for $\nabla^2_{\mathbf K_{k,p}^{(\nu)}} \,\upsilon_{k,p}^{(\nu)}(\mathbf m_{k,p},\mathbf K_{k,p}^{(\nu)},\boldsymbol \Theta_{k,p})$ on $C$.
\end{proof}
\begin{remark}
In our numerical setting, $\mathbf X_{k,p}$ is either a non-degenerate Gaussian $(\Sigma_{k,p}\succ 0)$ or a non-degenerate NIG with full-dimensional dispersion/shape $(\Gamma_{k,p}\succ 0)$. Under such non-degeneracy, the corresponding Esscher-tilted measures remain non-degenerate, hence $\mathrm{Cov}_{w_\Phi}[X_{k,p}]\succ 0$; see also \cite{prause_generalized_1999,gerber_option_1994}. Hence $v^{(\nu)}_{k,p}$ is strictly convex on $\delta^{(\nu)}_{K_{k,p}}$.
\label{rema:postive_definite_esscher}
\end{remark}

\subsection{Regularized damping}\label{appendix:regularization_damping}

This appendix complements the discussion in Section~\ref{subsec:opti_damp} by providing a principled regularization of the peak-minimizing damping rule. In order to handle the closed-to-boundary behavior of the minimizer, we need to account also for the width of the integrand around its peak $\mathbf{u} = \mathbf{0}_{\mathbb{R}^k}$, which controls for how fast the integrand decays away from $\mathbf{0}_{\mathbb{R}^k}$. Based on analysis about the asymptotic behavior of the integral around its peak from \cite[Chapter 4]{bruijn_asymptotic_2014}, a Taylor expansion of $\upsilon_{k,p}^{(\nu)}\paren{\mathbf{u};\mathbf{m}_{k,p},\mathbf{K}_{k,p}^{(\nu)},\textcolor{black}{\boldsymbol \Theta_{k,p}} }$ around $\mathbf{u}= \mathbf{0}_{\mathbb{R}^k}$ gives \footnote{From \eqref{eq:K_m_relate}, we recall $\upsilon_{k,p}^{(\nu)}\paren{\mathbf{0}_{\mathbb{R}^k};\mathbf{m}_{k,p},\mathbf{K}_{k,p}^{(\nu)},\textcolor{black}{\boldsymbol \Theta_{k,p}}} \equiv \upsilon_{k,p}^{(\nu)}\paren{\mathbf{m}_{k,p},\mathbf{K}_{k,p}^{(\nu)},\textcolor{black}{\boldsymbol \Theta_{k,p}}}$}  
\begin{equation}
    \upsilon_{k,p}^{(\nu)}\paren{\mathbf{u};\mathbf{m}_{k,p},\mathbf{K}_{k,p}^{(\nu)},\textcolor{black}{\boldsymbol \Theta_{k,p}} } = \upsilon_{k,p}^{(\nu)}\paren{\mathbf{m}_{k,p},\mathbf{K}_{k,p}^{(\nu)},\textcolor{black}{\boldsymbol \Theta_{k,p}}}-\frac{1}{2} \mathbf{u}^\top \kappa\paren{\mathbf{K}_{k,p}^{(\nu)}}\mathbf{u} +\mathcal{O}\left(\|\mathbf{u}\|^3\right)
\end{equation}
with $\kappa\paren{\mathbf{K}_{k,p}^{(\nu)}} = \nabla_{\mathbf{u}}^2 \upsilon_{k,p}^{(\nu)}\paren{\mathbf{u};\mathbf{m}_{k,p},\mathbf{K}_{k,p}^{(\nu)},\textcolor{black}{\boldsymbol \Theta_{k,p}} }|_{\mathbf{u}=\mathbf{0}_{\mathbb{R}^k}} \succeq 0$. 

Near $\mathbf{u} = \mathbf{0}_{\mathbb{R}^k}$, the component integrand of interest along any direction $\mathbf{u}$ can be approximated as
\begin{equation*}
    \bigl|h_{k,p}^{(\nu)}\paren{\mathbf{u};\mathbf{m}_{k,p},\mathbf{K}_{k,p}^{(\nu)},\textcolor{black}{\boldsymbol \Theta_{k,p}}}\bigr| 
    \;\approx\;
    \exp\!\left\{\upsilon_{k,p}^{(\nu)}\paren{\mathbf{m}_{k,p},\mathbf{K}_{k,p}^{(\nu)},\textcolor{black}{\boldsymbol \Theta_{k,p}}}\right\}
    \exp\!\left(\frac{1}{2} \mathbf{u}^\top \kappa\paren{\mathbf{K}_{k,p}^{(\nu)}}\mathbf{u}\right)
\end{equation*}
where $\kappa\paren{\mathbf{K}_{k,p}^{(\nu)}}$ denotes the local curvature (Hessian-type) matrix. 

Controlling the peak width required bounding the curvature relative to a 
reference geometry $\boldsymbol W_{k,p} \succ 0$. A minimum admissible width is enforced by the upper bound,
\begin{equation*}
    \kappa\paren{\mathbf{K}_{k,p}^{(\nu)}} \preceq r_{\max} \boldsymbol W_{k,p}
    \quad\Longleftrightarrow\quad
    \mathbf{u}^\top \kappa\paren{\mathbf{K}_{k,p}^{(\nu)}} \mathbf{u}
    \le r_{\max}\,\mathbf{u}^\top \boldsymbol W_{k,p} \mathbf{u},
    \label{eq:min_width_bound}
\end{equation*}
while a maximum admissible width follows from the lower bound
\begin{equation*}
    \kappa\paren{\mathbf{K}_{k,p}^{(\nu)}} \succeq r_{\min} \boldsymbol W_{k,p}
    \quad\Longleftrightarrow\quad
    \mathbf{u}^\top \kappa\paren{\mathbf{K}_{k,p}^{(\nu)}} \mathbf{u}
    \ge r_{\min}\,\mathbf{u}^\top \boldsymbol W_{k,p} \mathbf{u}.
    \label{eq:max_width_bound}
\end{equation*}
These curvature bounds can lead to a trust-region formulation for the damping selection:
\begin{equation*}
    \min_{\mathbf{K}_{k,p}^{(\nu)}} \;\upsilon\paren{\mathbf{m}_{k,p},\mathbf{K}_{k,p}^{(\nu)}}
    \quad\text{s.t.}\quad 
    \|\mathbf{K}_{k,p}^{(\nu)}\|_{\boldsymbol W_{k,p}} \le R,\;\; \mathbf{K}_{k,p}^{(\nu)}\in\delta_{K_{k,p}}^{(\nu)},
\end{equation*}
where 
\(
R := \lvert r_{\max}-r_{\min}\rvert, \qquad
\|\mathbf K_{k,p}^{(\nu)}\|_{\boldsymbol W_{k,p}}^{2}
:= \paren{\mathbf K_{k,p}^{(\nu)}}^{\top}\boldsymbol W_{k,p}\mathbf K_{k,p}^{(\nu)},
\qquad
\boldsymbol W_{k,p}\succ 0
\). 

The above derivation leads to the  Tikhonov-regularized problem \eqref{eq:tikhonov_penalized_damping}.
\begin{remark}
The natural choice for the weighting matrix $\boldsymbol W_{k,p}$ is given by the 
dispersion (or shape) matrix associated with the distribution of the 
 marginal loss vector $\mathbf{X}_{k,p}$. Concretely, for the models considered, 
$\boldsymbol W_{k,p}$ is chosen as follows:
\begin{table}[H]
    \centering
    \begin{tabular}{|c|c|}
        \hline
        \textbf{Model} & \textbf{Choice of $\boldsymbol W_{k,p}$} \\ \hline
        Gaussian & $\boldsymbol{\Sigma}_{k,p}$ \\ \hline
       NIG & $\boldsymbol{\Gamma}_{k,p}$ \\ \hline
    \end{tabular}
    \caption{Choice of the weighting matrix $\boldsymbol W_{k,p}$.}
    \label{tab:anisotropic_choice}
\end{table}
\label{note:anistropic_choice}
\end{remark}

\section{Supplementary results for Section \ref{subsubsec:domain_trans_QMC}}\label{Appendix:supp_results_domain_trans}

\subsection{Boundary Oscillation Analysis}\label{Appendix:boundary_oscillation}
This appendix analyzes boundary-induced oscillations of the transformed integrands in  Section~\ref{subsubsec:domain_trans_QMC}. For notational convenience, we write $ \;a_{k,p}^{(\nu)}\paren{G^{-1}\paren{\mathbf{v};\boldsymbol\Theta_{k,p}};\mathbf{m}_{k,p},\mathbf{K}_{k,p}^{(\nu)},\boldsymbol\Theta_{k,p}} \equiv a_{k,p}^{(\nu)}\left(G^{-1}\paren{\mathbf{v};\boldsymbol\Theta_{k,p}}\right)$. Fix a boundary face $ \mathbf B\subset\partial[0,1]^k$, and consider a Lipschitz path \(
\Upsilon:[t_0,t_1]\to[0,1]^k\), $\mathbf{v} := \Upsilon(t)$ 
approaching $\mathbf{B}$ (e.g., by fixing all but one coordinate and letting 
$v_j \to 0$). Throughout this appendix, we assume that the inverse transformation $G^{-1}(\cdot;\boldsymbol\Theta_{k,p})$ is almost everywhere differentiable on $(0,1)^k$ and locally absolutely continuous along Lipschitz paths $\Upsilon$. The Jacobian $J_{G^{-1}}(\mathbf v;\boldsymbol\Theta_{k,p})$ denotes the derivative of $G^{-1}$ w.r.t. $\mathbf v$.

The phase advance, or equivalently the total variation of 
$\varpi$, along the path $\Upsilon$ is given by
\begin{equation*}
    \begin{aligned}
       \mathrm{TV}_{[t_0,t_1]}(\varpi\circ\Upsilon)
&:=\int_{t_0}^{t_1} \left|\partial_{t}\;\varpi\paren{\Upsilon(t),\boldsymbol \Theta_{k,p}}\cdot\Upsilon'(t)\right|\,dt\\
&=\int_{t_0}^{t_1}\left|\big(J_{G^{-1}}\paren{\Upsilon(t),\boldsymbol \Theta_{k,p}}^\top \mathbf{m}_{k,p}\big)\cdot\Upsilon'(t)\right|\,dt. 
    \end{aligned}
\label{eq:total_variation_lip}
\end{equation*}
The \emph{number of oscillations accumulated along $\Upsilon$} is
\begin{equation*}
\begin{aligned}
    N_{\text{osc}}(\Upsilon;[t_0,t_1])
&:=\frac{1}{2\pi}\,\mathrm{TV}_{[t_0,t_1]}(\varpi\circ\Upsilon). \\
&= \frac{\int_{t_0}^{t_1}\left|\big(J_{G^{-1}}\paren{\Upsilon(t),\boldsymbol \Theta_{k,p}}^\top \mathbf{m}_{k,p}\big)\cdot\Upsilon'(t)\right|\,dt. }{2 \pi}\\
&\geq \frac{\left|\mathbf{m}_{k,p}^\top \brac{G^{-1}\paren{\Upsilon(t_1),\boldsymbol \Theta_{k,p}}-G^{-1}\paren{\Upsilon(t_0,\boldsymbol \Theta_{k,p}}}\right|}{2\pi}
\end{aligned}
\label{eq:nb_oscillation_lip}
\end{equation*}
Suppose that the amplitude of the integrand falls below the prescribed tolerance $\xi$ for $t\ge t^*$, where
\begin{equation}
    t^* := \inf\left\{t \in [t_0,t_1]:
    \abs{a\paren{G^{-1}\paren{\Upsilon ({t}),\boldsymbol\Theta_{k,p}}}}\,
    \abs{\det J_{G^{-1}}\paren{\Upsilon (t),\boldsymbol\Theta_{k,p}}}
    \le \xi \right\}.
\label{eq:amplitude_threshold}
\end{equation}
Let $\mathbf v^*:=\Upsilon(t^*)$ and define
\begin{align*}
    r_\mathbf{B}
    &:=\lim_{t\to t_\mathbf{B}}
    \frac{G^{-1}\paren{\Upsilon (t),\boldsymbol\Theta_{k,p}}}
         {\norm{G^{-1}\paren{\Upsilon (t),\boldsymbol\Theta_{k,p}}}}
     \;=\;
     \lim_{\mathbf{v}\to \mathbf{B}}
    \frac{G^{-1}\paren{\mathbf{v},\boldsymbol\Theta_{k,p}}}
         {\norm{G^{-1}\paren{\mathbf{v},\boldsymbol\Theta_{k,p}}}},\\
    U_*
    &:=\norm{G^{-1}\paren{\Upsilon(t^*),\boldsymbol \Theta_{k,p}}}
     \;=\;
     \norm{G^{-1}\paren{\mathbf{v}^*,\boldsymbol \Theta_{k,p}}}.
\end{align*}
Here, $r_\mathbf{B}$ is the \emph{normalized direction} of $G^{-1}(\mathbf v,\boldsymbol\Theta_{k,p})$ as $\mathbf v$ approaches the boundary face $\mathbf B$, and $U_*$ is the \emph{truncation radius} at the threshold point $\mathbf v^*=\Upsilon(t^*)$.

Moreover, if the phase $\varpi$ is monotone in a neighborhood of $\mathbf B$ along the path $\Upsilon(t)$ and $G^{-1}(\Upsilon(t))$ is asymptotically radial, then the number of oscillations near the boundary can be approximated by: \footnote{By 
symmetry, the same argument applies when $\mathbf{v} \to \mathbf{1}_{\mathbb{R}^k}$; for definiteness, we 
consider here the limiting behavior as $\mathbf{v} \to \mathbf{0}_{\mathbb{R}^k}$.}
\begin{equation}
    N_{\text{osc}}(\mathbf{B}) \;\approx\; 
\frac{1}{2\pi}\,\bigl|\mathbf{m}_{k,p}^\top G^{-1}(\mathbf{v}^*,\boldsymbol{\Theta})\bigr|
\;=\;
\frac{1}{2\pi}\,\bigl|\mathbf{m}_{k,p}^\top r_\mathbf{B}\bigr|\,U_*,
\end{equation}
For related analysis of multivariate oscillatory integrals, emphasizing the role of stationary points $\paren{\partial_{t}\;\varpi(\Upsilon(t),\boldsymbol \Theta_{k,p})\cdot\Upsilon'(t) = 0}$, we refer to \cite{huybrechs_construction_2007,iserles_computation_2006}. Our notion of quantifying the number of  oscillations via phase advance along a given path  follows the same spirit as steepest descent based algorithms; see, for example, \cite{gibbs_numerical_2024}.

\begin{remark}[Boundary–Admissibility Condition (BAC)]\
In order to analyze the oscillatory behavior near the boundaries, we must ensure that the mapped amplitude $a_{k,p}^{(\nu)}$ decays appropriately as we approach the hypercube boundary $\mathbf{B}$. Otherwise, by \eqref{eq:amplitude_threshold}, the threshold point  $\mathbf{v}^*$ may not exist. Thus, for every $\mathbf{v}\to\mathbf{B}$, we require
\begin{equation}
   \limsup_{\mathbf{v}\to \mathbf{B}} 
\;\abs{a_{k,p}^{(\nu)}\!\left(G^{-1}\paren{\mathbf{v},\boldsymbol{\Theta}_{k,p}}\right)\,
   }\abs{\det J_{G^{-1}}\paren{\mathbf{v},\boldsymbol\Theta_{k,p}}}
   \;<\;\infty.
\end{equation}
\label{remark:BAC_condition}
\end{remark}
The BAC condition ensures that the envelope does not counteract the decay induced by the extended CF, thereby guaranteeing the existence of a truncation radius and a finite oscillation count.

In Appendix~\ref{Appendix: Fourier_transform_loss}, we show that $\widehat{\ell}^{(\nu)}_{k,p}(\mathbf u+i \mathbf K)$ decays at most polynomially in $\|\mathbf u\|$ for fixed admissible $ \mathbf K$. By contrast, for the Gaussian and NIG models considered in Section \ref{sec:mapping_MSRM_Fourier}, the extended CF, $\Phi_{\mathbf X_{k,p}}(\mathbf u+i \mathbf K)$ admits exponential-type decay in $\|u\|$.  Consequently, as $\mathbf v\to  \mathbf B$ and $\|G^{-1}(\mathbf v)\|\to\infty$, the boundary behavior of the product $\Phi_{X_{k,p}}(G^{-1}(\mathbf v)+i \mathbf K)\,\widehat{\ell}^{(\nu)}_{k,p}(G^{-1}(\mathbf v)+i \mathbf K)$ is governed by $\Phi_{\mathbf X_{k,p}}$, up to polynomial factors.

So, in order to characterize the oscillatory behavior for our component integrands, we examine the behavior of its corresponding extended CF. The expression for $N_{\text{osc}}$ near the boundaries, for a given loss distribution, $\mathbf{X}$ is provided in Table \ref{tab:nb_oscillatory}. A detailed derivation is provided in Appendix \ref{Appendix:trans_specific_loss_distr}.
\begin{table}[H]
\centering
\renewcommand{\arraystretch}{2.5}
\begin{tabular}{|c|c|}
\hline
\textbf{Distribution} & $N_{\mathrm{osc}}(\mathbf{v})$ \\ \hline
\textbf{Gaussian}& $\displaystyle C(\mathbf{v}^*)\frac{|\mathbf{m}_{k,p}^\top r_{\mathbf{B}}|}{\sqrt{2}\pi} \sqrt{\frac{1}{\lambda_{\min}(\boldsymbol \Sigma_{k,p})}} $ \\ \hline
\textbf{NIG} & $\displaystyle C(\mathbf{v}^*)\frac{|\mathbf{m}_{k,p}^\top r_{\mathbf{B}}|}{2\pi} \frac{1}{\delta_{k,p} \sqrt{\lambda_{\min} (\boldsymbol \Gamma_{k,p})}}$ \\ \hline
\end{tabular}
\caption{Asymptotic scaling of the oscillation count $N_{\mathrm{osc}}$ near a boundary face $\mathbf B$ for the Gaussian and NIG reference models (see Appendix \ref{Appendix:trans_specific_loss_distr}). In the Gaussian case, $\boldsymbol\Sigma_{k,p}$ is defined in Example~\ref{ex:Gaussian_components}. In the NIG case, $\boldsymbol\Gamma_{k,p}$ and $\delta_{k,p}$ are defined in Example~\ref{ex:MNIG_components}.}
\label{tab:nb_oscillatory}
\end{table}
\begin{equation*}
C(\mathbf v^*)
=
\begin{cases}
\displaystyle
\sqrt{\log\!\left(\dfrac{C\,\bigl|\det J_{G^{-1}}(\mathbf v^*,\boldsymbol\Theta_{k,p})\bigr|}{\xi}\right)},
& \text{if } \mathbf X \sim \mathcal N, \\[8pt]
\displaystyle
\log\!\left(\dfrac{C\,\bigl|\det J_{G^{-1}}(\mathbf v^*,\boldsymbol\Theta_{k,p})\bigr|}{\xi}\right),
& \text{if } \mathbf X \sim \mathcal{NIG}.
\end{cases}
\end{equation*}
\begin{remark}
The estimates in Table \ref{tab:nb_oscillatory} distinguish   two cases:  
for the Gaussian case, the decay of the CF damps oscillations much more rapidly, while in the NIG setting, the slower exponential decay allows oscillatory behavior to persist over a wider region near the boundary. 
\label{remark: damps_oscillation_count} 
\end{remark}
\begin{remark}
The factor 
\(
    \frac{1}{\sqrt{\lambda_{\min}(\,\cdot\,)}} 
\)
originates from the truncation radius $U_{*}$. As $\lambda_{\min}(\,\cdot\,) \to 0$, the envelope decays increasingly slowly, which causes both $U_{*}$ and, consequently, $N_{\text{osc}}$ to diverge. In the degenerate case $\lambda_{\min}(\,\cdot\,)=0$, there is no decay along at least one direction. Hence $U_{*}$ is not finite and $N_{\text{osc}}$ cannot be estimated using formulas in Table \ref{tab:nb_oscillatory}. In this situation, the integrand exhibits non-decaying tails and might have piecewise ``kinked'' behavior near the boundary.
\label{remark:eigen_0_boundary}
\end{remark}
From Remarks \ref{remark:BAC_condition}, \ref{remark:eigen_0_boundary} together with 
Table \ref{tab:nb_oscillatory}, it follows that 
the oscillatory behavior of the integrand can be improved by:
\begin{itemize}
    \item reducing the projection of $\abs{\mathbf{m}_{k,p}^\top r_\mathbf{B}}$ onto the 
    preimage direction $r_\mathbf{B}$;
    \item matching the decay of the determinant at $\mathbf{v}^*$,
    $\abs{\det J_{G^{-1}}(\mathbf{v}^*,\boldsymbol{\Theta}_{k,p})}$ with the decay of the extended CF $\Phi_{\mathbf{X}_{k,p}}$, while ensuring that the BAC condition holds 
    and that the smallest eigenvalue $\lambda_{\min}$ of the dispersion matrices 
    $(\boldsymbol{\Sigma}_{k,p}, \boldsymbol{\Gamma}_{k,p})$ does not deteriorate (i.e.\ $\lambda_{\min} \not\to 0$).
\end{itemize}
Regarding the former point, we always have $\abs{\mathbf{m}_{k,p}^\top r_\mathbf{B}} \leq \norm{\mathbf{m}_{k,p}}\norm{r_{\mathbf{B}}} =\norm{\mathbf{m}_{k,p}}$, and $\mathbf{m}_{k,p}$ is determined through the optimization procedure, so what matters more for our transformation is the latter point, which provides the rationale for the density-driven transformation mentioned in Section \ref{subsubsec:domain_trans_QMC}.
\subsubsection{Approximating the number of oscillation for specific loss distributions}\label{Appendix:trans_specific_loss_distr}

\paragraph{The Gaussian case. }\mbox{}
\noindent Using the formula from Example \ref{ex:Gaussian_components} in  Appendix \ref{appendix:loss_vector_distr} for the extended CF of the component, when $\mathbf{v} \to \mathbf{B}$ we have
\[
\abs{a_{k,p}^{(\nu)}\paren{G^{-1}(\mathbf v,\boldsymbol\Theta_{k,p})}}\,
\abs{\det J_{G^{-1}}(\mathbf v,\boldsymbol\Theta_{k,p})}
\;\le\;
C \big|\det J_{G^{-1}}(\mathbf v,\boldsymbol\Theta_{k,p})\big|
\exp\!\left(
-\frac{1}{2}\lambda_{\min}(\boldsymbol\Sigma_{k,p})
\abs{G^{-1}(\mathbf v,\boldsymbol\Theta_{k,p})}^{2}
\right),
\]
where $C>0$. Hence, the truncation radius $U_*$ in
\eqref{eq:amplitude_threshold} satisfies
\[
U_*
=
\sqrt{
\frac{2}{\lambda_{\min}(\boldsymbol\Sigma_{k,p})}
\log\!\frac{
C \big|\det J_{G^{-1}}(\mathbf v^*,\boldsymbol\Theta_{k,p})\big|
}{\xi}
}.
\]
Therefore, the number of oscillations near the boundary
\begin{equation}
N_{\text{osc}}(\mathbf v^*)
\;\approx\;
\frac{\abs{\mathbf m_{k,p}^{\top} r_{\mathbf B}}}{2\pi}
\sqrt{
\frac{2}{\lambda_{\min}(\boldsymbol\Sigma_{k,p})}
\log\!\frac{
C \big|\det J_{G^{-1}}(\mathbf v^*,\boldsymbol\Theta_{k,p})\big|
}{\xi}
}.
\label{eq:number_osc_MN}
\end{equation}

\paragraph{The NIG case.} \mbox{} Using the formula from Example \ref{ex:MNIG_components} in Appendix \ref{appendix:loss_vector_distr}, when
$\mathbf v \to \mathbf B$ we have that the extended CF of the component satisfies
\begin{align*}
\Phi_{\mathbf X_{k,p}}
\paren{G^{-1}(\mathbf v,\boldsymbol\Theta_{k,p})+\mathrm i\,\mathbf K_{k,p}}
&\le
C\exp\!\left(
-\,\delta_{k,p}\,
\abs{\boldsymbol\Gamma_{k,p}^{1/2}
\paren{G^{-1}(\mathbf v,\boldsymbol\Theta_{k,p})}}
\right)
\\
&\le
C\exp\!\left(
-\,\delta_{k,p}\,
\sqrt{\lambda_{\min}(\boldsymbol\Gamma_{k,p})}\,
\abs{G^{-1}(\mathbf v,\boldsymbol\Theta_{k,p})}
\right).
\end{align*}
\[
\Rightarrow\quad
\begin{aligned}
&\abs{a_{k,p}^{(\nu)}\paren{G^{-1}(\mathbf v,\boldsymbol\Theta_{k,p})}}\,
\abs{\det J_{G^{-1}}(\mathbf v,\boldsymbol\Theta_{k,p})} \\
&\qquad\le\;
C\,\big|\det J_{G^{-1}}(\mathbf v,\boldsymbol\Theta_{k,p})\big|\,
\exp\!\left(
-\,\delta_{k,p}\,
\sqrt{\lambda_{\min}(\boldsymbol\Gamma_{k,p})}\,
\abs{G^{-1}(\mathbf v,\boldsymbol\Theta_{k,p})}
\right).
\end{aligned}
\]
Hence the truncation radius $U_*$ in \eqref{eq:amplitude_threshold} satisfies
\[
U_*\;=\;\frac{1}{\delta_{k,p}\,\sqrt{\lambda_{\min}(\boldsymbol \Gamma_{k,p})}}\,
\log\!\frac{C\big|\det J_{G^{-1}}(\mathbf{v}^*,\boldsymbol\Theta_{k,p})\big|}{\xi}.
\]
Consequently,
\begin{equation}
   N_{\text{osc}}(\mathbf{v}^*)\;\approx\;\frac{|\mathbf{m}_{k,p}^\top r_{\mathbf{B}}|}{2\pi}\,
\frac{1}{\delta_{k,p}\sqrt{\lambda_{\min}(\boldsymbol \Gamma_{k,p})}}\,
\log\!\frac{C\big|\det J_{G^{-1}}(\mathbf{v}^*,\boldsymbol\Theta_{k,p})\big|}{\xi}.
\label{eq:number_osc_MNIG}
\end{equation}

\subsubsection{Choice of matrix \texorpdfstring{$\tilde{\boldsymbol{\Sigma}}_{k,p}$}{Sigma-tilde (k,p)}}\label{Appendix:choice_of_trans_matrix}
Let 
\(
    \mathbf{y} := \Psi^{-1}(\mathbf{v}_{1:k},\boldsymbol{I}_k).
\)
\paragraph{The Gaussian case.} \mbox{}
A  change of variables yields
\begin{equation}
    \Bigl|\det J_{\mathcal{T}_{\text{Gauss}}(\mathbf{v})}\Bigr|
    \;=\; \underbrace{\,|\det\tilde {\boldsymbol L}_{k,p}|}_{\partial \mathbf u/\partial \mathbf y} \cdot \underbrace{\prod_{i=1}^k \frac{1}{\varphi(y_i)}}_{\mathbf y=\Psi^{-1}_{\text{Gauss}}(\mathbf v_{1:k})} 
    \;\propto\; \left|\det \tilde {\boldsymbol L}_{k,p}\right|
    \exp\!\left(\tfrac{1}{2}\|\mathbf{y}\|^2\right).
\end{equation}

\medskip
\noindent
Accordingly, the truncation radius $U_*$ becomes
\begin{equation}
\begin{aligned}
   \abs{a_{k,p}^{(\nu)}\!\left(\mathcal{T}_{\text{Gauss}}(\mathbf{v}^*)\right)}  \abs{\det J_{\mathcal{T}_{\text{Gauss}}}(\mathbf{v}^*)} 
    & \;\propto\; \left|\det \tilde {\boldsymbol L}_{k,p}\right| \exp\!\left(\tfrac{1}{2}\|\mathbf{y^*}\|^2\right) \Phi_{\mathbf{X}_{k,p}}^{\mathrm{Gauss}}\paren{\tilde {\boldsymbol L}_{k,p} \mathbf{y^*}+ \mathrm{i} \mathbf{K}_{k,p}^{(\nu)}}  \\
    &\;\propto\; 
    \left|\det \tilde {\boldsymbol L}_{k,p}\right|
    \exp\!\left[-\tfrac{1}{2}\,\mathbf{y^*}^\top 
        \paren{\tilde {\boldsymbol L}_{k,p}^{\top} \boldsymbol\Sigma_{k,p} \tilde {\boldsymbol L}_{k,p} - \boldsymbol I_k}\,
        \mathbf{y^*}\right] \\
    &\;=\;\left|\det \tilde {\boldsymbol L}_{k,p}\right|
    \exp\!\left[-\tfrac{1}{2}\,\mathbf{u^*}^\top 
        \paren{\boldsymbol\Sigma_{k,p} - \tilde{\boldsymbol\Sigma}_{k,p}^{-1}}\,
        \mathbf{u^*}\right],
\end{aligned}
\label{eq:bounding_envelop_Gauss}
\end{equation}
\noindent To satisfy the BAC condition (Remark \ref{remark:BAC_condition}), it is necessary that
\begin{equation}
    \boldsymbol\Sigma_{k,p} - \tilde{\boldsymbol\Sigma}_{k,p}^{-1} \succeq 0
\end{equation}
, which coincides with the condition derived in \cite[Section~3.2.1]{bayer_quasi-monte_2025} for the GBM model with $T=1$.

\noindent
The number of oscillations near the boundary is given by
\begin{equation}
    N_{\text{osc}}^{\text{Gauss}} \;\propto\;
    \frac{\| \mathbf{m}_{k,p}\|}
    {\sqrt{\lambda_{\min}\!\bigl(\boldsymbol\Sigma_{k,p} - \tilde{\boldsymbol\Sigma}_{k,p}^{-1}\bigr)}} 
    \;\sqrt{{C}}.
\label{eq:gaussian_osc}
\end{equation}
To avoid the degeneracy for $\lambda_{\min}$ (Remark \ref{remark:eigen_0_boundary}), we require
\[
    \boldsymbol\Sigma_{k,p} \succ \tilde{\boldsymbol\Sigma}_{k,p}^{-1}.
\]

\medskip
\noindent
A convenient parametrization is
\begin{equation}
    \tilde{\boldsymbol\Sigma}_{k,p}^{-1} = \frac{1}{c}\,\boldsymbol\Sigma_{k,p},
    \qquad c>1.
\label{eq:choice_transform_MN}
\end{equation}
\paragraph{The NIG case.} \mbox{}
We have:
\begin{equation}
\begin{aligned}
    \abs{\det J_{\mathcal{T}_{\mathrm{NIG}}}(\mathbf v)}
\;&=\;
\underbrace{w^{k/2}\,\abs{\det\tilde {\boldsymbol L}_{k,p}}}_{\partial \mathbf u/\partial \mathbf y}
\cdot
\underbrace{\prod_{i=1}^k \frac{1}{\varphi(y_i)}}_{\mathbf y=\Phi_{\mathrm{Gauss}}^{-1}(\mathbf v_{1:k})}
\cdot
\underbrace{\frac{1}{\psi_W(w)}}_{w=\Psi_W^{-1}(v_{k+1})}\\
& \propto w^{k/2}\,\abs{\det\tilde {\boldsymbol L}_{k,p}}\,
   \frac{\exp\!\left(\tfrac{1}{2}\|\mathbf{y}\|^2\right)}{\psi_W(w)}\\
& \propto w^{k/2}\,\abs{\det\tilde {\boldsymbol L}_{k,p}}\,
   \frac{\exp\!\left(\tfrac{1}{2w}\,\mathbf{u}^\top \tilde{\boldsymbol\Sigma}_{k,p}^{-1}\,\mathbf{u}\right)}{\psi_W(w)}.
\end{aligned}
\label{eq:MNIG_det_v}
\end{equation}
Define the \emph{effective determinant} $\mathcal{D}_U(\mathbf{u})$ for the transform of $\mathbf{u}$ as
\begin{equation}
    \mathcal{D}_U(\mathbf{u})
    := \left[\int_{\mathbb{R}^{+}} 
    \frac{dw}{\big|\det J_{\mathcal{T}_{\mathrm{NIG}}}(\mathbf v)\big|}\right]^{-1}.
\label{eq:inverse_laplace_density}
\end{equation}
This quantity represents the inverse of the multivariate Laplace density $\psi^{\mathrm{lap}}$. In fact,
\begin{equation}
    \mathcal{D}_U(\mathbf{u}) = \frac{1}{\psi^{\mathrm{lap}}(\mathbf{u)}} \;\propto\; 
    \frac{|\det\tilde {\boldsymbol L}_{k,p}|}{ K_\lambda\!\left(\sqrt{2\,\mathbf{u}^\top \tilde{\boldsymbol\Sigma}_{k,p}^{-1}\,\mathbf{u}}\right)},
\end{equation}
where $K_\lambda(\cdot)$ denotes the modified Bessel function of the third kind. The truncation radius at $U_*$ then becomes
\begin{equation}
    \abs{a_{k,p}^{(\nu)}\!\left(\mathbf{u}^*\right)}  \abs{\mathcal{D}_U(\mathbf{u}^*)} 
    \;\propto\; 
    \frac{\abs{\det\tilde {\boldsymbol \Sigma}_{k,p}}^{\tfrac{1}{2}}}{ K_\lambda\!\left(\sqrt{2\,\mathbf{u^*}^\top \tilde{\boldsymbol\Sigma}_{k,p}^{-1}\,\mathbf{u^*}}\right)} 
    \,\times\, \Phi_{\mathbf{X}_{k,p}}^{\mathrm{NIG}}\paren{\mathbf{u^*}+\mathrm{i} \mathbf{K}_{k,p}^{(\nu)}}.
\end{equation}
Following the same arguments as in \cite[Section~3.2.2]{bayer_quasi-monte_2025} with $T=1$, we can have:
\begin{equation}
     \abs{a_{k,p}^{(\nu)}\!\left(\mathbf{u}^*\right)}  \abs{\mathcal{D}_U(\mathbf{u}^*)} 
    \;\propto\; \abs{\det\tilde {\boldsymbol \Sigma}_{k,p}}^{\tfrac{1}{2}} \exp \brac{-\paren{\delta_{k,p} \sqrt{\mathbf{u^*}^\top \boldsymbol{\Gamma}_{k,p}\mathbf{u^*}}-\sqrt{2\mathbf{u^*}^\top\tilde {\boldsymbol{\Sigma}}^{-1}_{k,p}\mathbf{u^*}}}}
\label{eq:bounding_envelope_NIG}
\end{equation}
To satisfy the BAC condition (Remark \ref{remark:BAC_condition}), it is required that
\begin{equation}
    \delta_{k,p}^2 \boldsymbol \Gamma_{k,p}- 2 \tilde{\boldsymbol\Sigma}_{k,p}^{-1} \succeq 0.
\end{equation}
Analogous to the Gaussian case in \eqref{eq:gaussian_osc}, 
we impose the strict inequality
\begin{equation*}
    \delta_{k,p}^2 \boldsymbol \Gamma_{k,p}\succ  2 \tilde{\boldsymbol\Sigma}_{k,p}^{-1},
\end{equation*}
and choose the parametrization
\begin{equation}
    \tilde{\boldsymbol\Sigma}_{k,p}^{-1} = \frac{\delta_{k,p}^2 \boldsymbol \Gamma_{k,p}}{2c},
    \qquad c>1.
\label{eq:transform_MNIG_case}
\end{equation}

\section{Convergence analysis for the SAA method}\label{appendix:converge_analysis_SAA}

Before presenting the convergence analysis of the SAA method, we introduce the notation and estimators used throughout this appendix. Let $\{\mathbf X^{(i)}\}_{i=1}^{N}$ be i.i.d.\ copies of $\mathbf X$, drawn once and kept fixed throughout the optimization. The MC estimators of $g, \nabla g$ and $\nabla^2g$ are defined by
\begin{equation}
\begin{aligned}
   \hat g^{(0),\mathrm{SAA}}_{N}(\mathbf m)
    &:= \frac{1}{N}\sum_{i=1}^{N} \ell\!\big(\mathbf X^{(i)}-\mathbf m\big), \\[4pt]
    \hat g^{(1),\mathrm{SAA}}_{N}(\mathbf m)
    &:= \frac{1}{N}\sum_{i=1}^{N} \nabla\ell\!\big(\mathbf X^{(i)}-\mathbf m\big),\\
     \hat g^{(2),\mathrm{SAA}}_{N}(\mathbf m)
    &:= \frac{1}{N}\sum_{i=1}^{N} \nabla^2\ell\!\big(\mathbf X^{(i)}-\mathbf m\big).
\end{aligned}
\label{eq:SAA_estimate}
\end{equation}
and the corresponding MC estimators for the Lagrangian as:
\[
\mathcal{L}_{N}^{(0),\mathrm{SAA}}(\mathbf z)
:= \sum_{k=1}^d m_k + \lambda\,\hat g^{(0),\mathrm{SAA}}_{N}(\mathbf m),
\qquad
\mathcal{L}_{N}^{(1),\mathrm{SAA}}(\mathbf z)
:=
\begin{bmatrix}
\mathbf 1 - \lambda\,\hat g^{(1),\mathrm{SAA}}_{N}(\mathbf m)\\[3pt]
-\,\hat g^{(0),\mathrm{SAA}}_{N}(\mathbf m)
\end{bmatrix},
\]
and
\[
\mathcal{L}_{N}^{(2),\mathrm{SAA}}(\mathbf z)
:=
\begin{bmatrix}
-\lambda\,\hat g^{(2),\mathrm{SAA}}_{N}(\mathbf m)
&
\big(\hat g^{(1),\mathrm{SAA}}_{N}(\mathbf m)\big)^{\!\top}\\[3pt]
\hat g^{(1),\mathrm{SAA}}_{N}(\mathbf m)) & 0
\end{bmatrix}.
\]

To establish consistency and asymptotic efficiency of the SAA solution in
Propositions~\ref{prop:consistency_SAA_solution} and~\ref{prop:efficiency-SAA}, we adopt the framework in \cite[Section~5.2]{shapiro_lectures_2014} and impose the following condition.
\begin{assumption}\
\begin{enumerate}
    \item  For every $\mathbf{m_0} \in \mathbb{R}^d$, the mapping $\mathbf{m} \mapsto \nabla^2\ell(\mathbf{X-m})$ is continuous at $\mathbf{m}_0$ a.s.
    \item There exist integrable random variables $D_0,D_1,D_2$ such that, for all $\mathbf m\in M$ a.s.,
\[
|\ell(\mathbf X-\mathbf m)|\le D_0,\qquad
\|\nabla\ell(\mathbf X-\mathbf m)\|\le D_1,\qquad
\|\nabla^2\ell(\mathbf X-\mathbf m)\|\le D_2.
\]
\end{enumerate}
\label{assump:SAA_dominate_integrable}
\end{assumption}

\begin{proposition}[Consistency of the SAA solution]
\label{prop:consistency_SAA_solution}
Suppose Assumptions \ref{assump:regularity_exact_KKT} and \ref{assump:SAA_dominate_integrable} hold. Then, as $N \to \infty$, the SAA problem admits a unique optimal solution
$\mathbf z_{N}^{\mathrm{SAA},*} \in \mathcal{Z}$, and
\[
\mathbf z_{N}^{\mathrm{SAA},*}
\xrightarrow[]{\mathrm{a.s.}}
\mathbf z^* .
\]
\end{proposition}
\begin{proof}
   The proof follows the same line of arguments as that of Theorem~\ref{theorem:consistentcy-RQMC-Fou}. Uniform convergence of the SAA estimators is established in \cite[Proposition~7]{shapiro_monte_2003} ; here, we adapt this result to the functions $\ell$, $\nabla \ell$, and $\nabla^{2}\ell$ under the given norm $\norm{.}$. With these uniform convergence properties in place, the Banach fixed-point argument applies in the same manner as in the Fourier-RQMC case. 
\end{proof}

\begin{proposition}[Asymptotic behavior of the SAA solution]
\label{prop:efficiency-SAA}
Suppose Assumptions \ref{assump:regularity_exact_KKT} and \ref{assump:SAA_dominate_integrable} hold.
Then, as $N \to \infty$, the SAA solution satisfies
\begin{equation}
    \sqrt{N}
    \bigl(
        \mathbf{z}_{N}^{\mathrm{SAA},*}
        -
        \mathbf{z}^{*}
    \bigr)
    \xrightarrow{\mathrm{law}}
    \mathcal{N}\!\left(
        \mathbf{0},
        \boldsymbol V(\mathbf{z}^{*})
    \right),
    \label{eq:CLT_solution_SAA_1}
\end{equation}
where the asymptotic covariance matrix is given by
\begin{equation}
    \boldsymbol V(\mathbf{z}^{*})
    =
    \Bigl(
        {{\nabla_{\mathbf{z}}^{2}}\mathcal{L}}
        \bigl(\mathbf{z}^{*}\bigr)
    \Bigr)^{-1}
    \boldsymbol H(\mathbf{z}^{*})
    \Bigl(
        {{\nabla_{\mathbf{z}}^{2}}\mathcal{L}}
        \bigl(\mathbf{z}^{*}\bigr)
    \Bigr)^{-1},
    \label{eq:CLT_solution_SAA_2}
\end{equation}
with
\[
\boldsymbol H(\mathbf{z}^{*})
=
\mathrm{Var}_{\mathbf{X}}\!\left(
    {{\nabla_{\mathbf{z}}}\mathcal{L}}
    \bigl(\mathbf{X},\mathbf{z}^{*}\bigr)
\right).
\]
\end{proposition}
\begin{proof}
  The result follows by the same argument as in the proof Theorem~\ref{theorem:efficiency-RQMC-Fou} in Appendix \ref{appendix:proof_theorem_efficiency}, term~(A), with $S_{\mathrm{shift}}$ replaced by $N$; see \cite[Section~5.2.2]{shapiro_lectures_2014}.  
\end{proof}
In the numerical experiments, we estimate the statistical error at the solution for SAA by replacing $\boldsymbol H$ with the SAA estimator $\widehat{\boldsymbol H}^{\mathrm{SAA}}_{N}$, computed from $\mathrm{Var}\brac{\mathcal{L}_{N}^{(1),\mathrm{SAA}}(\mathbf z^*)}$, and by approximating ${\nabla_{\mathbf{z}}^{2}}\mathcal{L}$ with $\mathcal{L}_{N}^{(2),\mathrm{SAA}}$. These two quantities define $\boldsymbol V_N^{\mathrm{SAA}}(\mathbf{z}^{*})$, yielding
\begin{equation}
\varepsilon_{N}^{\mathrm{SAA}}(\mathbf z^*)
= \frac{C_\alpha}{\sqrt{N}}
\,\sqrt{\norm{\boldsymbol V_N^{\mathrm{SAA}}(\cdot;\mathbf z^*)}}.
\label{eq:stat_err_sol_SAA}
\end{equation}
where $C_\alpha$ is defined in \eqref{eq:RMSE_RQMC}.

\section{Fourier Representation of the MSRM problem}\label{Appendix: Fourier_trans_loss_func}
\subsection{Proof for Corollary \ref{coro:Multivariate_Fourier_pricing}}\label{Appendix:proof_corro_fourier_pricing}
Let 
$\ell^{(\nu)}_{\mathbf K^{(\nu)}}(\mathbf x):=e^{\langle \mathbf K^{(\nu)},\mathbf x\rangle}\ell^{(\nu)}(\mathbf x)$ and $\hat \ell^{(\nu)}_{\mathbf K^{(\nu)}}(\mathbf x):=e^{\langle \mathbf K^{(\nu)},\mathbf x\rangle}\hat \ell^{(\nu)}(\mathbf x)$. Now, using the inverse generalized Fourier transform theorem 
\cite{hormander_analysis_2009}, we express $\ell^{(\nu)}$ as:
 		\begin{equation}\label{eq:inverse generalized Fourier transform}
        \begin{aligned}
            \ell^{(\nu)}(\mathbf{x}-\mathbf{m}) &=  \Re \left[ (2 \pi)^{-d}  e^{\langle \mathbf{-K}^{(\nu)}, \mathbf{x-\mathbf{m}} \rangle}   \int_{\mathbb{R}^d} e^{  \mathrm{i} \langle \mathbf{w}, \mathbf{x-\mathbf{m}}\rangle} {\widehat{\ell}^{(\nu)}}_{\mathbf{K}^{(\nu)}}(\mathbf{w}) \mathrm{d}\mathbf{w}\right], \: \\
            &= \Re \left[ (2 \pi)^{-d} e^{\langle \mathbf{-K}^{(\nu)}, \mathbf{x-\mathbf{m}} \rangle}  \int_{\mathbb{R}^d} e^{  \mathrm{i} \langle \mathbf{w}, \mathbf{x-\mathbf{m}}\rangle} {\widehat{\ell}^{(\nu)}}\paren{\mathbf{w} + \mathrm{i} \mathbf{K}^{(\nu)}} \mathrm{d}\mathbf{w}\right], \: \mathbf{K}^{(\nu)} \in {\delta}_{\ell}^{(\nu)},\: \mathbf{x} \in \mathbb{R}^d
        \end{aligned}	
 		\end{equation}
        Here we have:
        \begin{equation*}
        \begin{aligned}
            {\widehat{\ell}^{(\nu)}}_{\mathbf{K}^{(\nu)}}(\mathbf{w}) &= \int_{\mathbb{R}^d} e^{  -\mathrm{i} \langle \mathbf{w}, \mathbf{x-\mathbf{m^*}}\rangle} e^{\langle \mathbf{K}^{(\nu)}, \mathbf{x-\mathbf{m^*}}  \rangle} \ell^{(\nu)}(\mathbf{x}-\mathbf{m^*}) \mathrm{d}\mathbf{x}\\
            & = \int_{\mathbb{R}^d} e^{  -\mathrm{i} \langle \mathbf{w}+\mathrm{i}\mathbf{K}^{(\nu)}, \mathbf{x-\mathbf{m^*}}\rangle} \ell^{(\nu)}(\mathbf{x}-\mathbf{m^*})\mathrm{d}\mathbf{x} = {\widehat{\ell}^{(\nu)}}(\mathbf{w}+\mathrm{i}\mathbf{K}^{(\nu)})
        \end{aligned}
        \end{equation*}
 		Let $\mathbb{E}[\ell^{(\nu)}(\mathbf{X} - \mathbf{m})] 
 				:= \int_{\mathbb{R}^d} \ell^{(\nu)}(\mathbf{x-\mathbf{m}}) f_{\mathbf{X}}(\mathbf{x}) \mathrm{d} \mathbf{x}$, then using  \eqref{eq:inverse generalized Fourier transform} and Fubini's theorem, we obtain
 		\begin{equation*}
 			\begin{aligned}
 				\hat g^{(\nu),\mathrm{Fou}}(\mathbf{m}) &=  \mathbb{E}[\ell^{(\nu)}(\mathbf{X} - \mathbf{m})] 
 				= \int_{\mathbb{R}^d} \ell^{(\nu)}(\mathbf{x-\mathbf{m}}) f_{\mathbf{X}}(\mathbf{x}) \mathrm{d} \mathbf{x} , \\
 				&=(2 \pi)^{-d} \mathbb{E}_{ f_{\mathbf{X}}}\Re \left[ (2 \pi)^{-d} e^{\langle \mathbf {-K}^{(\nu)}, \mathbf{x-\mathbf{m}} \rangle}  \int_{\mathbb{R}^d} e^{ \mathrm{i} \langle \mathbf{w}, \mathbf{x-\mathbf{m}}\rangle} {\widehat{\ell}^{(\nu)}}\paren{\mathbf{w} + \mathrm{i} \mathbf{K}^{(\nu)}} d\mathbf{w}\right] , \quad \mathbf{K}^{(\nu)} \in {\delta}_{\ell}^{(\nu)}\\
 				&=(2 \pi)^{-d}  \Re\left(\int_{\mathbb{R}^d} e^{\langle \mathbf{K}^{(\nu)}-\mathrm{i}\mathbf{w}, \mathbf{m} \rangle}\mathbb{E}_{ f_{\mathbf{X}}}\left[e^{\mathrm{i}\left\langle\mathbf{w}+\mathrm{i}\mathbf{K}^{(\nu)}, \mathbf{X}\right\rangle}\right] {\widehat{\ell}^{(\nu)}}(\mathbf{w}+\mathrm{i}\mathbf{K}^{(\nu)}) \mathrm{d} \mathbf{w}\right), \quad \mathbf{K}^{(\nu)} \in \delta_{K}^{(\nu)} ={\delta}_{\ell}^{(\nu)} \cap \delta_X \\
 				&= (2 \pi)^{-d}  \Re\left(\int_{\mathbb{R}^d} e^{\langle \mathbf{K}^{(\nu)}-\mathrm{i}\mathbf{w}, \mathbf{m} \rangle}\Phi_{\mathbf{X}}(\mathbf{w+\mathrm{i} \mathbf{K}^{(\nu)}}) {\widehat{\ell}^{(\nu)}}(\mathbf{w+\mathrm{i} \mathbf{K}^{(\nu)}}) \mathrm{d} \mathbf{w}\right).
 			\end{aligned}
 		\end{equation*}
 	The application of Fubini’s theorem is justified by Assumption~\ref{ass:A7}.

\subsection{Fourier transform of the given loss functions}\label{Appendix: Fourier_transform_loss}
 This section presents the Fourier transforms of the component integrands, $\ell^{(\nu)}_{k,p}$ for the loss functions in Example \ref{exam:cross_dependent_losses},  arising in the decomposition in Notation \ref{not:component_integrands}.

\paragraph{Exponential loss function in \eqref{eq:multi_entropy}.} 
For $k\in\{1,d\}$ and $\mathbf{p}=(p_1,\dots,p_k) \in \mathcal I_k$, the loss component can be represented as
\(
\ell_{k,p}(\mathbf x)=c_{k,p}\exp\!\big(\beta\,\mathbf 1^\top\mathbf x\big)
\)
with $\mathbf x\in\R^k$, where the constants $c_{k,p}$ depend on the selected component.

Following the domain decomposition induced by the transformation in
Notation~\ref{not:component_integrands}, we split the integral into
one-sided contributions and introduce componentwise damping parameters
$K^-_{k,p}, K^+_{k,p}\in\mathbb R^k$ such that
\[
K^-_{k,p} < \beta < K^+_{k,p}
.
\]
Under this choice, the damped Fourier transform of $\ell_{k,p}$ is well defined
and given by
\[
\widehat{\ell}^{(0)}_{k,p}(u+iK_{k,p})
= c_{k,p}
\prod_{j=1}^k
\left(
\frac{1}{K^+_{k,p}-\beta-i u_j}
+
\frac{1}{\beta-K^-_{k,p}+i u_j}
\right),
\]
where $K_{k,p}$ denotes the combined contour shift arising from the
positive and negative half-line contributions.
For its gradient and Hessian, letting $\mathbf 1_k := (1,\dots,1)^\top \in \mathbb R^k$,
we obtain
\begin{equation*}
\widehat{\ell}^{(1)}_{k,p}\paren{\mathbf u + \mathrm{i} \mathbf{K}_{k,p}}
=\beta\,\widehat{\ell}^{(0)}_{k,p}\paren{\mathbf u + \mathrm{i} \mathbf{K}_{k,p}}\,\mathbf 1_k,
\qquad
\widehat{\ell}^{(2)}_{k,p}\paren{\mathbf u + \mathrm{i} \mathbf{K}_{k,p}}
=\beta^2\,\widehat{\ell}^{(0)}_{k,p}\paren{\mathbf u + \mathrm{i} \mathbf{K}_{k,p}}\,\mathbf 1_k\mathbf 1_k^\top.
\end{equation*}
\paragraph{QPC loss function in \eqref{eq:multi_qcl}.} The linear component $\ell(x)=x$ does not require a Fourier transform, 
since its expectation and its derivatives can be computed directly under the law of the loss vector $\mathbf{X}$. 

Let $\theta\in\{0,1,2\}$ and define $\phi_\theta(x):=(x^\theta)^+$. 
Fix a damping parameter $K<0$ and set, for $y\in\mathbb R$,
\[
\widehat{\phi}_\theta(y+iK)
:= \int_{\mathbb R} e^{-iyx} e^{Kx} \phi_\theta(x)\,dx
= \int_0^\infty e^{(K-iy)x} x^\theta\,dx
= \frac{\theta!}{(-K+iy)^{\theta+1}}.
\]
Moreover, for $\nu\in\{0,1,2\}$ with $\nu\le\theta$,
\begin{equation}\label{E.2}
\widehat{\phi}^{(\nu)}_\theta(y+iK)
:= \int_{\mathbb R} e^{-iyx} e^{Kx} \phi^{(\nu)}_\theta(x)\,dx
= \frac{\theta!}{(-K+iy)^{\theta-\nu+1}}.
\end{equation}
For $k\in\{1,2\}$ and $\mathbf{p}=(p_1,\dots,p_k)\in\mathcal I_k$, 
the loss components can be written as
\[
\ell_{k,p}(\mathbf x_{k,p})
= c_{k,p}\prod_{j=1}^k \phi_{a(k)}( x_{p_j}),
\qquad
a(1)=2,\quad a(2)=1,
\]
with
\[
c_{k,p}=
\begin{cases}
\frac12, & k=1,\\[0.3em]
\alpha, & k=2.
\end{cases}
\]
Then, for $\nu\in\{0,1,2\}$ and damping vectors $\mathbf K_{k,p}<0$,
the Fourier transform of $\ell^{(\nu)}_{k,p}$ is given by
\[
\widehat{\ell}^{(\nu)}_{k,p}(\mathbf u+i \mathbf K_{k,p})
= c_{k,p}\prod_{j=1}^k
\widehat{\phi}^{(\nu)}_{a(k)}
\!\left(u_j + i(\mathbf K_{k,p})_j\right),
\]
with $\widehat{\phi}^{(\nu)}_\theta$ defined in \eqref{E.2}.

\subsection{Loss vector distributions and extended characteristic functions}\label{appendix:loss_vector_distr}
\begin{example}[Gaussian]\label{ex:Gaussian_components}
Let $\mathbf X\sim\mathcal N_d(\boldsymbol\mu,\boldsymbol\Sigma)$ with
$\boldsymbol\mu\in\mathbb R^d$ and symmetric positive definite
$\boldsymbol\Sigma\in\mathbb R^{d\times d}$. Then the marginal distribution of $\mathbf X_{k,p}$ is
\[
\mathbf X_{k,p}\sim\mathcal N_k(\boldsymbol\mu_{k,p},\boldsymbol\Sigma_{k,p}),
\qquad
\boldsymbol\mu_{k,p}=P_{k,p}\boldsymbol\mu,\quad
\boldsymbol\Sigma_{k,p}=P_{k,p}\boldsymbol\Sigma P_{k,p}^\top.
\]
where $P_{k,p}$ is the coordinate selection matrix introduced in Notation~\ref{not:component_integrands}. The extended CF of $\mathbf X_{k,p}$ is, for any $\mathbf y\in\mathbb C^k$,
\[
\Phi_{\mathbf X_{k,p}}({\mathbf y})
:=\mathbb E\!\left[e^{\mathrm i\langle \mathbf y,\mathbf X_{k,p}\rangle}\right]
=\exp\!\left(\mathrm i\,\mathbf y^\top\boldsymbol\mu_{k,p}
-\tfrac12\,\mathbf y^\top\boldsymbol\Sigma_{k,p}\mathbf y\right),
\]
in particular for $\mathbf y                       =\mathbf u+\mathrm i\,\mathbf K_{k,p}$ with $\mathbf u,\mathbf K_{k,p}\in\mathbb R^k$.
\end{example}
\begin{example}[Normal Inverse Gaussian (NIG)]\label{ex:MNIG_components}
Let $\mathbf X\sim\mathrm{NIG}_d(\alpha,\boldsymbol\beta,\delta,\boldsymbol\mu,\boldsymbol\Gamma)$, i.e.\ $\mathbf X$ is the Generalized Hyperbolic distribution with
$\lambda=-\tfrac12$, $\alpha>0$, $\delta>0$, $\boldsymbol\mu,\boldsymbol\beta\in\mathbb R^d$, and $\boldsymbol\Gamma\in\mathbb R^{d\times d}$ symmetric positive definite \cite{hammerstein_tail_2016},  satisfying  $\alpha^2 > \boldsymbol\beta^\top \boldsymbol\Gamma \boldsymbol\beta$. From Notation \ref{not:component_integrands}, let $\mathbf{p}^c : = \{1,\dots,d\} \setminus \mathbf{p}$ be the complement index list of $\mathbf{p}$, ordered increasingly. Define the selection matrix for the complement $P_{d-k,p^c} \in \mathbb{R}^{{d-k} \times d}$
. Let $\Pi_{k,p}\in\mathbb R^{d\times d}$ be any permutation matrix that reorders coordinates so that
$\Pi_{k,p}:= \begin{bmatrix}
    P_{k,p}\\
    P_{d-k,p^c}
\end{bmatrix}$, and write the permuted parameters as,
\[
\Pi_{k,p}\boldsymbol\mu=\binom{\boldsymbol\mu_1}{\boldsymbol\mu_2},\qquad
\Pi_{k,p}\boldsymbol\beta=\binom{\boldsymbol\beta_1}{\boldsymbol\beta_2},\qquad
\Pi_{k,p}\boldsymbol\Gamma\Pi_{k,p}^\top=
\begin{pmatrix}
\boldsymbol\Gamma_{11}&\boldsymbol\Gamma_{12}\\
\boldsymbol\Gamma_{21}&\boldsymbol\Gamma_{22}
\end{pmatrix},
\quad \boldsymbol\Gamma_{11}\in\mathbb R^{k\times k}.
\]
Then, by \cite[Theorem 1(a)]{hammerstein_tail_2016} (applied with $\lambda=-\tfrac12$), the selected marginal $\mathbf X_{k,p}$ satisfies:
\[
\mathbf X_{k,p}\sim \mathrm{NIG}_k\!\big(\alpha_{k,p},\boldsymbol\beta_{k,p},\delta_{k,p},\boldsymbol\mu_{k,p},\boldsymbol\Gamma_{k,p}\big),
\]
with
\[
\boldsymbol\mu_{k,p}=\boldsymbol\mu_1,\qquad
\boldsymbol\beta_{k,p}=\boldsymbol\beta_1+\boldsymbol\Gamma_{11}^{-1}\boldsymbol\Gamma_{12}\boldsymbol\beta_2,\qquad
\delta_{k,p}=\det(\boldsymbol\Gamma_{11})^{1/2}\,\delta,
\qquad
\boldsymbol\Gamma_{k,p}=\det(\boldsymbol\Gamma_{11})^{-1/k}\,\boldsymbol\Gamma_{11},
\]
and
\[
\alpha_{k,p}=\det(\boldsymbol\Gamma_{11})^{-1/(2k)}
\sqrt{\alpha^2-\big\langle \boldsymbol\beta_2,\,
(\boldsymbol\Gamma_{22}-\boldsymbol\Gamma_{21}\boldsymbol\Gamma_{11}^{-1}\boldsymbol\Gamma_{12})\,\boldsymbol\beta_2\big\rangle }.
\]
The extended CF is, for any $\mathbf y\in\mathbb C^k$,
\[
\phi^{\mathrm{NIG}}_{\mathbf X_{k,p}}(\mathbf y)
=\exp\!\left(
\mathrm i\,\mathbf y^\top\boldsymbol\mu_{k,p}
+\delta_{k,p}\Big(\gamma_{k,p}-\sqrt{\alpha_{k,p}^2-(\boldsymbol\beta_{k,p}+\mathrm i\,\mathbf y)^\top\boldsymbol\Gamma_{k,p}(\boldsymbol\beta_{k,p}+\mathrm i\,\mathbf y)}\Big)
\right),
\]
where $\gamma_{k,p}=\sqrt{\alpha_{k,p}^2-\boldsymbol\beta_{k,p}^\top\boldsymbol\Gamma_{k,p}\,\boldsymbol\beta_{k,p}} >0$. In particular, this applies to $\mathbf y=\mathbf u+\mathrm i\,\mathbf K_{k,p}$ with $\mathbf u,\mathbf K_{k,p}\in\mathbb R^k$.
\end{example}

\end{document}